\documentclass[USenglish, 11pt]{article}
\usepackage[T1]{fontenc}
\usepackage{amsthm,amsmath,amssymb,amsfonts,authblk,hyperref,enumitem}
\usepackage{graphicx}
\usepackage{color}
\usepackage{authblk}
\usepackage{cancel}
\usepackage{marginnote}
\usepackage{makecell}

\usepackage{caption}
\usepackage[linesnumbered,algoruled,boxed,lined,algo2e]{algorithm2e}
\usepackage{algorithm}

\RestyleAlgo{ruled}
\SetKwInOut{Parameter}{Parameters}
\SetKwComment{Comment}{/* }{ */}

\usepackage{url}

\usepackage{float}
\usepackage{tikz}

\newtheorem{theorem}{Theorem}[section]

\newtheorem{lemma}[theorem]{Lemma}
\newtheorem{claim}[theorem]{Claim}
\newtheorem{proposition}[theorem]{Proposition}

\newtheorem{corollary}[theorem]{Corollary}

\newtheorem{definition}{Definition}

\usepackage{fullpage}

\newcommand{\OO}{{\widetilde{O}}}
\newcommand{\N}{{\mathbb{N}}}
\newcommand{\MI}{{\mathrm{MIS}}}

\newcommand{\AND}{{\;\wedge\;}} 
\newcommand{\OR}{\bigvee} 

\newcommand{\poly}{{\mathrm{poly}}}

\newcommand{\comment}[1]{}

\newcommand{\Z}{\mathbb{Z}}

\newcommand{\cost}{\textrm{cost}}
\newcommand{\supp}{\textbf{supp}}

\newcommand{\LRGE}{\textsc{large}}

\newcommand{\child}{\textbf{child}}

\newcommand{\prob}{\textbf{Pr}}
\newcommand{\expected}{\mathbb{E}}

\newcommand{\ED}{\textbf{ED}}
\newcommand{\HAM}{\textbf{Ham}}

\newcommand{\emptystr}{\varepsilon}

\newcommand{\concat}{\circ}
\newcommand{\bigconcat}{\odot}

\newcommand{\eddelete}{\textbf{del}}
\newcommand{\edsubst}{\textbf{sub}}
\newcommand{\edinsert}{\textbf{ins}}

\newcommand{\recover}{\textbf{restore}}

\newcommand{\SK}{\textsc{sk}}
\newcommand{\sk}{\textsc{sk}}

\newcommand{\fingerprint}{\textsc{fingerprint}}

\newcommand{\indices}{I_{\neq}}

\newcommand{\tr}{\textbf{tr}}

\newcommand{\leafroot}{\textbf{leaf}$\rightarrow$\textbf{root}}

\newcommand{\costly}[1]{\textbf{costly}(#1)}
\newcommand{\gsize}{s}

\newcommand{\gramsize}{k}

\newcommand{\gpoint}[2]{\langle #1,#2 \rangle}

\newcommand{\suprange}{S}

\newcommand{\RECOVER}{\textsc{recover}}

\newcommand{\critset}{canonical difference summary}
\newcommand{\treecap}{\kappa}
\newcommand{\stringcap}{\kappa}
\newcommand{\locationcap}{\lambda}
\newcommand{\locprob}{\textbf{loc}}

\newcommand{\canon}{{\textbf{canon}}}

\newcommand{\occupied}{\textbf{occ}}

\newcommand{\PMR}{\textbf{TMR}}
\newcommand{\HMR}{\textbf{HMR}}

\newcommand{\mdecomp}{\textsc{main-decomp}}
\newcommand{\msketch}{\textsc{main-sketch}}

\newcommand{\msketchout}{\textbf{main-sk}}
\newcommand{\fullsketchout}{\textbf{sk}}
\newcommand{\fullsketch}{\textsc{ED-sketch}}
\newcommand{\edrecover}{\textsc{ED-recover}}
\newcommand{\mreconstruct}{\textsc{main-reconstruct}}

\newcommand{\bdecomp}{\textsc{basic-decomp}}
\newcommand{\vectorcondense}{\textsc{grammar-condense}}
\newcommand{\locationcondense}{\textsc{location-condense}}

\newcommand{\findstrings}{\textsc{find-strings}}
\newcommand{\findlocations}{\textsc{find-locations}}

\newcommand{\undefnd}{\textsc{undefined}}
\newcommand{\gthreshsk}{\textsc{ted-fingerprint}}
\newcommand{\threshsk}{\textsc{OR-fingerprint}}
\newcommand{\threshskout}{\textbf{OR-print}}
\newcommand{\hmrsketch}{\textsc{hmr-sketch}}
\newcommand{\hmrrecover}{\textsc{hmr-recover}}

\newcommand{\tval}{\text{value}}
\newcommand{\tprod}{\text{product}}
\newcommand{\tsq}{\text{square}}
\newcommand{\thash}{\text{hash}}
\newcommand{\dindex}{\text{index}}
\newcommand{\xval}{\text{x-val}}
\newcommand{\yval}{\text{y-val}}
\newcommand{\OCC}{\text{OCC}}
\newcommand{\MERGE}{\text{MERGE}}
\newcommand{\grid}{\textbf{Grid}}
\newcommand{\edges}{\textbf{E}}
\newcommand{\paramspace}{\Gamma}

\newcommand{\decode}{\textsc{basic-decode}}
\newcommand{\ztree}{\textbf{substrs}}

\newcommand{\gapthresh}{t}
\newcommand{\bvec}{\textbf{grams}}
\newcommand{\leftchild}{\textbf{left-size}}
\newcommand{\lvec}{\leftchild}

\newcommand{\LOS}{L}
\newcommand{\BASE}{W}
\newcommand{\zint}[1]{\textbf{Loc}(#1)}
\newcommand{\interval}{\textbf{I}}
\newcommand{\startindex}{\textbf{start}}
\newcommand{\bvsize}{{N}}
\newcommand{\bvsizeval}{n^{60}}

\newcommand{\loadpar}{R}

\newcommand{\dfactor}{{s_{{\mathrm{E}\rightarrow\mathrm{H}}}\,}}

\newcommand{\bkfailure}{{s_{{\mathrm{split}}}\,}}
\newcommand{\tdepth}{d}
\newcommand{\orgap}{s_{{\mathrm{OR}}}}
\newcommand{\orgapvalue}{2^{O(\sqrt{\log(n)\log\log(n)})}}

\newcommand{\bitlength}{\text{bit-length}}

\newcommand{\bigfield}{\mathbb{F}}

\newcommand{\fsketch}{\textbf{str-sk}}
\newcommand{\gsketch}{\textbf{loc-sk}}
\newcommand{\outputset}{\textbf{found-edges}}
\newcommand{\tstree}{\textbf{OR-prints}}
\newcommand{\fptree}{{\textbf{fingerprints}}^*}

\newcommand{\COMMENT}{\textsc{Comment}}
\newcommand{\stringmismatch}{\textbf{gram-mismatch}}
\newcommand{\locmismatch}{\textbf{size-mismatch}}
\newcommand{\reportednodes}{\textbf{top-nodes}}
\newcommand{\claimededges}{\textbf{candids}}
\newcommand{\claimedsize}{\textbf{found-size}}
\newcommand{\claimedlocation}{\textbf{found-start}}
\newcommand{\claimedstring}{\textbf{found-str}}
\newcommand{\claimedbv}{\textbf{found-gram}}

\newcommand{\claimednodes}{\textbf{loc-nodes}}

\newcommand{\stringnodes}{\textbf{str-nodes}}

\newcommand{\hmrload}[1]{\widehat{#1}}

\newcommand{\lexorder}{<_{\text{lex}}}

\newcommand{\nonemptynodes}{\textbf{vis-nodes}}

\newcommand{\DONE}{$\textbf{D}_1$}
\newcommand{\DTWO}{$\textbf{D}_2$}

\newcommand{\numrep}{10\gsize+50}

\newcommand{\normex}{\mathrm{\textbf{Normal}}\,}
\newcommand{\abnormex}{\mathrm{\textbf{Abnormal}}\,}

\newcommand{\loadestimator}{\textbf{load-ub}}

\newcommand{\CS}{CS}
\newcommand{\NCS}{NCS}
\newcommand{\critpar}{{s_{\mathrm{load}}\,}}
\newcommand{\LOC}{\textbf{LOC}}
\newcommand{\STR}{\textbf{STR}}

\newcommand{\eventT}{\textbf{T}}
\newcommand{\eventC}{\textbf{C}}

\newcommand{\eventF}{\textbf{F}}
\newcommand{\eventL}{\textbf{L}}

\begin{document}



\date{\today}

\title{Almost Linear Size Edit Distance Sketch}

\author[1,2]{Michal Kouck{\'{y}}\thanks{Email: koucky@iuuk.mff.cuni.cz. 
Part of the work was carried out during an extended visit to DIMACS, with support from the National Science Foundation under grant          
number CCF-1836666 and from The Thomas C. and Marie M. Murray Distinguished Visiting Professorship in the Field of Computer Science at Rutgers University.
Partially supported by the Grant Agency of the Czech Republic under the grant agreement no. 19-27871X and 24-10306S. 
This project has received funding from the European Union’s Horizon 2020 research and innovation programme under the Marie Skłodowska-Curie grant agreement No. 823748 (H2020-MSCA-RISE project CoSP). 
}}
\author[3]{Michael Saks\thanks{Email: msaks30@gmail.com.}}
\affil[1]{Computer Science Institute of Charles University,
Prague, Czech Republic}
\affil[2]{DIMACS, Rutgers University, Piscataway, NJ, USA}
\affil[3]{Department of Mathematics, Rutgers University, Piscataway, NJ, USA}

\maketitle

\begin{abstract}
	Edit distance is an important measure of string similarity.
    It counts the number of insertions, deletions and substitutions one has to make to a string $x$ to get a string $y$.
    In this paper we design an almost linear-size sketching scheme for computing edit distance up to a given threshold $k$.  
    The scheme consists of two algorithms, a sketching algorithm and a recovery algorithm. The sketching algorithm depends on the parameter $k$ and takes
    as input a string $x$ and a public random string $\rho$ and computes a  sketch $sk_{\rho}(x;k)$, which is a digested version of $x$.  
    The recovery algorithm is given two sketches $sk_{\rho}(x;k)$ and $sk_{\rho}(y;k)$ as well as the public
    random string $\rho$ used to create the two sketches, and (with high
    probability) if the edit distance $\ED(x,y)$ between $x$ and $y$ is at most $k$, will output $\ED(x,y)$ together
    with an optimal sequence of edit operations that transforms $x$ to $y$, and 
    if $\ED(x,y) > k$ will output \LRGE.   The size of the sketch output by the
    sketching algorithm
    on input $x$ is $k{2^{O(\sqrt{\log(n)\log\log(n)})}}$ (where $n$ is an upper bound on length of $x$).
    The sketching and recovery algorithms both run in time polynomial in $n$.
    The dependence of sketch size on $k$ is information theoretically optimal and
    improves over the quadratic dependence on $k$ in schemes of Kociumaka, Porat and Starikovskaya (FOCS'2021), and Bhattacharya and Kouck\'y (STOC'2023).
 \end{abstract}

\section{Introduction}

The \emph{edit distance} of two strings $x$ and $y$ measures how many edit operations (removing a symbol, inserting a symbol
or substituting a symbol by another) are needed to transform $x$ to  $y$.
Computing edit distance is a classical algorithmic problem.  For input strings of length at most $n$,
edit distance can be computed in time $O(n^2)$ using dynamic programming~\cite{WF74,MP80,G16}.
Assuming the Strong Exponential Time Hypothesis (SETH), this cannot be improved to truly sub-quadratic time $O(n^{2-\epsilon})$~\cite{BI15}.  
When parameterized by the edit distance $k=\ED(x,y)$, the running time has been improved to $O(n+k^2)$~\cite{LMS98}.
The edit distance of two strings can be approximated within a constant factor in time $O(n^{1+\epsilon})$ \cite{AN20,KS20,BR20,CDGKS20}.

This paper concerns \emph{sketching schemes} for edit distance, which consist of a
\emph{sketching algorithm}, parameterized by an integer $k$, that takes a string $x$ and (using a public random string $\rho$) maps it to a short {\em sketch} $\sk_{\rho}(x;k)$, and a \emph{recovery algorithm} that takes as input
two sketches $\sk_{\rho}(x;k)$ and $\sk_{\rho}(y;k)$ and the public random string $\rho$ and, with high
probability (with respect to $\rho$), outputs $\ED(x,y)$ when $\ED(x,y) \leq k$ and outputs $\LRGE$ otherwise.

The goal is to get polynomial time sketch and recovery algorithms that achieve the smallest
possible sketch length.
Jin, Nelson and Wu \cite{nelson_edit_sketch} proved that sketches must have length $\Omega(k)$.  
For edit distance it is not apriori clear whether sketches of size $k^{O(1)}n^{o(1)}$ exist, even non-uniformly.
The first sketching scheme with $poly(k)$ sketch size was found by Belazzougui and Zhang~\cite{belazzougui_zhang} who attained sketch  size $\OO(k^{8})$. 
(Here, and throughout the paper, $\OO(t)$ means $t\log^{O(1)}n$, where $n$ is an upper bound on the length of the strings.) 
This was improved to $\OO(k^{3})$  by Jin, Nelson and Wu~\cite{nelson_edit_sketch}, and then to $\OO(k^{2})$ by Kociumaka, Porat and Starikovskaya~\cite{editsketchfocs2021}.
The above sketches used {\em CGK random walks} on strings~\cite{CGK16} to embed the edit distance metric into the Hamming distance metric with distortion $O(k)$, plus additional techniques.
A different approach, based on a string decomposition technique, was used by Bhattacharya and Kouck\'y~\cite{BK23}. 
Here we will also use this technique.

The quadratic dependence on $k$ in the (very different) sketches of~\cite{editsketchfocs2021}  and 
\cite{BK23} and also in the exact computation algorithm of~\cite{LMS98} is suggestive that there may be something intrinsic to the problem that requires quadratic dependence on $k$ for both sketching
and evaluation.
In this paper, we show that this is not the case, by presenting an efficiently computable sketch of size $O(k\orgapvalue)$ which is $k$ times a "slowly" growing function of $n$, that is intermediate between $\log^{\omega(1)}n$ and $n^{o(1)}$.   Our main result is:
\begin{theorem}[Sketch for edit distance]\label{t-main}
   There is a randomized sketching algorithm $\fullsketch$ that on an input string $x$ of length at most $n$ with parameter $k<n$ and using a public random string $\rho$ 
   produces a sketch $\sk_{\rho}(x)$ of size $O(k \orgapvalue)$, and recovery algorithm  $\edrecover$ such that given two sketches $\sk_{\rho}(x)$ and $\sk_{\rho}(y)$ for two strings $x$ and $y$ and $\rho$, with probability at least $1-1/n$ (with respect to $\rho$), outputs $\ED(x,y)$ if $\ED(x,y) \leq k$ and $\LRGE$ otherwise.
   The running time of the $\fullsketch$ is $n^{O(1)}$ and the running time of $\edrecover$ is $\OO(\min(n^2,k^3 \orgapvalue))$.
\end{theorem}
 
We remark that we did not attempt to optimize the running time, or poly-log factors in the sketch sizes.
The running time of $\fullsketch$ is largely determined by the running time of the Ostrovsky-Rabani embedding~\cite{OR07} (see below) which runs in polynomial time, but we don't know the exponent. 
The amount of randomness the algorithm uses can be reduced to poly-logarithmic in $n$.

Our sketch has the additional property that the recovery algorithm also determines an optimal sequence of edit operations that transforms $x$ to $y$.

Sketching for Hamming distance has been studied extensively and is well understood.
Several approaches yield sketches of size $\OO(k)$ that can recover the Hamming distance and can solve the harder problem of \emph{mismatch recovery}, i.e.,
reconstructing the set of positions where the two strings differ, see e.g.~\cite{rollinghashSODA2019,PL07}; this sketch size
is information theoretically optimal.
Our construction for edit distance uses sketches for Hamming distance, but
we need a more refined version of Hamming distance sketches that allow for recovery of all differences
in regions of the strings where the density of differences is low, even if the overall Hamming distance
is large.  
We call this new sketch a {\em hierarchical mismatch recovery scheme}.
This is the main new technical tool.

Not much is known about sketching edit distance when we only want to approximate edit distance from the sketches.
For Hamming distance, there are known sketches of poly-logarithmic size in $n$ that allow
recovery of Hamming distance within a $(1+\epsilon)$-factor \cite{DBLP:journals/siamcomp/FeigenbaumKSV02}.
For edit distance nothing like that is known.
A more stringent notion of sketching is that of embedding edit distance metrics into $\ell_1$ metrics.
The best known result in this direction, by Ostrovsky and Rabani~\cite{OR07}, 
gives an embedding of edit distance into $\ell_1$ with {\em distortion} (approximation factor) $\orgap(n)=\orgapvalue$.
Interestingly, our sketch relies on this embedding to {\em fingerprint} strings by their approximate 
edit distance.
The dependence of our sketch size on $n$ in Theorem~\ref{t-main} can be stated more precisely
as $\OO(\orgap(n)^2)$.  Since our use of their embedding is ``black box'',
any improvement on the distortion factor for embedding edit distance into $\ell_1$ 
would give a corresponding improvement in our sketch size.

\section{Our technique}

The starting point of our sketch is the string decomposition algorithm of Bhattacharya and Kouck\'y~\cite{BK23},
The algorithm $\bdecomp$ (see Section~\ref{subsec:bdecomp}) takes a string $x$ and partitions it into fragments 
so that each fragment can be described concisely by a small context-free grammar.
The size of the grammar is at most $k'$ for some chosen parameter $k'$.
(In~\cite{BK23} this $k'$ is chosen to be $\OO(k)$.)
The partitioning uses randomness to select the starting point for each fragment
and has the property that for any given position in the string, the
probability that it will start a fragment is at most $p\approx 1/k'$.

The  partitioning process is \emph{locally consistent}, i.e.,
when applied to two strings $x$ and $y$ with probability at least $1-\OO(\ED(x,y)/k')$ the decompositions of $x$ and $y$ are \emph{compatible}. 
Being compatible means that they have the same number of fragments and $\ED(x,y)$ is equal to the sum
of edit distance of corresponding fragments.  
In particular, if $x$ and $y$ are compatibly split then to recover $\ED(x,y)$
it suffices to reconstruct each corresponding pair of fragments of $x$ and $y$ that are different and sum up their edit distances.

In \cite{BK23} this decomposition procedure was used
to obtain edit distance sketches of size $\OO(k^2)$ by
reducing the problem of
sketching for edit distance to the easier problem
of sketching for  (Hamming) mismatch recovery
mentioned in the introduction. 
This is the problem of sketching two equal length sequences so that from their sketches one can recover all the
locations where the two sequences differ. 

The reduction to mismatch recovery sketches proceeds as follows.
The string $x$ is decomposed into fragments.
The grammar of each fragment is encoded using some fixed-size encoding of length $\OO(k')$, 
the encodings are concatenated and the resulting sequence is sketched using a mismatch recovery scheme. 
Given the Hamming distance sketch for the sequence of fragments of $x$, and the corresponding
sketch for $y$, 
we can (with an additional technical trick) recover all corresponding pairs of grammars that are different.  
Each pair of corresponding grammars represents a pair of corresponding fragments of $x,y$, and $\ED(x,y)$
is equal to the sum of the edit distances of all recovered pairs of fragments.
There are at most $k$ corresponding pairs of grammars that differ, and each of them has $k'=\OO(k)$ bits.
The size of the Hamming sketch needed to perform this recovery is $\OO(k^2)$.
This is the sketch from~\cite{BK23}.

The sketch length is $\OO(k^2)$ because the sketch must handle two different extremes.
If all edit operations appear in large clusters of
size $k/C$, for $C=\OO(1)$,  and each cluster is contained in a single fragment pair, then
there are at most $C$ fragment pairs that are unequal and these could be
handled by sketches of size $\OO(k)$.
On the other hand, if the edit operations are well separated, then we could choose a value of $k'$ that
is $\OO(1)$ resulting in
 a partition into much smaller fragments each of which has grammar size $\OO(1)$. In this
 case the edit operations may appear in $\Omega(k)$ different fragment pairs, but because
 the grammar size of the fragments is $\OO(1)$ we can again manage with sketch size $\OO(k)$.
Of course, the distribution of edit operations will rarely match either of these extremes,
there may be clusters of edit operations of varying size and density.

{\em Decomposition tree.}
A natural approach to handling this is to build decompositions for many different
values of $k'$.  We start with a decomposition obtained from parameter $k'=k_0$
which is larger, but not much larger than $k$.  
We then apply the decomposition  again to each fragment of the first decomposition, using
the smaller parameter $k_1=k_0/2$.  
We iterate this recursively where the value $k_i$ of $k'$ at recursion level $i$
is $k_0/2^i$ stopping when  $k_i$ is $\OO(1)$.
This would decompose the string into smaller and smaller fragments ending with fragments described by constant size grammars.
We call this a {\em decomposition tree} of the string.  We can then do separate sketches for
each of the levels of the tree.  The sketches at the top of the tree (small $i$) are used
to find  edit operations that occur together in large clusters, and there can't be too many of these clusters.
Sketches at the bottom of the tree (large $i$) are used to find edit operations in regions
where the edit operations are well spread.  These edit operations may
be spread over $\Omega(k)$ pairs
of fragments, but each such pair has  edit distance $\OO(1)$. The intermediate
levels would handle cases where the density of edit operations is intermediate between these extremes.

To be more explicit, the sketch associated to decomposition level $i$ is responsible for recovering pairs of compatible fragments at level $i$ whose edit distance is
(roughly) comparable to or larger than $k_i$,
that could not be recovered from the sketches for previous levels of the decomposition. 
There are at most $\OO(k/k_i)$ such pairs of fragments (counting only compatibly split pairs).
Since the fragments are represented by grammars of size at most $k_i$, the level
$i$ scheme will need to find at most $\OO(k)$ mismatches for compatibly split fragments.

For compatible pairs of fragments at level $i$ whose edit distance is small compared to $k_i$,
the decomposition procedure will split them compatibly (with {\em fairly high probability}\footnote{Here {\em fairly high probability} 
means probability at least $1-1/\poly\log n$ for a sufficiently large poly-logarithmic function.}) and the
edit distance for these will be recovered from the sketches
corresponding to deeper levels of the decomposition.  In this way, all of the edit
operations will be found by some level of the decomposition.  

While the sketch for level $i$ only needs to identify at most $\OO(k)$ mismatches, it can
not use an ordinary  mismatch recovery sketch for $\OO(k)$ mismatches, because
the sequences of the grammar encodings
for the level $i$ fragments may differ in many more than $\OO(k)$ positions, since
the $\OO(k)$ positions account only for the differences due to pairs of fragments that must be recovered
at level $i$ but not those due to differences coming from other fragment pairs as explained further.
To deal with this, 
we will need the extension of mismatch recovery
mentioned earlier that handles hierarchical mismatch recovery (which will be discussed
in more detail later in this section.)

{\em Grammar representation.}
From the sketch at a level $i$ which corresponds to grammars of size at most $k_i$
we hope to recover those fragments that have roughly $k_i$ edit operations that were not recovered
from the sketches at earlier levels.
Here it becomes important how we represent the grammars to the Hamming sketch.
Each grammar is a subset of {\em rules} from a domain of polynomial size.
The grammars that we produce have the following useful property: 
If we represent two fragments of $x$ and $y$ by grammars and the fragments have edit distance at most $h$
then the grammars will differ in at most $\OO(h)$ rules.
So if we represent the grammars by their characteristic vectors
the vectors will have Hamming distance at most $\OO(h)$.
This is useful since the fragments that contain much less than $k_i$ edit operations  
and that should be recovered from a Hamming sketch for some level $>i$
will contribute to the Hamming sketch of level $i$ by at most $\OO(k)$ differences in their grammars in total.
For fragments that have edit distance $\ge k_i$ (that were not recovered from earlier levels) we want to recover their complete grammars.

\emph{Watermarking using edit distance fingerprints.}
For the level $i$ encoding,
we want that if two corresponding fragments are at edit distance at least $t_i \approx k_i$, 
then recovering the mismatches in their
grammar encodings is enough to fully recover both grammars and therefore both fragments.  To ensure
this, we will \emph{watermark} the grammar encodings by a special fingerprint (computed from the entire fragment),
and replace each $1$ in the grammar encoding by the fingerprint.
Assuming
that the watermark of the two fragments are different, then this will allow one
to recover both grammars from the set of mismatches.
The watermark we use is a {\em threshold edit distance fingerprint}.
This is a randomized fingerprinting scheme depending on parameter $t$,
that maps each string to an integer so 
that if two strings have edit distance more than $t$, their fingerprints differ with high probability,
and if two strings have  edit distance less than  $t/P$, for some parameter $P>1$ 
then their fingerprints will be the same with fairly high probability (with no
promise if the edit distance lies in the interval $[t/P,t]$.) The parameter
$P$ is a measure of the quality of the fingerprinting algorithm; with smaller $P$ being
higher quality.
Such a fingerprinting scheme can be obtained from an embedding of the edit distance metric into the $\ell_1$ metric.
We use the embedding of Ostrovsky and Rabani~\cite{OR07} which has
distortion $\orgap=2^{\sqrt{\log n \log \log n}}$
and this distortion translates into the parameter $P$ of the fingerprinting scheme.

{\em Canonical edit operations.}
For the sketching procedure as outlined, the recovery algorithm will recover corresponding
pairs of unequal but compatible fragments and compute their edit distance.  While the recovered fragment pairs
are likely to encompass most edit differences between the two strings, it will typically
miss some pairs, and may also incorrectly recover a small fraction of the fragment pairs.  This
means that the set of edit operations recovered will only be approximate.  To address this we will need
to do multiple independent sketches and when doing recovery, we combine
the outcomes recovered from each independent sketch (as described below).  As mentioned earlier
if $\ED(x,y)=k$, there could be many different sets of
$k$ edit operations that transform $x$ to $y$.
To output edit operations consistently among independent runs of the sketching algorithm
we will always opt for a {\em canonical} choice of the edit operations.
The canonical choice prefers insertions into $x$ over substitutions which in turn are preferred over deletions from $x$.
This preference is applied on edit operations from left to right in $x$ so
the choice of edit operations corresponds to the left-most shortest path in the usual edit distance graph of $x$ and $y$.
Importantly, the canonical path is consistent under taking substrings of $x$ and $y$.
See Section~\ref{subsec:grids} for details on the choice of the path.

{\em Bad splits.}
For two strings $x$ and $y$,
as we split the fragments of each through successive levels we will inevitably split some fragments near an edit operation;
this is called a {\em bad split}.  A bad split might cause $x$ to be divided at a particular location but not $y$.
This causes two problems: (1)  sub-fragments of the badly split fragment may not
align, and (2) the bad split may cause fragments to the right of the badly split fragment to be
misaligned.  Both problems need to be dealt with. 

{\em Regularization of trees.}
To handle the second issue we regularize the decomposition tree.
The decomposition tree is of depth at most $O(\log n)$ and it might have degree up to $n$.
We make it regular as follows: for any node with fewer than $n$ children we append dummy children to the
right of the real children; these are thought of as representing empty strings. We do this
at all levels so that the tree has degree $n$ in all the internal nodes and 
all the leaves are at depth exactly $\tdepth = O(\log k + \sqrt{\log n \log \log n})$. 
Each node of the tree is labeled by a fragment of $x$ and its corresponding grammar.
This ensures that the underlying tree for the decomposition is the same for every input string.
Now, in constructing the level $i+1$ decomposition from the level $i$, the fragments
that correspond to the children of node $v$ in the two trees are paired with each other.
A bad split arising from the decomposition of the two fragments corresponding
to tree node $v$ may  cause misalignment among the descendants of $v$ but will not affect the
alignment of nodes to $v$'s right.

{\em Sparse representation.}
The regularized tree is of super-polynomial size $\le n^{\log n}$, but only has
$\OO(n)$ nontrivial nodes that represent nonempty strings.  As a result, the tree can
be constructed in polynomial time, and represented concisely by the set of ordered pairs
of (node, fragment) for nonempty fragments. Furthermore
the sketching algorithm works well with the sparse representation and its running
time is at most polynomial in the number of nontrivial nodes.

{\em Mismatch floods.}
The more significant issue caused by bad splits is the first one, that a bad split at a node
may result in bad splits at many of its descendants. Indeed, if the partition of the fragments
$x_v$ and $y_v$ 
of $x$ and $y$ at node $v$ are not compatibly split, then the fragment pairs
of its children will not be aligned and the edit distances
between fragments corresponding to children of $v$ may be arbitrarily large.
This misalignment
will propagate down the tree possibly resulting in huge total edit distance between 
fragments for many nodes in the subtree rooted at $v$, and  the grammar encodings
of these fragments may have a very large Hamming distance.


A node in the tree that is in the subtree of a badly split node is referred to as \emph{flooded} if its corresponding fragments need many edit operations.
At level $i$ of the tree,
we only need to recover the grammars corresponding to the unflooded
nodes of the level (because we expect that the edit operations for each flooded node
will be recovered by the sketches corresponding to an unflooded ancestor in the tree.)
The usual mismatch recovery sketch does not allow for such selective recovery of unflooded portions, because
the flooded portion may cause the total number of mismatches to far exceed the  capacity
of the mismatch recovery scheme.
Thus we design a new variant of Hamming schemes, which we call {\em hierarchical mismatch recovery scheme}, that recovers differences in the unflooded parts.

{\em  Hierarchical mismatch recovery.}
The {\em hierarchical mismatch recovery scheme} is applied to
a vector that is indexed by the leaves of a tree.  
The specification of the problem includes an assignment
of a positive capacity $\treecap_v$ to each node $v$ in the tree, where 
all nodes at level $i$ have the same capacity $\treecap_i$.
For any two strings $x$ and $y$, the capacity function induces a {\em load function}
$\hmrload{\treecap}$, where the value $\hmrload{\treecap}_v$
is 1 on leaves where $x$ and $y$ differ and 0 on other leaves, and the load on a node $v$ is the minimum of the sum of the loads
of its children and its own capacity.  A node $v$ is {\em underloaded} if  its load
is less than $\treecap_v/R$ for some parameter $R$, and a leaf is \emph{accessible}
if every node on its path to the root is underloaded.  The scheme is required to recover
the mismatch information for all accessible leaves.  

The sketch is implemented 
as follows: Let $d$ be the depth of the tree.
For each node at level $d-1$ we apply a (standard) Hamming mismatch recovery scheme to the vector
of its children that handles (slightly more than) $\treecap_{d-1}$ mismatches.  Then working
the way up the tree, for each node  we apply a mismatch recovery to the vector consisting
of the sketches computed at each of its children. This is passed upwards and the final
sketch at the root is the output.  The specific mismatch recovery scheme used
is a {\em superposition scheme} which is described in Section~\ref{subsec:superposition}. 
The hierarchical sketches for two strings $x$ and $y$ will allow for recovery
of mismatches occurring at all accessible leaves. 

The intuition behind this is as follows.
For two strings $x$ and $y$, label each node by the pair of intermediate sketches
for $x$ and $y$ at that node.  This pair of sketches
encodes information about the mismatches of $x$ and $y$ at the leaves
of its subtree. This information at node $v$ is a {\em compression} to $\OO(\treecap_i)$ bits
of the information from its children. Inductively one can show that for an underloaded node the compressed string it computes
has enough bits to preserve the  information about accessible mismatches that is encoded
in the sketches of its children.  For overloaded nodes, the compression will destroy the information about
its children, but this is not a problem because the scheme is not required to recover
mismatch information for leaves below an overloaded node.  An important thing to note
is that the impact of overloaded nodes among $v$'s children (which may have
a large number of mismatches among its leaves) is controlled by the
fact that the sketch size at those nodes is restricted by its capacity.

In our application to sketching the grammars we assign the capacity to each level of the tree 
so that the capacity is proportional to the grammar size $k_i$ we expect at that level in the decomposition tree.
Nodes that correspond to misaligned fragments (due to a bad split at an ancestor) and have a huge edit distance
will correspond to overloaded nodes, and as discussed above, the hierarchical scheme contains the
damage caused by the error flood to the subtree below the occurrence of the bad split.
The idea behind the choice of the parameters is that the probability of a bad split is proportional to the number of edit operations in the fragment;
it is roughly $\OO(\textit{\# of edit op's}/k_i)$.
Hence, in expectation each edit operation is responsible for $\OO(1)$ mismatches resulting from bad splits that contribute to possible flooding of a node.
We can adjust the parameters so that the flooding of a node is in expectation only a tiny fraction of its capacity.
We can then apply Markov's inequality to argue that with a good probability nodes along a chosen path are not overloaded, i.e., flooded.
Details for our hierarchical mismatch recovery scheme are given in Section~\ref{subsec:HMR}.

{\em Parameters.}
For our edit distance sketch we will set the parameters as follows: $\treecap_0 = \OO(\orgap^2 k)$, $\treecap_i = \treecap_0 / 2^i$,
$k_0=\OO(\orgap k)$, $k_i = k_0/2^i$ and $t_0=\OO(\orgap k)$, $t_i = t_0 /2^i$.
Our sketch will be obtained by applying the hierarchical mismatch recovery scheme to each level of the decomposition tree with those parameters.

{\em Infrequent bad splits.}  The decomposition procedure is designed so that for a node $v$
at level $i$ with  fragments $x_v$ and $y_v$, if $\ED(x_v,y_v) \leq k_i/C$ where
$C=\text{polylog}(n)$ the decompositions of $x_v$ and $y_v$ will be compatible with probability
$1-1/\text{polylog}(n)$.  While this is near 1, it is likely a non-trivial
fraction of nodes will fail to be compatibly split even though the edit distance between
the strings is small compared to $k_i$. As a result, for any given
edit operation, there may be a small but non-trivial chance that the
recovery algorithm fails to recover it.  To ensure
that we get all of the edit operations,
we will need to run the scheme $O(\log n)$ times and include those edit operations produced by
more than half the runs to guarantee that every edit operation gets recovered with good probability.

{\em Location, location, location!}
The recovery procedure identifies pairs of fragments that differ, and can therefore
reconstruct the canonical edits between those fragments.  But as described so far,
there is nothing that allows the reconstruction to pinpoint where these fragments
appear in the full string.  Without this information, we can not combine information
obtained from the independent sketches as just described.  
Our sketch will need an additional component to properly position each recovered fragment.

{\em Location tree.}
We will use technique of Belazzougui and Zhang~\cite{belazzougui_zhang} (suggested to us by Tomasz Kociumaka).
We turn the decomposition tree into a binary tree by expanding each node with $n$ children
into a binary tree of depth $\log(n)$ with $n$ leaves.
For each node in the binary tree, we record the length of the substring represented by its left child.
We watermark this size by the usual Karp-Rabin fingerprint of the node substring,
and we sketch the sizes using the hierarchical mismatch sketch as in the grammar tree, i.e., level by level.
Fragments that contain edit operations will differ in the Karp-Rabin fingerprint so they will reveal the size of their left child.
For a given node that contains an edit operation all of its ancestors on the path to the root will also be watermarked.
Hence we will be able to recover the information about the size of all the left children along the path.
That suffices to calculate the position of each differing fragment.
We will use the same setting of capacities for the hierarchical mismatch scheme that we use for grammars.
That is clearly sufficient as grammars are larger objects than a single integer.

{\em Putting things together.}
The actual sketch consists of multiple independent sketches.  Each of these
sketches is the output of the hierarchical mismatch recovery scheme applied
to each of the levels of the grammar decomposition tree,  and applied
to each of the levels of the binary location tree. 
Its full details are provided in Section~\ref{sec:ED scheme}.

{\em Recovery.}
We  briefly explain the recovery of edit operations from sketches for two strings.
The reconstruction starts by running the recovery procedure for all the hierarchical mismatch recovery sketches.
This recovers various pairs of grammars with information about their location within the original input strings $x$ and $y$.
From those grammars we pick only those which do not have any ancestor grammar node recovered as well.
(Descendant grammars are superseded by the ancestor grammars.)
For every pair of grammars we reconstruct the fragments they represent and compute the associated edit operations.
(The edit operations could be actually computed without decompressing the grammars~\cite{ED_compressed_string_Soda22}.) For each pair of recovered fragments,
we use the location tree sketches to recover
the exact location of that fragment.

So for each  pair of recovered fragments we calculate the canonical sequence of edit operations
and given the exact location of the fragment, we can determine
the exact location within $x$ and $y$ where each edit occurred.  
We repeat this for each independent copy of the sketch.
Our final output is  the set of edit operations  that appear in the majority of the copies.
These are all the operations we were supposed to find.
We provide details of the analysis in Section~\ref{subsec:main result}.
The analysis of the process is somewhat delicate as one needs to deal with various dependencies among the probabilistic events.
Overall it reflects the intuition provided above.

\section{Preliminaries}

\subsection{Strings, sequences, trees, and string decomposition}
\label{subsec:strings}

For an alphabet $\Gamma$, $\Gamma^*$ denotes the set of (finite) strings of symbols from $\Gamma$,
$\Gamma^n$ is the strings of length exactly $n$ and $\Gamma^{\leq n}$ is the set of strings of length at most $n$.
We write $\emptystr$ for the empty string of length 0.  
The length of string $x$ is denoted $|x|$.   
For an index $i \in \{1,\ldots,|x|\}$, $x_i$ is the $i$-th symbol of $x$.

We will consider the input alphabet of our strings to be $\Sigma=\{0,1,\dots\}$ which is the set of natural numbers.
For strings of length up-to $n$ we will assume they are over sub-alphabet $\Sigma_n=\{0,1,\dots,n^3-1\}$
so typically we will assume our input comes from the set $\Sigma_n^{\le n}$.
This assumption is justified as for computing edit or Hamming distance of two strings of length at most $n$ we can hash any larger alphabet randomly to the range $\Sigma_n$ without affecting the distance of the two strings with high probability.

If $\Gamma$ is linearly ordered then $\Gamma^*$ is also linearly ordered under lexicographic order, denoted by $\lexorder$, given by
$x \lexorder y$ if $x$ is prefix of $y$ or if $x_j<y_j$ where $j$ is the least index $i$ for which $x_i \neq y_i$.
We usually write $x < y$ instead of $x \lexorder y$. 

We write $x \concat y$ for the concatenated string $x$ followed by $y$ and for a list $z_1,\ldots,z_k$ of strings
we write $\bigconcat_{i=1}^k z_i$ for $z_1 \concat \cdots \concat z_k$.


\paragraph{Substrings and fragments}
\label{subsec:strings and fragments}
For a string $x$ and an interval $I \subseteq \{1,\dots, |x|\}$, a string $z$ is a \emph{substring of $x$ located at $I$} if $|z|=|I|$ and for all $i\in I$, $z_{i - \min(I)+1}=x_i$.
We denote this substring by $x_I$.  When using intervals to index substrings, it is convenient
to represent intervals in the form $(i,j]=\{i+1,\ldots,j\}$ and $(i,i]$ denotes the empty set for any $i$.
(So a substring is always a consecutive sub-sequence of a string.)
We can also say that $z$ is the {\em substring of $x$ starting at position $\min(I)$}.
Furthermore, $z$ is a {\em substring of $x$} if $z$ is a substring of $x$ starting at some position.  
However, the statement that $z$ is a substring of $x$ says nothing about where $z$ appears in $x$, and there may be multiple (possibly overlapping) occurrences of $z$ in $x$.
For us it will be important where a substring appears.
For a string $x$ and an interval $I \subseteq \{1,\dots, |x|\}$, the {\em fragment located at $I$} is the pair $x_I$ together with $I$. 
 

\paragraph{Sequences and Hamming Distance}
\label{subsec:functions}
We will consider finite sequences of elements from some domain.
It will be convenient to allow sequences to have index sets other than the usual integers $\{1,\ldots,n\}$. 
If $D$ is any set, a \emph{$D$-sequence} is an indexed collection $a=(a_i:i \in D)$.  
A \emph{$D$-sequence over the set $A$} is a $D$-sequence with entries in $A$. 
$A^D$ denotes the set of $D$-sequences over $A$.
The \emph{Hamming Distance} $\HAM(x,y)$ between two $D$-sequences $x$ and $y$ is the number of indices $i \in D$ for which $x_i \neq y_i$.
We let $\indices(x,y)=\{i\in D; x_{i} \neq y_{i}\}$.

\paragraph{Trees}
\label{subsec:trees} 
Our algorithm will organize the processed data in a tree structure.
To simplify our presentation we will give the tree very regular structure.
For finite sets $\LOS_1,\ldots,\LOS_d$, $T(\LOS_1 \times \cdots \times \LOS_d)$ denotes the rooted tree of depth $d$ where
for each $j \in \{1,\ldots,d\}$ every internal node $v$ at depth $j$ has $|\LOS_j|$ children, and the edges
from $v$ to its children are labeled by distinct elements of $\LOS_j$.
Each node $v$ at depth $j$ is identified with the length $j$ sequence of edge labels on the path from the root to $v$; under this correspondence
the set of nodes at level $j$ is $\LOS_1 \times \cdots \times \LOS_j$.   The root is therefore the empty sequence $\emptystr$.
 For an internal node $v$ at depth $j-1$,
its children are nodes of the form $v \concat a$ where $a \in \LOS_j$.  Also the path from $\emptystr$ to $v$ at depth $j$
is equal to the sequence of nodes $v_{\leq 0},v_{\leq 1},\ldots,v_{\leq j}$ where $v_{\leq i}$ is the prefix of $v$ of length $i$.

Usually in this paper, the sets $\LOS_1,\ldots,\LOS_d$ are all equal to the same set $\LOS$ and in this case
$T(\LOS_1 \times \cdots \times \LOS_d)$ is denoted $T(\LOS^d)$.  Usually, $\LOS$ is a  linearly ordered set and it is useful
to visualize the planar drawing of $T$ in which the left-to-right order of the children of an internal node  corresponds to the total ordering on the edge labels.


\paragraph{String decompositions, and tree decompositions}
\label{subsec:tree decomp}
A \emph{decomposition} of a string $x$ is a sequence $z_1,\ldots,z_r$ of strings
such that $x=\bigconcat_{i =1}^r z_i$. 
More generally if $\LOS$ is a linearly ordered set then an $\LOS$-sequence $(z_{i}:i \in \LOS)$ where each $z_{i}$ is a string is a decomposition
of $x$ if $x=\bigconcat_{i \in \LOS} z_i$ where the concatenation is done in the order determined by $\LOS$.
Given a decomposition $(z_{i}:i \in \LOS)$ of $x$, each substring $z_i$ is naturally
associated to a location interval $\zint{z}_{i}=(s_i,t_i]$ where $s_i = \sum_{j \in \LOS, j<i} |z_j|$
and $t_i = s_i + |z_i|$.

Supppose that $z=(z_v:v \in \LOS_1 \times \cdots \times \LOS_d)$ is a labeling of the leaves of
$T=T(\LOS_1 \times \cdots \times \LOS_d)$ by strings.  Given
such a labeling, we extend it to all nodes of the tree so that for
any node $v$, $z_v$ is defined to be the concatenation of
the strings $z_{\ell}$ where $\ell$ is a leaf below $v$ and the concatenation is done in lexicographic order according to the leaves.
If $x_{\emptystr}$ is the string labeling the root, then
the $T$ with the node labeling  $z$ is said to be a \emph{decomposition tree for $x$}.
Note that for any $j \in \{1,\ldots,d\}$, $z^j=(z_v:v \in \LOS_1 \times \cdots \times \LOS_j)$ is the decomposition of $x$ corresponding
to the level $j$ nodes in the tree, and $z^j$ is a refinement
of the decomposition $z^{j-1}$.
In a decomposition tree $z$, each node $v$ corresponds to a specific
fragment whose location in $x$ is the interval $\zint{z}_{v}=(s_v,t_v]$ where $s_v = \sum_{u\in L^{d}, u<v} |z_u|$ and $t_v = s_v + |z_v|$.
Hence, $x_{\zint{z}_{v}} = z_v$.


\subsection{Edit distance and its representation in grid graphs}
\label{sec:edit distance}

For $x \in \Sigma^*$, we consider three {\em edit operations} on $x$:

\begin{itemize}

\item  $\edinsert(i,a)$ where $i \in \{1,\ldots,|x|+1\}$ and $a \in \Sigma$, which means insert $a$ immediately following the prefix of length $i-1$. In the resulting sequence
the  $i$-th entry is $a$.
\item  $\eddelete(i)$ where $i \in \{1,\ldots,|x|\}$, deletes the $i$-th entry of $x$.
\item $\edsubst(i,b)$: replace $x_i$ by $b$.
\end{itemize}

For strings $x,y$, the \emph{edit distance of $x$ and $y$}, $\ED(x,y)$, is the minimum length of a sequence of operations that transforms $x$ to $y$.
It is well-known and easy to show that $\ED(x,y)=\ED(y,x)$.

\paragraph{Representing edit distance by paths in weighted grids.}
\label{subsec:grids}
\label{subsec:weighted grids}
We define $\grid{}$ to be the directed graph whose vertex set $V(\grid)$ is the set $\N \times \N$ (\emph{points})
and whose edge set $\edges(\grid)$ consists of three types of directed edges: {\em horizontal edges} of the form
 $\gpoint{i}{j}\to \gpoint{i+1}{j}$, {\em vertical edges} of the form $\gpoint{i}{j}\to\gpoint{i}{j+1}$
and {\em diagonal edges} of the form $\gpoint{i}{j} \to \gpoint{i+1}{j+1}$ for any $i,j \geq \N$.
For non-empty intervals $I,J \subseteq \N$ not-containing zero, we define the $\grid_{I\times J}$ to be the subgraph of $\grid$ induced on the
set $(I \cup \{\min(I)-1\}) \times (J \cup \{\min(J)-1\})$, and for $P \subseteq E(\grid)$, the restriction of $P$ to $I\times J$ is $P_{I\times J} = P \cap  E(\grid_{I\times J})$.
We call $I\times J$ a {\em box}.
A directed path from $\gpoint{\min(I)-1}{\min(J)-1}$ to $\gpoint{\max(I)}{\max(J)}$ is called a \emph{spanning path} of $\grid_{I\times J}$.

As is well known, the edit distance problem for a pair of strings $x,y$ can be represented as a shortest path problem on a grid
with weighted edges (see e.g. \cite{LMS98}).
The \emph{grid of $x,y$}, $\grid(x,y)$, is the subgraph $\grid_{(0,|x|] \times (0,|y|]}$
with edge set $\edges(x,y) \subseteq \edges(\grid)$.
For an edge $e=\gpoint{i}{j} \to \gpoint{i'}{j'}$ in $\grid(x,y)$, let $x_e = x_{i'}$ if $i'=i+1$ and $x_e = \emptystr$ if $i'=i$.
Similarly, let $y_e = y_{j'}$ if $j'=j+1$ and $y_e = \emptystr$ if $j'=j$.
We assign a cost to edge $e$ to be 0 if $y_e=x_e$ and it is 1 otherwise.
In particular, every horizontal edge and every vertical edge costs 1,
and diagonal edges cost 0 or 1 depending on whether the corresponding symbols of $x$ and $y$ differ.
An edge of non-zero cost is \emph{costly}. 
If $P$ is a set of edges, the \emph{costly part of $P$}, $\costly{P}$, is the set of costly edges.
Define the cost of $P$ to be $\cost(P)=|\costly{P}|$.

We define an \emph{annotated edge} to be a triple $(e,a,b)$ where $a,b \in \Sigma \cup \{ \emptystr \}$.
The \emph{$(x,y)$-annotation of $e$} is the annotated edge $(e,x_e,y_e)$, which is denoted $e^+(x,y)$. 
An annotated edge $(e,a,b)$ is said to be \emph{$(x,y)$-consistent} or simply \emph{consistent}
if $a=x_e$ and $\beta=y_e$.
We emphasize that each edge $e$ has a unique consistent annotation with respect to any given $x$ and $y$. 

For a set of edges $P$ we write $P^+(x,y)$ for the set of annotated edges
$\{e^+(x,y):e \in P\}$.  
In particular if $P$ is a path then $\costly{P^+(x,y)}$ is the set of costly edges of $P$ with their $x,y$-annotations.
When the pair $x,y$ of strings is fixed by the context (which is almost always the case) we  write $e^+$ for $e^+(x,y)$ and 
for a set $P$ of edges, we write $P^{+}$ for $P^+(x,y)$. 
In particular, $\edges^+$ for the set $\{e^+:e \in \edges(x,y)\}$.

It is well known and easy to see that there is a correspondence between spanning paths of $\grid(x,y)$ and sequences of edit operations
that transform $x$ to $y$ where a sequence of $k$ edit operations corresponds to a spanning path of cost $k$. 
Thus, we will refer to a spanning path of $\grid(x,y)$ as an {\em alignment of $x$ and $y$}.
We have:

\begin{proposition}
\label{prop:alignment}
$\ED(x,y)$ is equal to the minimum cost of an alignment of $x$ and $y$, i.e., the minimum over
all alignments $P$ of $|\costly{P}|$.
\end{proposition}

We have the following:
\begin{proposition}
\label{prop:costly}
\begin{enumerate}
\item 
An alignment $P$ of $x$ and $y$ is uniquely determined by  $\costly{P}$.
\item 
For any alignment $P$ of $x$ and $y$, given the set $\costly{P^+}$ of
annotated costly edges and either of the strings $x$ and $y$, the other string is determined.
\end{enumerate}
\end{proposition}

\begin{proof}
For the first part,
an alignment $P$ can be partitioned into subpaths (sets of edges) $P_0,P_1,\ldots,P_k$
where for each even $i$, $P_i$ consists  of 0-cost edges, and for $i$ odd, $P_i$ consists
of costly edges. All the paths are nonempty
except for (possibly) the first.  Trivially $\costly{P}$  is the union of
all the paths $P_i$ for $i$ odd.  For each $i$ even, we can determine $P_i$ since we know
the start and end vertex, and the path consists entirely of diagonal edges.

For the second part, we show that  $\costly{P^+}$ and $x$ determine $y$; the result
with $x$ and $y$ exchanged is proved similarly.  From the first part of the
proposition $\costly{P}$ determines $P$. Let $e_1,\ldots,e_t$ be the subsequence of 
$P$ obtained by deleting all vertical edges.  Then $|y|=t$, since the undeleted 
edges are in 1-1 correspondence with the indices of $y$.  Furthermore
$y_j$ is determined by the $j$-th undeleted edge as follows: if the $j$-th edge is costly,
(either vertical, or a costly diagonal edge) then its $y$ annotation is equal to $y_j$,
and if the $j$-th edge is a non-costly diagonal edge $\gpoint{i-1}{j-1} \to \gpoint{i}{j}$,
then $y_j=x_i$.
\end{proof}

\label{subsec:fragment pairs}

A pair $x,y$ of strings together with a box $I \times J$ with $I \subseteq (0,|x|]$ and $J \subseteq (0,|y|]$ specifies the edit distance sub-problem $\ED(x_I,y_J)$.  
$\grid_{I \times J}(x,y)$ denotes the (edge-weighted) sub-graph of $\grid(x,y)$ induced on $(I \cup \{\min(I)-1\}) \times (J \cup \{\min(J)-1\})$.

There may be many optimal alignments.  
We will need a {\em canonical alignment} for each $x$ and $y$ that is unique.  
For any graph $\grid_{I\times J}(x,y)$, define the \emph{canonical alignment of $\grid_{I\times J}(x,y)$} as follows:
Associate each path in $\grid$ to the sequence from $\{\mathbf{vertical}, \mathbf{diagonal}, \mathbf{horizontal}\}$ which records
the edge types along the path. 
The \emph{canonical alignment of $\grid_{I\times J}(x,y)$} is the optimal spanning path of $\grid_{I\times J}(x,y)$ that is lexicographically maximum with respect to the order $\mathbf{vertical}> \mathbf{diagonal} > \mathbf{horizontal}$.  
The {\em canonical alignment $\canon(x,y)$ of $x$ and $y$} is the canonical alignment of $\grid(x,y)$. 

The proof of the following is left to the reader.

\begin{proposition}
\label{prop:subgrid}
For strings $x,y$ and box $I \times J \subseteq (0,|x|] \times (0,|y|]$, the (edge-weighted) graph $\grid_{I\times J}(x,y)$ is isomorphic (in the graph theoretic sense) to $\grid(x_I,y_J)$
and so $\ED(x_I,y_J)$ is equal to the length of the shortest spanning path of $\grid_{I \times J}(x,y)$.
Also, the canonical alignment of $x_I$ and $y_J$ is isomorphic to the canonical alignment of $\grid_{I \times J}(x,y)$ (when viewed as paths of their respective graphs).
\end{proposition}

Let $P$ be an alignment of $x$ and $y$.  
A box $I \times J$ is \emph{compatible with $P$} provided that 
$P$ passes through $\gpoint{\min(I)-1}{\min(J)-1}$ and $\gpoint{\max(I)}{\max(J)}$, and 
for such a box, the \emph{restriction $P_{I \times J}$ of $P$ to $I \times J$} is 
the portion of $P$ joining $\gpoint{\min(I)-1}{\min(J)-1}$ and $\gpoint{\max(I)}{\max(J)}$. 
This restriction is an alignment for $\grid_{I \times J}(x,y)$. 
The following proposition is straightforward.

\begin{proposition}
\label{prop:optimal restriction}
If $P$ is an optimal alignment of $x$ and $y$ and $I \times J$ is compatible with $P$ then $P_{I \times J}$ is an optimal alignment for the sub-problem $\grid_{I \times J}(x,y)$.
\end{proposition}


For strings $x$ and $y$, we say a box $I \times J$ is $(x,y)$-compatible if $I \times J$ is compatible with the canonical alignment $\canon(x,y)$.
We have:

\begin{proposition}
\label{prop:canonical}
Let $x,y$ be strings and let $I \times J$ be a box that is $(x,y)$-compatible.  
The restriction $\canon(x,y)_{I \times J}$ is equal to the canonical alignment of $\grid_{I \times J}(x,y)$.
\end{proposition}

\begin{proof}
By Proposition~\ref{prop:optimal restriction}, the restriction $\canon(x,y)_{I \times J}$ is an optimal spanning path
for $\grid_{I \times J}(x,y)$.  We claim that this restriction is equal to the canonical alignment of $\grid_{I \times J}(x,y)$.  
Otherwise we could modify $\canon(x,y)$ by replacing the subpath on $I\times J$ by the canonical alignment of $\grid_{I \times J}(x,y)$ 
and the result would be lexicographically larger.
\end{proof}

\subsection{Computational Considerations}
\label{sec:compconsider}

Our construction uses a finite field of large characteristics to design hierarchical mismatch recovery schemes (a certain kind of Hamming sketches).
For that our algorithm picks a prime $p$ of size roughly $n^{O(\log n)}$  during its initialization.
Such a prime can be picked at random and tested for its primality using standard algorithms (randomized or deterministic).
The prime $p$ can be stored as a part of the sketch.
All arithmetics over $\bigfield_p$ can be done in time poly-logarithmic in $p$ or $n$, respectively.
Our sketch could be modified to use primes of only polynomial size but for the simplicity of presentation we use a large prime.




\subsection{Function families and randomized functions}
\label{subsec:fcn family}

Let $A,B$ be sets.  
We consider families $\textbf{f}=(f_{\rho}:\rho \in \paramspace)$ of functions from $A$ to $B$. The
set $\paramspace$ is the \emph{randomizing parameter space} and the subscript $\rho \in \paramspace$ is the \emph{randomizing parameter}.
If $\mu$ is a probability distribution over $\paramspace$, $\mu$ induces a distribution on the
function family $(f_{\rho}:\rho \in \paramspace)$, and we say that $f_{\rho}$ is a randomized function.  
We often suppress the index $\rho$ and say that $f$ is a randomized function.

Evaluating a randomized function at a domain element $x$ is in two parts.  
The first part is the selection of the parameter $\rho \in \paramspace$ according to $\mu$, which is done by a randomized algorithm.  
Once $\rho$ is selected, $f_{\rho}$ is a fixed (deterministic) function and 
the second part is to evaluate $f_{\rho}(x)$; 
the algorithm for doing this may possibly be randomized.  
We say that $(f_{\rho},\mu)$ is \emph{efficiently computable} provided that the sampling of  $\rho$ according to $\mu$ and the evaluation $f_{\rho}(a)$ can be done in time polynomial in the size of the bit-representation of $a$.

A randomized $A \rightarrow B$  function $(\textbf{f}, \mu)$ is \emph{pairwise independent} provided that
for all $a,a' \in A$ and $b,b' \in B$ such that $a \neq a'$ we have
$\prob_{\rho \sim \mu} [(f_{\rho}(a)=b) \AND (f_{\rho}(a')=b')] = \frac{1}{|B|^2}$.
We have:

\begin{proposition}[Dietzfelbinger~\cite{Dietzfelbinger96}]
\label{prop:pairwise}
For all positive integers $m,n$ with $m$ a power of 2 there is a pairwise independent family of functions $(\textbf{h},\mu)$ 
mapping $\{0,\ldots,n-1\}$ to $\{0,\ldots,m-1\}$ that is efficiently computable, where members of the family can be specified with $O(\log m + \log n)$ bits.
\end{proposition}

We note here a version of the Chernoff-Hoeffding bound~\cite{AS08} that will be needed in our algorithm analysis.

\begin{lemma} 
\label{lemma:CH}
Let $X_1,\ldots,X_r$ be i.i.d. 0-1 random variables with $\prob[X_i=1]=p$. For $\delta \geq 0$,
\[
\prob[|\frac{1}{r}\sum_{i=1}^r X_i - p| \geq \delta] \leq 2e^{-2\delta^2r}.
\]
\end{lemma}

\section{Three auxiliary procedures}
\label{sec:auxiliary}

The sketching scheme for edit distance that we present in Section~\ref{sec:ED scheme} uses several auxiliary sketching schemes.  Three of these are from previous work and one is new to this paper.  In this section, we describe the key properties of the previous three schemes.

\subsection{Fingerprinting}
\label{subsec:fingerprinting}

The \emph{fingerprinting problem} for a $\Sigma^*$ is to define a family of sketching functions that distinguishes between distinct
elements of $\Sigma^*$.  More formally, given a parameter $n$ we want a family $\fingerprint_{\rho}(\cdot;n):\Sigma_n^{\le n} \rightarrow \Z^+$
such that for any strings $x,y \in \Sigma_n^{\le n}$ with $x \neq y$ $$\prob_\rho[\fingerprint_{\rho}(x;n) = \fingerprint_{\rho}(y;n)] < 1/n^4.$$  

The following classic result of Karp-Rabin \cite{rabin_karp} provides an efficient fingerprinting scheme. We refer to it as the {\em Karp-Rabin fingerprint}.

\begin{theorem}
\label{thm:fingerprint}
There is an efficiently computable randomized function  $\fingerprint_{\rho}(\cdot;n):\Sigma_n^{\le n} \rightarrow \{1,\dots,n^5\}$ such that for any $x,y \in \Sigma_n^{\le n}$ with $x \neq y$,
$\prob_\rho[\fingerprint_{\rho}(x;n) = \fingerprint_{\rho}(y;n)] < 1/n^4$.  The number of bits needed to describe $\rho$ is
$O(\log n)$. The time to compute $\fingerprint_{\rho}(x;n)$ is $O(|x| \cdot \log^{O(1)} n)$.
\end{theorem}

\subsection{Threshold edit distance fingerprinting}
\label{subsec:ED threshold}

In the \emph{threshold edit distance fingerprinting} we
are given a parameter $n$ and a threshold parameter $k$. 
We want a family of sketching functions  $\gthreshsk_{\rho}(\cdot;k,n) : \Sigma^{*} \to \Z^+$ 
such that for all $x,y \in \Sigma_n^{\leq n}$ such that $\ED(x,y)\geq k$, $\prob[\gthreshsk_\rho(x;k,n) = \gthreshsk_\rho(y;k,n)] \leq 1/n^4$,
i.e. the sketch is very likely to distinguish strings $x,y$ that are far.  We also want that for some \emph{inaccuracy gap} $s>1$,
for strings $x,y$ such that $\ED(x,y)<\frac{k}{s}$, $\prob[\gthreshsk(x) \neq \gthreshsk(y)]$ is small.  Precisely we want that 
for all $x,y$: 

\[
\prob[\gthreshsk_{\rho}(x) \neq \gthreshsk_{\rho}(y)] \leq \frac{\ED(x,y)}{k} \cdot s.
\]
Note that the requirement becomes easier as $s$ gets larger.
The problem of constructing a threshold edit distance fingerprinting scheme with inaccuracy gap $s>1$ is closely related to the problem of
finding an approximate embedding function $f$ that maps strings to vectors in $\mathbb{R}^d$ (for some $d$) so that
$\ED(x,y)$ is approximated by $|f(x)-f(y)|_1=\sum_{i=1}^d |f(x)-f(y)|$ within a factor $s$.    
This latter problem was investigated by Ostrovsky and Rabani  \cite{OR07} and their results yield the following consequence:

\begin{theorem}
\label{thm:OR}
There is an algorithm $\threshsk(x;k,n)$ for the threshold edit distance fingerprinting problem that
for arbitrary $k,n \in \mathbb{N}$ has inaccuracy gap $\orgap=2^{O(\sqrt{\log(n)\log\log(n)})}$ and for any string $x \in \Sigma^{\le n}_n$ gives a fingerprint of value at most $O(n^4)$.
The fingerprinting algorithm uses a randomness parameter of length $O(\log^3 n)$ and runs in time polynomial in $n$ and $k$.
\end{theorem}


We refer to $\threshsk(x;k,n)$ as Ostrovsky-Rabani fingerprint.

We make a few remarks on this theorem.\footnote{We had difficulties reading the paper \cite{OR07}. 
Indeed, there is a minor correctable issue in the way Lemmas 8 and 9 are presented in the paper. 
The lemmas claim output from $\ell_1$ which allows for vectors with arbitrary real numbers. Indeed, because of the scaling in their proofs
they do output vectors with real numbers.
However, both lemmas need to be applied iteratively and they assume input to be a $0$-$1$-vector. 
This disparity can be corrected using standard means but it requires additional effort. 
}
\cite{OR07} does not provide explicit bound on the amount of randomness needed. 
In particular, Lemma 9 samples random subsets and uses tail-bound inequalities to bound the probability of bad events.
One can use $O(\log n)$-wise independent samples to reduce the necessary randomness~\cite{SSS93}.
 
\subsection{The procedure \bdecomp}
\label{subsec:bdecomp}

The previous sketch for edit distance of \cite{BK23}  gave
sketches of size $\tilde{O}(k^2)$ that allowed for recovery of edit distance between strings $x,y$ with $\ED(x,y) \leq k$.   That work is the
starting point for this paper, and one of the key
procedures in their algorithm is an essential part of our sketch-and-recover scheme, and we now describe its important properties.

\cite{BK23} gives a decomposition algorithm \bdecomp{} that takes a string $x$ of length at most $n$ and partitions $x$ into
a sequence of fragments where each fragment is represented by a small {\em grammar} of size $\OO(k)$.
A {\em grammar} is a subset of {\em rules} from a certain domain of size polynomial in $n$. 
For us the actual meaning of the grammar is irrelevant; 
however we use a procedure  $\decode$ implicitly defined in \cite{BK23} that takes as input a grammar 
and either outputs a string in $\Sigma^{\leq n}$ or \undefnd.
We will represent grammars by their characteristic vector within their domain. 
We say that a bit-vector is \emph{$t$-sparse} if it has at most $t$ 1's in it.

We now restate the properties of algorithm \bdecomp{} in a form suitable for us.
The algorithm takes as input a string $x$, parameter $n$ and an integer \emph{sparsity parameter} $\gramsize$, 
where $|x|\le n$.
We denote this by $\bdecomp(x;n,\gramsize)$.

$\bdecomp(x;n,\gramsize)$ outputs:
\begin{itemize}
\item A string decomposition of $x$, $z=(z_i:i \in \BASE)$, where $\BASE$ is the set $\{0,1\}^{\lceil \log n \rceil}$ with the lexicographic order.  (Some strings may be empty.)
\item A collection $\bvec=(\bvec_i:i \in \BASE)$ of $\gramsize$-sparse bit-vectors of length $\bvsize=\bvsizeval$ such that if $z_i \neq \emptystr$ then $\decode(\bvec_i)=z_i$ (so $\bvec_i$ is a $\gramsize$-sparse encoding of $z_i$)
and if $z_i=\emptystr$ then $\bvec_i$ is the all 0 vector  $0^{\bvsize}$. Furthermore, $\bvec_i$ has the following
\emph{minimality property}: for any bit-vector $b$ that is bit-wise less than $\bvec_i$, $\decode(b)$ is \undefnd.
(For technical convenience we require $\decode(0^\bvsize)$ to be also \undefnd{} although $0^\bvsize$ represents the empty string.)
Each $\bvec_i$ is represented as a list of positions that are set to $1$.
\end{itemize}

%
%
%

\begin{theorem}
\label{thm:BK}
The algorithm $\bdecomp$ has the following properties for all $n,\gramsize$:  
\begin{enumerate}
\item
\label{BK1}
For any input $x$ of length at most $n$, 
the running time of $\bdecomp(x;n,\gramsize)$ as well as the total number of ones in $\bvec_i$'s is bounded by $O(|x| \cdot \log^{O(1)} n)$.
\item 
\label{BK2}
For any pair of inputs $x,y$, suppose 
$z=(z_i:i \in \BASE)$ and $\bvec=(\bvec_i:i \in \BASE)$ is the output of $\bdecomp(x;\gramsize)$, and 
$z'=(z'_i:i \in \BASE)$ and $\bvec'=(\bvec'_i:i \in \BASE)$ is the output of $\bdecomp(y;\gramsize)$.  
\begin{enumerate}
\item
\label{BK2a}
With probability at least $1-\bkfailure\cdot \frac{\ED(x,y)}{\gramsize }$ where $\bkfailure=O(\log^4 n)$,
for all $i \in W$, the box $\zint{z}_i \times \zint{z'}_i$ is $(x,y)$-compatible whenever $z_i$ is non-empty, and $z_i$ is non-empty iff $z'_i$ is non-empty. 
\item
\label{BK2b}
For all $i \in W$, $\HAM(\bvec_i,\bvec'_i) \leq \dfactor \cdot \ED(z_i,z'_i)$ where $\dfactor=O(\log^2 n)$.
\end{enumerate}
\end{enumerate}
\end{theorem}

The significance of Item~\ref{BK2a} is that if for each $i\in\BASE$, the box $\zint{z}_i \times \zint{z'}_i$ is $(x,y)$-compatible 
then $\ED(x,y)=\sum_{i\in \BASE} \ED(z_i,z'_i)$.

Note, for Item 1, \cite{BK23} allows the decomposition procedure to fail with probability $O(1/n)$. 
We modify the procedure that instead of failing it produces the trivial decomposition of its input into fragments of length 1.
When $x\neq y$, this increases the failure probability in Item 2a by a negligible amount.
(For $x=y$, the procedure never fails.)
For Item 2a, \cite{BK23} does not specify the success probability explicitly the way we do but our formula is immediate from their analysis
of probability of what they call {\em undesirable split}. 
The upper bound on the difference between two grammars in terms of edit distance of $x$ and $y$ is also implicit in the same analysis. 

Given a grammar $G$ by its sparse representation as a list of $t$ positions of 1's, $\decode(G)$ runs in time $O(t+|x|)$ if $G$ represents a string $x$
and in time $O(t)$ otherwise. A grammar representing $x$ contains at most $\dfactor \cdot |x|$ ones.

\section{Sketch-and-recover schemes}
\label{subsec:difference recovery}

We consider the following general setting: We have a universe of objects (such as
strings, matrices, or graphs) for which we have a specific 
way of comparing two objects in the universe.  In this paper the universe is strings over a given
alphabet and for each pair $x,y$ of strings, we are interested in determining $\ED(x,y)$.
A {\em sketch-and-recover} scheme consists of two algorithms: 
a sketching algorithm that takes a string $x$ over some alphabet and produces a short \emph{sketch}  $\SK(x)$ 
and a {\em recovery algorithm} that takes a pair of sketches $\SK(x)$ and $\SK(y)$ and recovers the
desired information about the relationship between $x$ and $y$.


The sketch and recovery algorithms typically take auxiliary parameters.
For edit distance, the auxiliary parameters are $n$, an upper bound on the length
of strings the sketch will handle, and $k$, a threshold parameter such that given any two
inputs $x,y$ if $\ED(x,y) \leq k$ the recovery algorithm will determine $\ED(x,y)$.  (The precise
requirements of our scheme will be given later.)

Sketch-and-recover schemes are typically randomized.  We assume that the algorithms use
a \emph{randomizing parameter} $\rho$, which is computed from a publicly available random string
so that the same value of $\rho$ is used by the sketching and recovery algorithms.
We allow the recovery algorithm to fail with a small probability over $\rho$ (with output \undefnd{}).
We usually suppress the parameters $\rho,n$ and $k$ and write simply $\SK(x)$.
By the {\em total sketch} we mean the pair $(\rho,\SK_{\rho}(x;k,n))$, the sketch together with its randomizing parameter. 

We want to implement sketch-and-recover schemes by efficient algorithms.  Here efficiency is measured by:
\begin{itemize}
\item The \emph{sketch length} $\bitlength(\SK_{\rho}(x; k,n))$, which is the number of bits in the binary encoding of  $\SK_{\rho}(x; k,n)$.
\item The \emph{randomness length} 
$\bitlength(\rho)$ is the number of bits for 
 (the binary encoding of) the parameter $\rho$.  We think of $\rho$ as being determined by a randomized algorithm
which runs using a shared public string of random bits.  We could just take $\rho$ to be this shared
string, but in some cases the number of bits needed to represent $\rho$ may be less than the number of bits needed
to generate $\rho$.  For example, if $\rho$ is supposed to be uniformly distributed among $s$-bit primes, then $\bitlength(\rho)=s$ but the
number of random bits needed by an algorithm to generate $\rho$ is larger than $s$. 
\item The running time
of the sketch and recovery functions.  
\end{itemize}

Our primary focus here  is on achieving small sketch length.  We also want the running time of the sketch and recovery
algorithms to be at most polynomial in the length of the strings being sketched, but  we don't try
to optimize the polynomial here.

In designing a sketch for recovering edit distance, we want our recovery algorithm
to do more than just recover the edit distance of two strings $x$ and $y$ from their sketch;
we would also like the recovery algorithm to provide enough information so that given
one of the strings, and the output of the recovery algorithm, one can determine
the other string.  For example, in the related but easier case of sketches for Hamming distance,
we want to recover the \emph{mismatch information} for $x,y$,  $\MI(x,y)$, which is the set $\{(i,x_i,y_i):i \in \indices(x,y)\}$, where $\indices(x,y)=\{i:x_i\neq y_i\}$.  We refer
to $(i,x_i,y_i)$ as a \emph{mismatch triple} for $x,y$

The analog of mismatch information for edit distance
is the sequence of edit operations which is determined by the \emph{set of costly
annotated edges of the canonical alignment of $x$ and $y$}, as defined in Section 
~\ref{sec:edit distance}.  In proposition~\ref{prop:costly} it is shown that an alignment is
completely determined by its set of costly edges and also that given the set of
costly edges with their annotations and either string $x$ or $y$, one can recover
the other string.  We require the alignment to be canonical because the optimal alignment
is not, in general, unique and we need that
the set being recovered is uniquely determined by $x$ and $y$.

\subsection{Hierarchical  mismatch recovery} 
\label{subsec:HMR}
 
The problem of Hamming distance recovery 
can be viewed as the special case of edit distance recovery when the only edit operations
allowed are substitutions (which preserve the position of letters in the string).
To formalize the problem we fix an index set $D$ and alphabet $\Gamma$
and take as our universe the set of $D$-sequences over  $\Gamma$.  
We want a sketch-and-recover scheme that supports \emph{mismatch recovery},
which is the recovery of the mismatch information 
$\MI(u,w)$ defined in Section~\ref{subsec:difference recovery}.  

In the standard formulation, we have a parameter $k$ and require recovery of the mismatch information
$\MI(x,y)$
from the sketches of $x$ and $y$ whenever their Hamming distance is at most $k$.
We will need a more general formulation of mismatch recovery in which we only require recovery of mismatch information for some indices.
Those indices will be determined for each pair of strings $u$ and $w$ separately.

This more general formulation is called \emph{targeted mismatch recovery}.  
Before giving a formal definition, we give a simple example.  
Consider sequences of length $n=m^2$ for some integer $m$, 
where we think of a sequence $u$ as the concatenation of $m$ fragments $u^1,\ldots,u^m$ each of length $m$.  
For strings $u,w$ we say that fragment $i$ is overloaded if there is more than one mismatch in the fragment, 
and underloaded if there is at most one mismatch in the fragment. 
Let $F(u,w)$ be the set of mismatch indices belonging to underloaded fragments.   
The set $F(u,w)$ defines a targeted mismatch recovery problem where our goal is to provide sketching and recovery algorithms 
so that for any pair $u,w \in \Gamma^D$, given the sketches of $u$ and $w$, 
the recovery algorithm finds all mismatch triples corresponding to $F(u,w)$.  Thus
one needs to recover up to $m$ mismatches in underloaded fragments even if other fragments are overloaded fragments. This particular example has a simple sketch of size $\OO(m)$ by separately constructing
an $\OO(1)$-size sketch of each fragment that supports recovery of a  single error.

In general, a \emph{targeted mismatch recovery} problem is specified by a target function $F$ which for each pair $u,w \in \Gamma^D$, is a subset $F(u,w)$ of the set $\indices(u,w)$ of indices $i$ such
that $u_i\neq w_i$.  We write  $\MI_F(u,w)$ as the set $\{(i,u_i,w_i):i \in F(u,w)\}$, this
is the subset of mismatch triples corresponding to indices in $F(u,w)$.
The targeted mismatch recovery problem for a target function $F$ is denoted $\PMR(F)$.   
The output of the $\RECOVER$ function applied to two sketches is required to be  a set of triples $(i,a,b)$ where $i \in D$
and $a,b \in \Gamma$. 
Such a triple is called a \emph{mismatch triple}.  
The success conditions for $\PMR(F)$ for the pair $u,w$ are:

\begin{description}
\item[\textbf{Soundness}.] $\RECOVER(\SK(u),\SK(w)) \subseteq \MI(u,w)$, i.e. every mismatch claim is correct.
\item[\textbf{Completeness}.]  $\MI_F(u,w) \subseteq \RECOVER(\SK(u),\SK(w))$ so the algorithm recovers all mismatch triples indexed by $F(u,w)$. 
\end{description}

Notice, we do not require that $\RECOVER(\SK(u),\SK(w)) \subseteq \MI_F(u,w)$, i.e., we allow
the recovery algorithm to output mismatch triples not indexed by $F(u,w)$ as long as they are
correct mismatch triples.  Thus, while Completeness depends on the target $F$, Soundness does not.
A scheme for targeted mismatch recovery with target function $F$ has failure probability at most $\delta$ provided that for any  $u,w \in \Gamma^D$, with probability at least $1-\delta$, both Completeness and Soundness hold.
Here, the probability is taken over the randomness of the sketching and recovery algorithms.
(Our recovery algorithms will actually be deterministic, so all the randomness is coming from the sketch.)

\emph{Hierarchical  mismatch recovery} (\HMR) is a special case of targeted mismatch recovery
It is  applicable to $D$-sequences indexed by a product set $D=\LOS_1 \times \cdots \times \LOS_d$, which we view as the leaves of the tree $T=T(\LOS_1 \times \cdots \times \LOS_d)$ as in Section~\ref{subsec:trees}.
The target function depends on a \emph{capacity function} $\treecap$ and an \emph{overload parameter} $\loadpar$.    The capacity function is a positive valued function defined on the levels of the tree
with $\treecap_j$ denoting the capacity for nodes at level $j$.
We require  $1\le \treecap_{\tdepth} \le \treecap_{\tdepth-1} \le \cdots \le \treecap_{0}$.  
The overload parameter $\loadpar$ is a positive integer.

The triple consisting of $\LOS_1 \times \cdots \times \LOS_d$, $\treecap$, and $\loadpar$ determines
a \emph{hierarchical mismatch recovery problem} which is a targeted mismatch recovery
problem whose target set on sequences $u,w$, denoted $F_{T,\treecap,\loadpar}(u,w)$, is defined as follows.


First, the 
\emph{load function  $\hmrload{\treecap}(u,w)$} induced by $\treecap$ and the  pair $u,w$ is defined
on the vertices of the tree $T$ as follows:
 
\begin{itemize}
\item For a leaf $v\in \LOS^d$, $\hmrload{\treecap}_v=1$ if $u_v\neq w_v$, and is 0 otherwise.
\item For an internal node $v$ at level $j<d$, $\hmrload{\treecap}_v=\min(\treecap_j, \sum_{v' \in \child(v)}\hmrload{\treecap}_{v'})$.
\end{itemize}

Clearly $\hmrload{\treecap}_v \leq \treecap_j$ for every vertex $v$ at level $j$.
We say that an internal node $v$ at level $j$ is \emph{$\loadpar$-overloaded}, if $\hmrload{\treecap}_v \geq \frac{1}{\loadpar}\treecap_j$ and is \emph{$\loadpar$-underloaded} otherwise.
(Later we will fix the parameter $\loadpar$ to be $4\tdepth$, and refer to nodes simply as overloaded or
underloaded.)

Notice that a $\loadpar$-underloaded $v$ satisfies $\hmrload{\treecap}_v = \sum_{v' \in \child(v)} \hmrload{\treecap}_{v'}$.
A leaf is said to be \emph{$\loadpar$-accessible} if its path to the root consists entirely of $\loadpar$-underloaded nodes.

We define $F_{T,\treecap,\loadpar}(u,w)$ to be the set $\loadpar$-accessible leaves $\ell$ where $u_\ell \neq w_\ell$.  Intuitively, the leaves below a  $\loadpar$-overloaded node are ``crowded'' with mismatches and we do not require them to be recovered.  
We denote by $\HMR(T,\treecap,\loadpar)$ the targeted mismatch recovery problem with target function $F_{T,\treecap,\loadpar}$.
In Section~\ref{subsec:hmr proof} we will prove:

\begin{theorem}
\label{thm:HMR}
Let $T=T(\LOS_1 \times \cdots \times \LOS_d)$ be a level-uniform tree as defined in Section~\ref{subsec:trees}
and let $\treecap=(\treecap_0,\dots,\treecap_\tdepth)$ be a capacity function  $1\le \treecap_{\tdepth} \le \treecap_{\tdepth-1} \le \cdots \le \treecap_{0} \le \prod_{j=1}^d|\LOS_j|$.  
Assume that for all $j \in \{0,\ldots,d\}$, $|\LOS_j|$ and all $\treecap_j$ are powers of two. 
Let $\Gamma=\mathbb{F}_p$ where $p \geq 4 \prod_{j=1}^d|\LOS_j|^2$.
There is a sketch-and-recover scheme for hierarchical mismatch recovery for $(\LOS_1 \times \cdots \times \LOS_d)$-sequences over $\Gamma$
defined by procedures $\hmrsketch$ and $\hmrrecover$ that given $\delta>0$ satisfies:
\begin{enumerate}
\item If $\loadpar \geq 4d$ then for any two inputs $x$ and $y$, the probabilty that $\HMR(T,\treecap,\loadpar)$  fails to satisfy Soundness and Completeness is at most $\delta$.
\item The sketch bit-size is $O(\treecap_0 \cdot \log |\Gamma|\cdot (\sum_{j=1}^d\log |\LOS_j| + \log(1/\delta)))$ bits. 
\item The sketching algorithm runs in time $O(\prod_{j=1}^d|\LOS_j| \cdot \log^{O(1)} |\Gamma| \cdot \log(1/\delta))$.
\item The recovery algorithm runs in time $O(\treecap_0 \cdot  \log^{O(1)} |\Gamma| \cdot \log(1/\delta))$.
\item If the $(\LOS_1 \times \cdots \times \LOS_d)$-sequence $u$ is given via the sparse representation $\{(j,u_j):j \in \supp(u)\}$ where
$\supp(u)=\{\ell \in \LOS_1 \times \cdots \times \LOS_d:u_{\ell} \neq 0\}$ then the time to construct the sketch
is $O((\treecap_0 + |\supp(u)|) \cdot \log^{O(1)} |\Gamma| \cdot \log(1/\delta))$.
\item The number of mismatch pairs output by the algorithm is at most the capacity $\treecap_0$ of the root.
\end{enumerate}
\end{theorem}

 The sketch $\hmrsketch(u)$ depends  on  $\LOS_1 \times \cdots \times \LOS_d$, $\treecap$, $\loadpar$, and the error parameter $\delta$ and when we apply this sketch in our edit distance algorithm we denote it by  $\hmrsketch(u;\LOS_1 \times \cdots \times \LOS_d,\treecap,\loadpar,\delta)$ for the sketch of $u$
and $\hmrrecover(u',w';\LOS_1 \times \cdots \times \LOS_d,\treecap,\loadpar,\delta )$ for the recovery function which takes as input the sketches $u'$ and $w'$ output by $\hmrsketch$ on strings
$u'$ and $w'$.  The dependence of both the sketching and recovery functions on the random
parameter $\rho$ is left implicit.

We mention that our application will require a sketch-and-recovery scheme
for the trivial case that $d=0$.  In this case the tree is a single node
and $u \in \Gamma$ and we take the sketch to consist of the value itself, and
recovery is trivial.

\subsection{Superposition sketch-and-recover schemes for targeted mismatch recovery}
\label{subsec:superposition}

The next two subsections present a general approach to targeted mismatch recovery.  We apply this approach
in Section~\ref{subsec:hmr proof} to construct the sketch-and-recover scheme $\hmrsketch$ and $\hmrrecover$
for hierarchical mismatch recovery and prove Theorem~\ref{thm:HMR}.

We make the following assumptions:
\begin{itemize}
\item $D=\{0,\ldots,|D|-1\}$.
\item $\Gamma$ is the field $\mathbb{F}_p$ for some prime larger than $|D|$ so $D \subseteq \Gamma$.
\end{itemize}

If these assumptions do not hold we can often reduce our situation to one where it does hold.  
We do not need $D$ to be integers; it is enough that there is an easily computable 1-1 mapping $m$ from $D$ to the nonnegative integers.  
Letting $m=\max_{j \in D} m(j)+1$, we can think of $D$ as a subset
of $\{0,\ldots,m-1\}$, and enlarge the domain to $\{0,\ldots,m-1\}$ defining
any $D$-sequence to be 0 on those indices outside of $D$.  
Similarly we can replace the range $\Gamma$ by $\mathbb{F}_p$ for some $p$
whose size is at least $\max\{|\Gamma|,|D|\}$, where we interpret
$\Gamma$ as a subset of $\mathbb{F}_p$ via some easily computable 1-1 map.

Let $u,w$ be $D$-sequences over $\Gamma$.
Here we use a basic technique from \cite{PL07} (see also \cite{DBLP:conf/soda/CliffordEPR09,rollinghashSODA2019} for related constructions), that allows for the recovery of $\MI(u,w)$ at a specific index $i$. 
For a parameter $\alpha \in \Gamma$, the \emph{trace} of  $u \in \Gamma^D$ (with respect to $\alpha$), denoted $\tr_{\alpha}(u)$ is the $D$-sequence over $\Gamma^4$ where for each $i \in D$, $\tr_{\alpha}(u)_i$ has entries:

\begin{eqnarray*}
\tr_{\alpha}(u)_{i,\tval}&=&u_i,\\
\tr_{\alpha}(u)_{i,\tprod}&=&i\cdot u_i,\\
\tr_{\alpha}(u)_{i,\tsq}&=&u_i^2,\\
\tr_{\alpha}(u)_{i,\thash}&=&\alpha^iu_i.
\end{eqnarray*}
  
All the calculations are done over $\mathbb{F}_p$.
We refer to a vector in $\Gamma^4$ with indices from $\{\tval$, $\tprod$, $\tsq$, $\thash\}$ as a \emph{trace vector} and we refer to $\alpha$ as the \emph{trace parameter}

For $u,w \in \Gamma^D$, the \emph{trace difference of $u,w$} is $\Delta_{\alpha}(u,w)=\tr_{\alpha}(u)-\tr_{\alpha}(w)$.  Here, for each $i$, $\Delta_{\alpha}(u,w)_i=\tr_{\alpha}(u)_i-\tr_{\alpha}(w)_i$ is a trace vector obtained by coordinate-wise subtraction.   

Define the function $\recover$ which maps trace vectors $t$  to $\Gamma^3$ as follows:
\begin{eqnarray*}
\recover(t)_{\dindex}&=&\frac{t_{\tprod}}{t_{\tval}}\\
\recover(t)_{\xval}& = & \frac{t_{\tsq}+t_{\tval}^2}{2t_{\tval}}\\
\recover(t)_{\yval}&=& \frac{t_{\tsq}-t_{\tval}^2}{2t_{\tval}}.
\end{eqnarray*}  

It is easy to check:

\begin{proposition}
\label{prop:trace}
The mismatch information for $(u,w)$ at $i$ is completely determined by  
$\Delta_\alpha(u,w)_i$ as follows: $i$ is a mismatch index of $u,w$ if and only if
$\Delta_\alpha(u,w)_i \neq 0$ and for such an $i$,
\begin{eqnarray*} 
\recover(\Delta_\alpha(u,w))_{i,\dindex} &=& i,\\
\recover(\Delta_\alpha(u,w))_{i,\xval} &=& u_i,\\
\recover(\Delta_\alpha(u,w))_{i,\yval} &=& w_i,
\end{eqnarray*}
and therefore $\recover(\Delta_\alpha(u,w))_i=(i,u_i,w_i)=\MI(u,w)_i$. 
\end{proposition}
We remark that the division by 2 in the definition of $\recover$ is the reason why we need that $\Gamma$ does not have characteristic 2.
Also, note that $\tr_\alpha(u)_{i,\thash}$ is not used in $\recover$, but is used later to check soundness.
In standard binary representation of integers, all arithmetic operations over $\mathbb{F}_p$ that are necessary to
compute trace or its restoration at a single coordinate can be computed in time $O(\log^{O(1)} p)$.


We are now ready to define the class of \emph{superposition sketches} for functions from $D$ to $\Gamma$.

\begin{definition}[Superposition sketch]
\label{def:superposition}
Let $\suprange$ be a set, $h:D \rightarrow \suprange$ and $\alpha \in \Gamma$.
The \emph{superposition sketch induced by $(\alpha,h)$} is the  function $\tr_{\alpha,h}$ 
that maps $u \in \Gamma^D$ to $ \tr_{\alpha,h}(u) \in (\Gamma^4)^{\suprange}$ where for $j \in \suprange$:

\[
\tr_{\alpha,h}(u)_j=\sum_{i \in h^{-1}(j)} \tr_{\alpha}(u)_i.
\]

In words, the function $h$ is used to partition $D$ into $|\suprange|$ classes $h^{-1}(j)$, and  $\tr_{\alpha,h}(u)$  at $j \in \suprange$ is the sum of the trace vectors of $u$ corresponding to indices of $D$ in the class $h^{-1}(j)$.   The size (in bits) of the output is
 $O(|\suprange|\log |\Gamma|)$.
\end{definition}

Note that we can compute the superposition sketch of any $D$-sequence over $\Gamma$ easily: Initialize  $\tr_{\alpha,h}(u)$ to all zero
and then for each $i \in D$ add the trace vector $\tr_{\alpha}(u)_i$ to  $\tr_{\alpha,h}(u)_{j}$ where $j=h(i)$.

For $u,w \in \Gamma^D$ and $h:D \rightarrow \Gamma$, a mismatch index $i \in D$ is \emph{recoverable{} for $u,w,h$}  
if $h^{-1}(h(i)) \cap \indices(u,w) = \{ i \}$.
We now define a procedure $\RECOVER_{\alpha,h}$ that recovers all recoverable indices. 

\begin{algorithm}[H]
\label{algo:H-recover}
   \caption{$\RECOVER_{\alpha,h}(\tr_{\alpha,h}(u),\tr_{\alpha,h}(w), \alpha, h )$}\label{alg-compress}
   \KwIn{Traces $\tr_{\alpha,h}(u),\tr_{\alpha,h}(w)$ for two strings $u,w \in \Gamma^D$, trace parameters $\alpha\in \Gamma$, $h : D \rightarrow \suprange$.}
   \KwOut{The set $M_{\alpha,h}(u,w)$ of mismatch triples.}
   
   \vspace{1mm}
   \hrule\vspace{1mm}
    Let $\Delta_{\alpha,h}(u,w)=\tr_{\alpha,h}(u)-\tr_{\alpha,h}(w)$.   

    Let $J$ be the set of $j \in \suprange$ such that $\Delta_{\alpha,h}(u,w)_{j,\tval} \neq 0$.  

   \textbf{Rebuilding step:} For each $j \in J$ let $z_j=\recover(\Delta_{\alpha,h}(u,w)_j)$.
   
   \textbf{Filtering step:} Let $I = \{ j \in J,\;
z_{j,\dindex} <|D|$ $\wedge$ $\Delta_{\alpha,h}(u,w)_{j,\thash} =\alpha^{z_{j,\dindex}} (z_{j,\xval}-z_{j,\yval})\}$.

   Return $M_{\alpha,h}(u,w) = \{z_j:j \in I\}$.

\end{algorithm}

In the rebuilding step, the algorithm produces a list of mismatch triples by applying $\recover$  to every trace vector it can among
the $\Delta_{\alpha,h}(u,w)_j$.  In the filtering step, it eliminates some of these mismatch triples, and then it outputs the rest.  
The following lemma shows that (1) The rebuilding step produces all mismatch triples corresponding to recoverable indices (and possibly some others)
and (2) The filtering step  with high probability eliminates all mismatch triples corresponding to indices that are not recoverable, so Soundness holds with high probability.

\begin{lemma}
\label{lemma:MIh}
Let $u,w \in \Gamma^D$ and $h:D\to \suprange$ be fixed.
Let $I$ and $J$ be as in the $\RECOVER_{\alpha,h}(u,w)$. 
For each $j \in \suprange$:
\begin{enumerate}
\item If $|h^{-1}(j) \cap \indices(u,w)|= 0$ then $j \not\in J$.
\item If $|h^{-1}(j) \cap \indices(u,w)|= 1$  then $z_j \in M_{\alpha,h}(u,w)$ and $z_j=(i,u_i,w_i)$ where $i$ is the unique
mismatch index such that $h(i)=j$. 
\item The probability that $M_{\alpha,h}(u,w)$ outputs a triple that is not in $\MI(u,w)$ (i.e. that Soundness fails) is at most $\frac{(|D|-1)\cdot|\suprange|}{|\Gamma|}$ over a uniformly random choice of $\alpha \in \Gamma$. 
\end{enumerate}
\end{lemma}

\begin{proof}
For the first part, if $|h^{-1}(j) \cap \indices(u,w)|= 0$ then $\Delta_{\alpha,h}(u,w)_j$ is equal to $(0,0,0,0)$ and so $j \not\in J$.

For the second part, suppose $i$ is the unique mismatch index in $h^{-1}(j)$. Then $\Delta_{\alpha,h}(u,w)_j= \Delta_\alpha(u,w)_i$. 
Since $u_i \neq w_i$, $\Delta_{\alpha}(u,w)_{i,\tval} \neq 0$ and so $z_j=\recover(\Delta_{\alpha,h}(u,w)_i)=(i,u_i,w_i) \in \MI(u,w)$, by Proposition~\ref{prop:trace}.
This triple is not eliminated by the filtering step since $z_{j,\dindex}=i<|D|$ and 
$\Delta_{\alpha}(u,w)_{i,\thash}=(u_i-w_i)\cdot \alpha^i=(z_{j,\xval}-z_{j,\yval})\cdot \alpha^{z_{j,\dindex}}$. 

For the third part, we only need to look on $z_j$ where $|h^{-1}(j) \cap \indices(u,w)| > 1$
and show that they will be filtered out. 
Namely, we will show that $\prob_\alpha[z_j \in M_{\alpha,h}] \leq \frac{|D|-1}{|\Gamma|}$. Then 
a union bound over $j \in \suprange$ completes the proof.
If $\Delta_{\alpha,h}(u,w)_{j,\tval}=0$ then $j \not\in J$ irrespective of $\alpha$.  
Otherwise the filtering step excludes $z_j$ unless
$z_{j,\dindex}<|D|$ and
$\Delta_{\alpha,h}(u,w)_{j,\thash}-\alpha^{z_{j,\dindex}} (z_{j,\xval}-z_{j,\yval})=0$.  
Letting $I'$ be the set of mismatch indices in $h^{-1}(j)$,  
the left hand side of this equation equals $\sum_{i' \in I'} \alpha^{i'}(u_{i'}-w_{i'})-\alpha^{z_{j,\dindex}}(z_{j,\xval}-z_{j,\yval})$.  
Here, $z_{j,\dindex}, z_{j,\xval}, z_{j,\yval}$ do not depend on $\alpha$ so
for fixed $h,u,w$  this is a nonzero polynomial in $\alpha$ of degree at most $|D|-1$ (nonzero because the sum has at least two non-zero terms, and at most one is cancelled by the subtracted term). 
This polynomial has at most $|D|-1$ roots, so the probability that $\alpha$ is a root is at most $\frac{|D|-1}{|\Gamma|}$.
\end{proof}

The following result gives upper bounds on the running time of the sketch and recover algorithms, and on the space needed for the sketch
\begin{proposition}
\label{prop:superposition time}
Let $\Gamma=\mathbb{F}_p$ and let $D=\{0,\ldots,|D|-1\}$ with $|D| \leq p$.
Let $\suprange$ be a set, $h:D \rightarrow \suprange$ and $\alpha \in \Gamma$.  
The superposition sketch $\tr_{\alpha,h}$ maps a $D$-sequence $u$ over $\Gamma$
to a sketch of bit-length $O(|\suprange|\log |\Gamma|)$. 
The running time for the sketch algorithm is $O(|D| \cdot T)$
and the running time of the recover algorithm is $O(|\suprange| \cdot T)$, 
where $T$ is an upper bound on the time to perform a single arithmetic operation over $\Gamma$
and evaluate $h$ at a single point.

Furthermore if a $D$-sequence $u$ over $\Gamma$ is given via a sparse representation,
via $\{(j,u_j):j \in \supp(u)\}$ where $\supp(u)=\{j \in D:u_j \neq 0\}$
then the running time of the sketch algorithm is $O((|\suprange|+|\supp(u)| )\cdot T)$.
\end{proposition}



\subsection{Randomized superposition sketches}
\label{subsec:randomized superposition}

In order to apply the superposition sketch we need to select a good $h$.
However, one can hardly hope that if $\suprange$ is comparable in size to $\MI(u,w)$ then one can find a single $h : D \to \suprange$ 
for which all mismatch indices $\indices(u,w)$ will be recoverable.
Hence, we will try superposition sketches for \emph{multiple randomly chosen}  $h$'s.  
Fix a (small) family $H \subseteq \{ h: D \to \suprange\}$
and a probability distribution $\mu$ on $H$ (not necessarily uniform).
For $\beta \leq 1$, we say that $i$ is \emph{$\beta$-recoverable for $u,w,\mu$} 
provided that for $h \sim \mu$, the probability that $i$ is recoverable for $u,w,h$ is at least $\beta$. 
(Recall that if $i$ is recoverable for $u,w,h$ then for any choice of $\alpha$, the output of recovery procedure from the sketches $\tr_{\alpha,h}(u)$ and $\tr_{\alpha,h}(w)$ includes the triple $(i,u_i,w_i)$.)

We select hash functions $h_1,\ldots,h_\ell$ independently according to $\mu$, 
for some \emph{redundancy parameter} $\ell$.
We also select trace parameters $\alpha_1,\ldots,\alpha_\ell$ uniformly at random from $\Gamma$.  
The sketch of $u$ consists of the sequences $h_1,\ldots,h_\ell$ and $\alpha_1,\ldots,\alpha_\ell$ together with
$\tr_{\alpha_1,h_1}(u),\ldots,\tr_{\alpha_\ell,h_\ell}(u)$.  For the recovery algorithm, given the sketches for $u$ and $w$ we compute each of the sets
$M_{\alpha_i,h_i}(u,w)$ for $i \in [\ell]$ and define 
 $M_{\alpha_1,\ldots,\alpha_\ell,h_1,\ldots,h_\ell}(u,w)$ to be the set of triples that appear in strictly more than half of the sets.

We refer to a scheme of the above type as a \emph{randomized superposition scheme}.
Such a scheme is determined
by the distribution $(H,\mu)$ over hash functions and the redundancy $\ell$.

The size of the sketch in bits (not including the description of the hash functions used)
is $O(\ell \cdot |\suprange| \cdot \log|\Gamma|)$.    
The description of the hash functions depends on the method used to represent members of $H$. 
For standard explicit choices of $H$ (such as explicit families of $O(1)$-wise independent functions),
members of $h$  are represented in $O(\log|H|)$ space. 

\begin{proposition}
\label{prop:random superposition}
Let $D$ and $\suprange$ be sets and let $\Gamma$ be a field of size at least $4(|D|-1)|\cdot |\suprange|$.
For each pair of strings $u,w \in \Gamma^D$, suppose
$F(u,w)$ is a subset of $D$.
Let  $(H,\mu)$ be a distribution over hash functions from $D$ to $\suprange$.   
Suppose that for every $u,w$, every index belonging to
$F(u,w)$ is 3/4-recoverable for $u,w,\mu$.  
Then for any $\delta>0$,
the superposition sketch using $(H,\mu)$ with 
redundancy $\ell \geq 8 (\ln|D| + \ln \frac{1}{\delta}+ 2)$
satisfies the Completeness and Soundness conditions for
 $F$
with   failure probability at most $\delta$.
\end{proposition}

\begin{proof}
We show that  each of Soundness and Completness fail with probability at most $\delta/2$.

To bound the probability that Completeness fails, we claim that for  each 3/4-recoverable mismatch index $i \in D$, it is the case that for every fixed choice of $\alpha_1,\ldots,\alpha_\ell$,
the  probability (with respect to the choice of $h_1,\ldots,h_\ell$) that $M_{\alpha_1,\ldots,\alpha_\ell,h_1,\ldots,h_\ell}(u,w)$ does not include index $i$ is at most $\delta/(2D)$.
If this claim is true then summing over the at most $|D|$ mismatch indices that are 3/4-recoverable, the probability that the output of the scheme omits at least one such index is at most $\delta/2$.

So we verify the claim.
Fix $\alpha_1,\ldots,\alpha_\ell$ and a 3/4-recoverable mismatch index $i \in D$.
For $j \in \{1,\ldots, \ell\}$, let $Z_j$ be the 0-1 indicator of the event that $M_{\alpha_j,h_j}$ 
includes index $j$.
Since $j$ is 3/4-recoverable, $Z_{j}$ is 1 with probability at least $3/4$.  
Let $Z=\sum_{j=1}^\ell Z_{j}$. 
The expectation of $Z$ is at least $3\ell/4$.  
Index $i$ is not output by the algorithm if $Z \leq \ell/2$.  By the Chernoff-Hoeffding bound 
(Lemma~\ref{lemma:CH}), the probability  that $Z \leq \ell/2$ is at most $2e^{-\ell/8} \leq e^{-\ln|D|-\ln(1/\delta)-1} \leq \frac{\delta}{2D}$, as required to prove the claim.

Next we bound the probability that Soundness fails.
For each $j \in \{1,\ldots, \ell\}$, let $Y_j$ be the indicator of the event that 
$M_{\alpha_j,h_j}(u,w)$ contains at least one triple that is not in $\MI(u,w)$.  
By part 3 of Lemma~\ref{lemma:MIh},
for any fixed value of $h_j$, the probability, with respect to $\alpha_j$
that  $M_{\alpha_j,h_j}$ contains any triples
that do not belong to $\MI(u,w)$ is at most $\frac{(|D|-1)|\suprange|}{|\Gamma|}$ which is at most $1/4$ by hypothesis.  
Let $Y = \sum_j Y_j$.
If $Y \leq \ell/2$ then there is no incorrect triple that is in the output of more than $\ell/2$ out of the $\ell$ runs, and so it suffices to
show that $\prob[Y>\ell/2] \leq \frac{\delta}{2}$.  
We have $\expected[Y]=\sum_j\expected[Y_j] \leq \ell/4$. 
By the Chernoff-Hoeffding bound, the probability that the total number of such triples exceeds $\ell/2$ is  at most $2e^{-\ell/8} \leq \frac{\delta}{2}$, as required.
\end{proof}

As an example we give a simple application of random superposition schemes that recover all mismatch triples (with high probability) whenever $\HAM(u,w) \leq K$.
As mentioned in the introduction better sketches already exist.

\begin{corollary}
\label{cor:simple hamming}
For any $C>1$, there is a randomized  superposition scheme that
given alphabet $\Gamma$ of size at most $n$, domain $D$, $|D|\le n$, and positive integer $K$, for any string from $\Gamma^D$ produces a sketch of size $O(K \log^2 n + \log^2 n )$, and
on input of two sketches for strings $u,w \in \Gamma^D$ with $\HAM(u,w) \leq K$ recovers all mismatches with probability $\ge 1-1/n^{C}$.
\end{corollary}

\begin{proof}
Choose $\suprange$ to be a set of size between $4K$ and $8K$.
Let $H$ be a pairwise independent family of hash functions from $D$ to $\suprange$.  By Proposition~\ref{prop:pairwise}
we can describe  an $h \in H$ by $O(\log(n))$ bits, and the superposition sketch induced by $h$ requires $O(K \log n)$ bits.

We claim that for any $u,w \in \Sigma^D$ such that $\HAM(u,w)\leq K$ every mismatch index is $3/4$-recoverable for $u,w,\mu$.
For $h \sim \mu$, the probability that $i$ is not recoverable for $u,w,h$ is as most $\sum_{i'} \prob[h(i)=h(i')]$
where the sum is over $i' \neq i$ such that $i'$ is a mismatch index.  By pairwise independence,
$\prob[h(i)=h(i')] =1/|\suprange|$ and so the probability that $i$ is not recoverable for $u,w,h$ is at most $K/|\suprange|\leq 1/4$.

The result now follows immediately from
 Proposition~\ref{prop:random superposition}.
\end{proof}

\subsection{Proof of Theorem~\ref{thm:HMR}} 
\label{subsec:hmr proof}

In this section we show how to apply the randomized superposition schemes of Section~\ref{subsec:randomized superposition} to construct
a sketch-and-recover scheme for the hierarchical mismatch recovery and prove Theorem~\ref{thm:HMR}.
The reader should review the set-up for hierarchical mismatch recovery in Section~\ref{subsec:HMR}.

Recall that the scheme depends on the domain $\LOS_1 \times \cdots \times \LOS_d$,
the capacity function $\treecap$ where $\treecap=(\treecap_0,\dots,\treecap_d)$,
for each $j$, $\treecap_j \ge \treecap_{j+1}$ and an overload parameter $\loadpar$ which we fix to be $4d$. We view $\LOS_1 \times \cdots \times \LOS_d$ as the set of leaves of $T=T(\LOS_1 \times \cdots \times \LOS_d)$.
In preparation for describing the sketch-and-recover scheme, each node  $v$ at level $j<d$ of $T$ is
associated with 
a set of \emph{buckets} which are ordered pairs $(v,i)$ where $1 \leq i \leq \treecap_j$.   
We refer to $(v,i)$ as a $v$-bucket, and a bucket of the root is a \emph{root-bucket}.   
Each leaf $\ell$ has only one bucket, $(\ell,1)$.

The hash functions of our superposition scheme are \leafroot{} functions, which
are functions that map the set of leaves 
$\LOS_1\times \cdots \times \LOS_d$ to the set of root-buckets. We will use Proposition~\ref{prop:random superposition} to prove that the scheme works by:

\begin{enumerate}
\item Describing a distribution $\mu$ over \leafroot{} functions.
\item Showing that for any pair of strings $u,w$, every leaf that belongs
to $F_{T(\LOS_1 \times \cdots \times \LOS_d),\treecap,\loadpar}(u,w)$ is $3/4$-recoverable for  $u,w,\mu$.
\end{enumerate}

To describe the distribution $\mu$ on \leafroot{} functions, we consider a specific representation of a \leafroot{} functions, and for this we need the notions
of a \emph{trajectory} and \emph{routing functions}.

For a leaf $\ell \in \prod_{j=1}^d \LOS_j$, the modes on the path from $\ell$ to the root $\emptystr$ in $T$ are the nodes identified by the strings $\ell_{\leq d}, \ell_{\leq d-1}, \ldots, \ell_0$.
We define a \emph{trajectory for $\ell$} to be a sequence of buckets one for each node on the path from $\ell$ to $\emptystr$,
$(\ell_{\leq d},i_d),(\ell_{\leq d-1},i_{d-1}),\ldots,(\ell_0,i_0)$ where $i_j \in \{1,\ldots,\treecap_j\}$.
A trajectory is uniquely  determined by the leaf $\ell$ and the sequence of indices $(i_{d},i_{d-1},\ldots,i_0)$.
Note that $i_d$ must equal 1. 

We want a way to specify a trajectory for every leaf.  
We do this using a collection $r=(r_j : j < d)$ of \emph{routing functions}, one for each internal level of the tree.
The routing function $r_j$ is a function from $\LOS_{j+1} \times \{1,\ldots,\treecap_{j+1}\}$ to $\{1,\ldots,\treecap_j\}$.  
For each $v$ at level $j$ $r_j$ is used to specify a function
that maps buckets corresponding to children of $v$ to buckets of $v$ as follows:
for $a \in \LOS_{j+1}$ and $i \in \{1,\ldots,\treecap_{j+1}\}$, the bucket $(v \concat a,i)$ is mapped to $(v,r_j(a,i))$. 
Thus the collection of routing functions determines a trajectory for every leaf $\ell$ with sequence of indices $i_d(\ell)=1$
and for $j<d$, $i_j(\ell)=r_j(\ell_{j+1},i_{j+1}(\ell))$.  
This induces the \leafroot{} mapping that maps each $\ell\in \prod_{j=1}^d \LOS_j$ to the bucket $(\emptystr,i_0(\ell))$.

We are now ready to specify the distribution $\mu$.
For each  level $0 \leq j < d$, let  $H_j=\{ h: \LOS_{j+1} \times \{1,\ldots,\treecap_{j+1}\} \to \{1,\ldots,\treecap_j\}\}$ be a pairwise independent family of routing functions for level $j$.
Independently select $r_0,\ldots,r_{d-1}$ from $H_1,\ldots,H_{d-1}$.
The distribution $\mu$ on \leafroot{} functions is the distribution induced by the selection
of $r_0,\ldots,r_{d-1}$. 
Since $|\LOS_{j+1}|$ and all values of $\treecap_j$ are powers of 2,  
by Proposition~\ref{prop:pairwise}, we can choose $H_j$ so that
members of $H_j$ can be indexed with
 $O(\log|\LOS_{j+1}|+\log \treecap_{j+1} +\log \treecap_j)$ bits. 
Thus the total number of bits to represent a \leafroot{} function in the family is $O(\sum_{j=1}^d\log|\LOS_j|+ d\log \treecap_0))$.

\begin{lemma}
\label{lemma:hierarchical}
For $T, \LOS_1,\ldots,\LOS_d$ and $\treecap$ as in Theorem~\ref{thm:HMR},  
let $\mu$ be the distribution on \leafroot{} functions 
induced by choosing routing functions $r_0,\ldots,r_{d-1}$ independently from pairwise independent distribution.  
For any two strings $u,w$ in $\Gamma^{\LOS_1\times \cdots \times \LOS_d}$, 
every $\ell \in F_{T,\treecap,4d}(u,w)$, i.e., every $4d$-accessible mismatch leaf,  is $3/4$-recoverable for  $u,w,\mu$.  
\end{lemma}

\begin{proof}
Let $\ell \in \LOS_1 \times \cdots \times \LOS_d$ be a $4d$-accessible leaf with respect to $u,w$ where $u_\ell \neq w_\ell$. 
We must show that $\ell$ is $3/4$-recoverable. 
Recall that $\ell$ is $4d$-accessible if each node along the path from $\ell$ to root is $4d$-underloaded, 
i.e., for each $j<d$, $\hmrload{\treecap}_{\ell_{\le j}} < \treecap_j/4d$. 
Let $(r_j:0 \leq j <d)$ be the sequence of random routing functions selected as above and let $f$ be the induced \leafroot{} function.

Let $(i_d,\ldots, i_0)$ denote the sequence of indices $(i_d(\ell),\ldots,i_0(\ell))$ for the trajectory of $\ell$. 
This is a random variable depending on the choice of $r_0,\ldots,r_{d-1}$.
By definition, $\ell$ is not recoverable if and only if there is a mismatch leaf $\ell' \neq \ell$ such that
$f(\ell)=f(\ell')$.    
If $\ell' \neq \ell$ is a leaf such that $f(\ell')=f(\ell)$ then 
the trajectories of $\ell$ and $\ell'$ have non-empty intersection.  
We say that $\ell$ and $\ell'$ \emph{merge at level $j$}
if they are in different buckets at level $j+1$, but in the same bucket at level $j$.   
(Note that once the trajectories merge, they remain the same all the way to the root.)

For $j \in \{0,\ldots,d-1\}$, let $\MERGE_j$ be the event
that there is a mismatch leaf $\ell' \neq \ell$ that merges with $\ell$ at level $j$.  We now fix $j$ and prove that $\prob[\MERGE_j] \leq 1/4d$.  
This will finish the proof, since summing over all the levels, we will get that
the probability that $\ell$ is not recoverable is at most $1/4$.

We condition on the $r_{j+1},\ldots,r_{d-1}$, which determines the trajectory of all leaves up to level $j+1$.  In particular this
determines $i_{j+1},\ldots,i_d$,  

Consider the set of child buckets of $\ell_{\leq j}$.  
These have the form $(\ell_{\leq j} \concat a,i)$ where $(a,i) \in \LOS_{j+1} \times \{1,\ldots,\treecap_{j+1}\}$.  
This includes the bucket $(\ell_{\leq j}\concat \ell_{j+1},i_{j+1})$ 
on the trajectory of $\ell$. 

 Say that a bucket $(v,i)$ is \emph{occupied} if it lies on the trajectory of some
mismatch leaf.
Let $\OCC$ be the set of pairs $(a,i) \neq (\ell_{j+1},i_{j+1})$ such
that $(\ell_{\leq j} \concat a, i)$ is occupied.  The event $\MERGE_j$ is equivalent to
the event that there is an $(a,i) \in \OCC$ such that $r_j(a,i)=r_j(\ell_{j+1},i_{j+1})$.
For each $(a,i) \in \OCC$, $\prob[r_j(a,i)=r_j(\ell_{j+1},i_{j+1})]=\frac{1}{\treecap_j}$ since
$r_j$ is a pairwise independent map whose range has size $\treecap_j$, and so the conditional
probability of $\MERGE_j$ given $r_{d-1},\ldots,r_{j+1}$ is at most $|\OCC|/\treecap_j$.

We need to upper bound $|\OCC|$.  
Let $\occupied(v)$ denote the number of occupied $v$-buckets.  
In the above analysis $|\OCC|=\sum_{a\in \LOS_{j+1}} \occupied(\ell_{\leq j} \concat a)-1$.
We claim:

\begin{proposition}
\label{prop:occupied}
For any choice of routing functions $r_{d-1},\dots,r_0$:
\begin{enumerate}
\item  For any internal node $v$ at level $j'<d$, $\occupied(v) \leq \min(\treecap_{j'},\sum_{v' \in \child(v)} \occupied(v'))$.
\item  For any node $v$ at level $j'\leq d$, $\occupied(v) \leq \hmrload{\treecap}(v)$.
\end{enumerate} 
\end{proposition}

\begin{proof}
For the first part, let $v$ be an internal node. 
Then $\occupied(v)$ is trivially at most $\treecap_{j'}$.  
Also, a $v$-bucket is occupied if and only if some occupied child maps to it, so $\occupied(v) \leq \sum_{v' \in \child(v)} \occupied(v')$.

For the second part, if $v$ is a leaf then
$\occupied(v)=1$ if $v$ is a mismatch leaf and 0 otherwise, so $\occupied(v)=\hmrload{\treecap}(v)$.
If $v$ is an internal node, the first part implies $\occupied(v) \leq \min(\treecap_{j'},\sum_{v' \in \child(v)} \occupied(v'))$
and applying induction and the definition of $\hmrload{\treecap}(v)$ we have
 $\min(\treecap_{j'},\sum_{v' \in \child(v)} \hmrload{\treecap}(v'))=\hmrload{\treecap}(v)$.
\end{proof}

Thus $|\OCC| \le \sum_{a\in \LOS_{j+1}} \hmrload{\treecap}(\ell_{\leq j} \concat a)-1$.
We know that $\hmrload{\treecap}(\ell_{\leq j}) \le \treecap_{j} / 4d$ and in particular, $\hmrload{\treecap}(\ell_{\leq j}) < \treecap_{j}$.
Hence, $\hmrload{\treecap}(\ell_{\leq j}) = \sum_{a\in \LOS_{j+1}} \hmrload{\treecap}(\ell_{\leq j} \concat a) > |\OCC|$.
It follows that for any choice of the routing functions $r_{d-1},\ldots,r_{j+1}$, $|\OCC|  < \treecap_j/4d$.  
Therefore $\prob[\MERGE_j] \leq \frac{1}{4d}$, as required to complete the proof of the theorem.
\end{proof}

We are ready to conclude Theorem~\ref{thm:HMR}.
Let us define the sketching function $\hmrsketch(u; T,\treecap,\delta)$ for targeted mismatch recovery $\HMR(T,\treecap,4d)$ 
to be the superposition sketching function on the tree $T$ with the capacity function $\treecap$, 
where the distribution on \leafroot{} functions is as given in the above lemma, and with the  
redundancy set to $\lceil 8(\sum_{j=1}^d\ln|\LOS_j| +\log(1/\delta)+2)\rceil$.
Let us define $\hmrrecover$ to be the associated recovery function.
To conclude the correctness of the scheme (the first item of Theorem~\ref{thm:HMR}) we apply Proposition~\ref{prop:random superposition} 
together with Lemma~\ref{lemma:hierarchical} with parameters set as follows: 
$D=\prod_{j=1}^d |\LOS_j|$,  $\suprange=\{1,\dots,\treecap_0\}$, redundancy $\ell = \lceil 8(\sum_{j=1}^d \ln |\LOS_j|+\log(1/\delta)+2)\rceil$,
and $(H,\mu)$ as defined above for $(T,\treecap)$.
(Notice, the theorem hypothesis that $|\Gamma| \ge 4 \prod_{j=1}^d |\LOS_j|^2$ implies that $|\Gamma| \ge 4(|D|-1)\cdot |\suprange|$ as required by Proposition~\ref{prop:random superposition}.)
The sketch consists of $O(\ell \cdot |\suprange|)$ elements from $\Gamma$ so it takes $O(\treecap_0 \cdot \log |\Gamma|\cdot (\log\prod_{i=1}^d|\LOS_j| + \log(1/\delta)))$ bits.
Evaluating a hash function from $H$ at a single point takes time $O(d \log^{O(1)} |\Gamma|)$
so by Proposition~\ref{prop:superposition time}, the sketching algorithm runs in time $O(\ell \cdot \prod_{i=1}^d|\LOS_j| \cdot d \log^{O(1)} |\Gamma|)=O(\prod_{i=1}^d|\LOS_j| \cdot \log^{O(1)} |\Gamma| \cdot \log(1/\delta))$.
If $u$ is given via its sparse representation then the time to construct the sketch
is $O((\treecap_0 + |\supp(u)|) \cdot \log^{O(1)} |\Gamma| \cdot \log(1/\delta))$.
The recovery algorithm runs in time $O(\treecap_0 \cdot  \log^{O(1)} |\Gamma| \cdot \log(1/\delta))$ as required.
Finally, each mismatch pair that is output by the recovery algorithms must appear in more than half of the $\ell$ redundant sketches.
As each sketch outputs at most $|\suprange|$ elements, 
the number of mismatch pairs output by the algorithm is at most $\treecap_0$.




\section{Sketch-and-recover scheme for edit distance}
\label{sec:ED scheme}

In this section we describe our sketch-and-recover scheme for edit distance. 

This scheme  takes  a \emph{distance parameter $k$} and size parameter $n$ 
and has the property that
given strings $x,y$ of length at most $n$, with probability at least $2/3$,  if $\ED(x,y) \leq k$
the recover algorithm on sketches of $x$ and $y$ will output the set of costly edges in the canonical alignment of $x$ and $y$, 
as defined in Section~\ref{sec:edit distance}.  If $\ED(x,y) > k$ the recover algorithm
should output $\LRGE$.
The strings $x$ and $y$ are over the alphabet $\Sigma_n=\{1,\ldots,n^3\}$.
We make the simplifying assumption that $k\le n/4$ otherwise each string can serve as its own sketch. 


Let $\gsize=\log n$.
Our goal is that the length of the  sketch  should be ``nearly linear'' in $k$,
which means $k \times T$ where $T$ is ``small''. 
We achieve this
with $T=\orgap^2 \gsize^{O(1)}$, where $\orgap=2^{O(\sqrt{\gsize \log \gsize})}$ is the gap parameter from Theorem~\ref{thm:OR}, so this
is the dominant term.  Our construction is multi-layered and intricate but fairly clean.  We are careless with factors of $\gsize$ and
of constants in the sketch size; it is clearly possible to 
reduce some of these factors, but it would complicate our already involved construction.

An important note is that we use the gap threshold algorithm from Theorem~\ref{thm:OR} as a black box.  Improving
the value of $\orgap$ in Theorem~\ref{thm:OR} will translate into a reduction in the sketch size. 

The sketch and recover algorithms build recursive structures that are represented as labelled trees $T(\BASE^{\tdepth})$ or $T(\{0,1\}^{\tdepth \gsize})$.  
Both of these trees have a number of nodes that is superpolynomial in $n$, 
so explicitly evaluating the labeling at all nodes leads to superpolynomial running time. 
We can avoid this because each of these trees has the property that all but polynomially many nodes are labelled by a \emph{default value}, and the set of nodes not labelled by the default value is a rooted subtree.  
Thus we only have to explicitly compute the labeling at nodes whose value is not the default value.  
For example, a key step in building the sketch for $x$ is to build a decomposition for $x$ on the tree $T(\BASE^{\tdepth})$
where we assign $x$ to the root, and work our way down the tree: if node $v$ has been assigned $x_v$ then
the children of $v$ are assigned a (carefully chosen) decomposition of $x_v$.  
At most polynomially many nodes will be assigned a nonempty string (because there are at most $|x|$ nonempty substrings at each level of the tree) and the default value is the empty string $\emptystr$. 
As we build the tree, when we encounter a node $u$ that is assigned the empty string whose parent is assigned a nonempty string, 
we note that $x_u=\emptystr$ and implicitly assign  all descendants to $\emptystr$ without doing any computation on the descendants.  
The set of explicitly assigned nodes is a rooted subtree of polynomial size and any node not appearing in the subtree is implicitly 
assigned to $\emptystr$, and our algorithms only pay for the explicitly assigned values.

\subsection{Pseudo-code for edit distance sketch-and-recover algorithms}
\label{subsec:sketch code}

Our sketching algorithms and their analysis depend on various parameters. 
We review the parameters here. 
They are implicitly part of the input of all of our procedures.

The following parameters are set by the user of the scheme: $k$ is the edit distance bound that the scheme is designed to handle,
and $n$ is the upper bound on the size of the input strings.
We assume $n$ is a power of two and that it is larger than $40k\orgap$, where $\orgap$ is the distortion of the Ostrovsky-Rabani fingerprint.
We use $\gsize = \log n$ as a shorthand.

We recall the following parameters that were introduced in formulating the properties of the Ostrovsky-Rabani fingerprint and the function $\bdecomp$. 
They are used to define other parameters of our algorithm:
$\bvsize=\bvsizeval$ is the length of bit-vectors that represent grammars output by $\bdecomp$.  
$\dfactor=O(\log^ 2n)$ appears in Part~\ref{BK2b} of Theorem~\ref{thm:BK} concerning $\bdecomp$.
It bounds how much is the edit distance of two strings multiplied when we measure the Hamming distance of their representation by grammars.
$\bkfailure=O(\log^4 n)$ is the multiplier that appears in the failure probability in Part~\ref{BK2a} of Theorem~\ref{thm:BK}.
The failure probability is the ratio between the edit distance of two strings and the desired grammar sparsity multiplied by $\bkfailure$.
$\orgap=2^{O(\sqrt{\log(n)\log\log(n)})}$ is the gap parameter ({\em distortion}) from Theorem~\ref{thm:OR} concerning the Ostrovsky-Rabani fingerprints.  

Our algorithm uses some large field $\bigfield=\mathbb{F}_p$ to represent fingerprints and compute hierarchical mismatch recovery schemes. 
Here $p$ is a prime whose bit length is at least $4 (\log \bvsize \cdot \log n)$ which is sufficient for all those applications.  
We assume that the prime $p$ is selected during the initialization of our algorithm (see Section~\ref{sec:compconsider}).

The sketching algorithm for a string $x$ operates by implicitly building some tree data structures to represent $x$, which are then compressed.  
The trees underlying these data structures are fixed.
The decomposition tree produced by $\mdecomp$ is $T(\BASE^{\tdepth})$, 
where $\BASE=\{0,1\}^{\gsize}$ indexes its edges, 
and $\tdepth=\lceil \log (20k \cdot \orgap) \rceil$ is its depth.
For $i \in \{1,\ldots,\tdepth\}$, $T^j$ denotes the subtree $T(\BASE^{j})$ of $T(\BASE^{\tdepth})$.
$\mdecomp$ also builds some data structures on the binary tree $T(\{0,1\}^{\gsize\tdepth})$.  
The nodes at level $j$ of $T(\BASE^{\tdepth})$ correspond in the obvious way to the nodes at level $\gsize j$ in $T(\{0,1\}^{\gsize\tdepth})$.

The following parameters are chosen carefully for the algorithm and its analysis. 
They bound quantities at nodes of $T(\BASE^{\tdepth})$:
$\gapthresh_i$ is the threshold of Ostrovsky-Rabani fingerprint used for substrings at level $i$ of the tree,
$\gramsize_i$ is the sparsity of grammar bit-vectors at level $i$ of the tree,
$\stringcap_i$ is the capacity of nodes at level $i$ of the hierarchical mismatch recovery scheme on $T(\BASE^{\tdepth})$.
The quantities shrink by factor of 2 at each level and the latter two quantities can be defined in terms of $t_i$.
The parameters are chosen to satisfy various constraints that arise in the analysis 
and we summarize the constraints in Section~\ref{sec:constaints}.

\setlength{\tabcolsep}{10pt} 
\renewcommand{\arraystretch}{1.7} 
\begin{center}
\begin{tabular}{l|l|l}
 \makecell{Ostrovsky-Rabani\\ fingerprint threshold} 
&  \makecell{Grammar sparsity} 
& \makecell{Hierarchical mismatch \\ recovery capacity}   \\
 $\gapthresh_0 = \lceil 20k \cdot \orgap \rceil_2$ 
& $\gramsize_0=(32\bkfailure\log^6n)\cdot \gapthresh_0$ 
& $\stringcap_0=(512 \bkfailure \orgap \log^{12} n)\cdot \gapthresh_0 $  \\
 \makecell[l]{$\gapthresh_i = \gapthresh_{i-1}/2 $ \\  $\hphantom{\gapthresh_i} = \gapthresh_0/2^{d-i}$ } 
& \makecell[l]{$\gramsize_i=\gramsize_{i-1}/2$  \\ $\hphantom{\gramsize_i}=(32\bkfailure\log^6n)\cdot \gapthresh_i$ } 
& \makecell[l]{$\stringcap_i=\stringcap_{i-1} /2 $ \\ $\hphantom{\stringcap_i}=(512\bkfailure \orgap \log^{12} n)\cdot \gapthresh_i $ }
\end{tabular}
\end{center}

Here by $\lceil a \rceil_2$ we mean the smallest power of two larger than $a$. 
We set $\critpar=64\bkfailure \, \orgap \log^6 n$, which is used in Proposition~\ref{prop:loadestimator 2}.
For $j \in \{1,\ldots,\tdepth\}$, $\stringcap^{j}$ is the capacity function on $T^j$ where for $i \in \{1,\ldots,j\}$,
$\stringcap^{\leq j}_i=\stringcap_i$  and $\stringcap^{\leq j}_{j+1}=1$.
Additionally we need to define a capacity for a hierarchical mismatch recovery scheme on $T(\{0,1\}^{\gsize \tdepth})$.
For $j \in \{1,\ldots,\tdepth\gsize-1\}$, $\locationcap^{\leq j}$ is the capacity function on $T(\{0,1\}^j)$ where for $i \in \{1,\ldots,j\}$,
$\locationcap^{\leq j}_i=\stringcap_{\lfloor i/\gsize \rfloor}$.

The main component of our sketching algorithm is procedure $\mdecomp$. 
It takes the input string $x$ and recursively applies procedure $\bdecomp$ to it with sparsity parameter $\gramsize_j$, where $j$ is the recursion level.
Each call to procedure $\bdecomp$ splits a substring of $x$ into smaller substrings, each represented by a {\em grammar} of sparsity $\gramsize_j$.
We label the tree $T(\BASE^{\tdepth})$ by the iteratively obtained substrings of $x$.
The set $\nonemptynodes$ traces which nodes of the tree are non-empty.  

\begin{algorithm}[H]
\label{algo:mdecomp}
   \caption{$\mdecomp(x )$}
   \KwIn{String $x \in \Sigma_n^{\le n}$.}
   \KwOut{Tree decomposition $\ztree$ of $x$  indexed by $\BASE^{\le\tdepth}$ with associated tree $\tstree$ of Ostrovsky-Rabani fingerprints and grammars $\bvec$,  $\nonemptynodes$ which is the support of $\ztree$,
         refinement $\ztree^*$ of the tree decomposition of $x$ indexed by $\{0,1\}^{\le\tdepth \gsize}$  with the tree of fingerprints $\fptree$ and left child size tree $\lvec^*$, and $\nonemptynodes^*$ which is the support of $\ztree^*$.}
   
   \vspace{1mm}
   \hrule\vspace{1mm}

   Set $\ztree_{\emptystr} = x$; $\nonemptynodes = \{\emptystr\}$. 
   
   \For{$j = 0, \dots, \tdepth-1$}{

      \For{$v \in \nonemptynodes \cap \BASE^{j}$}{

          \If{$\ztree_{v}\neq \emptystr$}{

               Call $\bdecomp(\ztree_{v};\gramsize_{j+1})$ to obtain decomposition $(z_{i}:i \in \BASE)$ of $\ztree_v$  with associated grammars $(G_i:i \in \BASE)$ where $G_i \in \{0,1\}^{\bvsize}$.


               \For{each $i \in \BASE$}{
               
                   \If{$z_i \neq \emptystr$}{
                       $\ztree_{v \concat i} = z_i$;
                       
                       $\bvec_{v \concat i} = G_i$; 
                       
                       $\tstree_{v \concat i} = \threshsk(z_i;\gapthresh_{j+1})$;
                       
                       $\nonemptynodes = \nonemptynodes \cup \{v \concat i\}$.
                    }
                }
          }

      }
      
    }

    $\nonemptynodes^* = \nonemptynodes \cap \{0,1\}^{\tdepth  \gsize}$;

    \lFor{each $v \in \nonemptynodes^*$}{$\ztree^*_v=\ztree_v$}

    \For{$j= \tdepth \gsize -1, \dots, 0$}{

       \For{each $v \in \{0,1\}^j$ where $v\concat 0$ or $v\concat 1 \in \nonemptynodes^*$}{

           $\ztree^*_v = \ztree^*_{v \concat 0} \concat \ztree^*_{v \concat 1}$;

           $\fptree_v = \fingerprint(\ztree^*_v)$;
           
           $\lvec^*_v = |\ztree^*_{v \concat 0}|$;

           $\nonemptynodes^* = \nonemptynodes^* \cup \{v\}$.
       }  
    }

    Output $\ztree$, $\bvec$, $\tstree$, $\nonemptynodes$, $\ztree^*$, $\lvec^*$, $\fptree$, $\nonemptynodes^*$.

\end{algorithm}

Thus we build two vertex labelings of the tree $T(\BASE^{\tdepth})$:  
$\ztree$ is the labeling by substrings of $x$ that is a decomposition tree for $x$, 
and $\bvec$ labels each node, except the root, by a bit-vector that represents a grammar that encodes $\ztree_v$.
All of the labelings are represented via a sparse representation.  
The algorithm constructs the set $\nonemptynodes$ of nodes consisting of the root $\ztree_{\emptystr}$ and all nodes $v$ for which $\ztree_{v} \neq \emptystr$. 
The labels $\ztree_v$ and $\bvec_v$ are assigned values explicitly only if $v \in \nonemptynodes$; 
for $v \not\in \nonemptynodes$, $\ztree_v$ and $\bvec_v$ are not explicitly assigned, and are implicitly set, respectively, to $\emptyset$ and the 0 vector.

Additionally, for each non-empty substring $\ztree_v$ we also calculate its Ostrovsky-Rabani fingerprint for a suitable threshold parameter.
The fingerprints form another labeling $\tstree$ of  $T(\BASE^{\tdepth})$.
The fingerprints are interpreted as elements from $\bigfield$.

In the second part of $\mdecomp$ we build a refinement $\ztree^*$ of $\ztree$. 
$\ztree^*$ is a vertex labeling of the binary tree $T(\{0,1\}^{\tdepth \gsize})$
where $\ztree^*_v$ is the concatenation of $\ztree_u$ for $u \in \BASE^{\tdepth}$ that have $v$ as its prefix.
(Here we view $u$ as a binary string of length $\tdepth \gsize$.) 
This can be efficiently built bottom-up starting from leaves.
For each $\ztree^*_v$ we calculate its Karp-Rabin fingerprint $\fptree_v$
and we set $\lvec^*_{v}$  to be the length of $\ztree^*_{v \concat 0}$, the left child of $v$.
Hence,  $\fptree$ and $\lvec^*$ are both labelings of  $T(\{0,1\}^{\tdepth \gsize})$.
The labeling $\ztree$ of $T(\BASE^{\tdepth})$ corresponds to a partial labeling of the binary tree $T(\{0,1\}^{\tdepth \gsize})$ where only nodes at levels that are multiples of $\gsize$ are labeled.  
$\ztree^*$ "interpolates" this partial labeling to a full labeling of the binary tree.  
As with $\ztree$, $\ztree^*_v$ is only explicitly defined for nodes where it is nonempty.
The set $\nonemptynodes^*$ traces the non-empty nodes $v$ of $\ztree^*$.
For all other $v$, $\ztree^*_v$ is implicitly equal to $\emptystr$ and  $\lvec^*_{v}$ is implicitly set to 0. 



After computing the decomposition of a string $x$, the decomposition is {\em condensed} into hierarchical mismatch recovery sketch
in procedures $\vectorcondense$ and $\locationcondense$. 
The former is used to sketch the grammar tree $\bvec$ and the latter is used to sketch the tree of child-sizes $\lvec^*$. 
The first sketch will allow us to recover grammars of differing substrings between $x$ and $y$, 
the second will be used to determine the exact position of the differing substrings in $x$ and $y$.
The $\vectorcondense$ takes each of the layers of $\bvec$, and applies to it $\hmrsketch$ with a suitable choice of parameters.
Each layer of $\bvec$ is a sequence of grammars represented by a binary vector.
Each of the binary vectors is multiplied by a corresponding Ostrovsky-Rabani fingerprint so that 
if two corresponding nodes in the decomposition of $x$ and $y$ are labelled by substrings of $x$ and $y$ of large edit distance 
then the multiplied grammar vectors will differ in each non-zero entry.
Those non-zero positions will be recovered from the hierarchical sketch if the path to the root in the tree of grammars is not overloaded.
The collection of the hierarchical sketches, one for each layer of $T(\BASE^{\tdepth})$, is the sketch output by $\vectorcondense$.

\begin{algorithm}[H]
\label{algo:vectorcondese}
   \caption{\vectorcondense$(\bvec,\tstree,\nonemptynodes)$}
   \KwIn{Tree of grammars $\bvec$ indexed by $\BASE^{\le\tdepth}$ with support $\nonemptynodes$, and tree $\tstree$ of associated Ostrovsky-Rabani fingerprints.}
   \KwOut{$\HMR$ sketch for each level of the $\bvec$.}
   
   \vspace{1mm}
   \hrule\vspace{1mm}

   \For{each $j \in \{1,\ldots,\tdepth\}$}{

       Let $\bvec^j$ be the all-zero $D$-sequence over $\bigfield$ where $D=\BASE^j \times \{1,\ldots,\bvsize\}$

       \For{$v \in \nonemptynodes \cap \BASE^j$ and $i \in \{1,\dots, \bvsize\}$}{
           $\bvec^j_{v\concat i} = \tstree_v \times_\bigfield \bvec_{v,i}$.
       }

       $\fsketch^j = \hmrsketch(\bvec^j;T(\BASE^j \times \{1,\dots,\bvsize\}),\stringcap^{\le j},4(\tdepth+1),1/n^4)$.
   
   }

   Output $\fsketch=(\fsketch^1,\ldots,\fsketch^{\tdepth})$.

\end{algorithm}

Similarly, each of the $\gsize \tdepth$ layers of $\lvec^*$ consists of integers from $\{1,\dots,n\}$.
Each of the integers is offset by a corresponding Karp-Rabin fingerprint from $\fptree$ so that 
if the fingerprints in $x$ and $y$ differ the corresponding values will differ.
(Here, we assume that the field $\bigfield$ is larger than the maximum value of the fingerprint times $n$ so the mapping is invertible.)   
Again,  hierarchical mismatch recovery sketch of each fingerprinted layer with suitable parameters allows to recover sizes of left substrings for differing pairs of corresponding nodes in $\ztree^*$ for $x$ and $y$.
Summing-up the sizes of the left subtrees along a path from some node in $T(\{0,1\}^{\tdepth \gsize})$ to the root allows to recover
the position of a differing substring in the decomposition of $x$ and $y$, respectively.
So the collection of the hierarchical sketches, one for each layer of $\lvec^*$, is the sketch output by $\locationcondense$.

\begin{algorithm}[H]
\label{algo:locationcondense}
   \caption{\locationcondense$(\lvec^*,\fptree,\nonemptynodes^*)$}
   \KwIn{Tree of Karp-Rabin fingerprints $\fptree$ for the refined tree decomposition $\ztree^*$ of $x$ indexed by $\{0,1\}^{\le\tdepth \gsize}$ with left child size tree $\lvec^*$ of $\ztree^*$ supported on $\nonemptynodes^*$.}
   \KwOut{$\HMR$ sketch for each level of the $\lvec^*$.}
   
   \vspace{1mm}
   \hrule\vspace{1mm}


   \For{each $j \in \{0,\ldots,\tdepth\gsize -1\}$}{

       Let $\lvec^j$ be the all-zero $D$-sequence over $\bigfield$ where $D=\{0,1\}^j$.

       \For{$v \in \nonemptynodes^* \cap \{0,1\}^j$}{
           $\lvec^j_v = n \times_\bigfield \fptree_v+_\bigfield \lvec^*_v$.
       }

       $\gsketch^j = \hmrsketch(\lvec^j;T(\{0,1\}^j),\locationcap^{\leq j}, 4 \tdepth \gsize,1/n^4)$.   
   }

   Output $\gsketch=(\gsketch^0,\ldots,\gsketch^{\tdepth\gsize-1})$.

\end{algorithm}

Procedure $\msketch$ builds a single instance of a sketch for $x$ 
using a single call to each of the procedures $\mdecomp$, $\vectorcondense$, $\locationcondense$.
$\fullsketch$ applies the procedure $\msketch$  $(\numrep)$-times using independent randomness to get a sequence of sketches for $x$.
Their concatenation is the sketch for $x$.

\begin{algorithm}[H]
\label{algo:msketch}
   \caption{\msketch$(x)$}
   \KwIn{String $x \in \Sigma_n^{\le n}$.}
   \KwOut{A sketch $\msketchout$ of $x$.}
   
   \vspace{1mm}
   \hrule\vspace{1mm}

    Invoke $\mdecomp(x)$ which outputs $\ztree$, $\bvec$, $\tstree$, $\nonemptynodes$, $\ztree^*$, $\lvec^*$, $\fptree$,  $\nonemptynodes^*$.
 
    $\fsketch = \vectorcondense(\bvec, \tstree, \nonemptynodes)$.
    
    $\gsketch = \locationcondense(\lvec^*, \fptree, \nonemptynodes^*)$.
    
    Output $\msketchout=(\fsketch,\gsketch)$.

\end{algorithm}

\begin{algorithm}[H]
\label{algo:fullsketch}
   \caption{\fullsketch$(x,k,n)$}
   \KwIn{String $x \in \Sigma_n^{\le n}$, integer parameters $k$ and $n$.}
   \KwOut{A sketch $\fullsketchout$ of $x$.}
   
   \vspace{1mm}
   \hrule\vspace{1mm}

   Determine $\gapthresh_1,\dots, \gapthresh_{\tdepth}, \stringcap^1,\dots,\stringcap^{\tdepth}, \locationcap^1,\dots,\locationcap^{\tdepth \gsize}$.

   $\threshskout = \threshsk(x;20k\orgap)$.

   \lFor{$i \in \{1,\ldots, \numrep\}$}{ $\msketchout_i = \msketch(x)$} 
   
    Output $\fullsketchout = (\threshskout, \msketchout_1,\ldots,\msketchout_{\numrep})$. 
    
\end{algorithm}

   




 
    


    

\subsection{Pseudo-code for the recovery algorithm}
\label{subsec:recovery code}

The recovery of edit distance from sketches of $x$ and $y$ starts in procedures $\findstrings$ and $\findlocations$.
$\findstrings$ takes two sketches $\fsketch(x)$ and $\fsketch(y)$ that were produced by $\vectorcondense$.
$\fsketch(x)$ consists of $d$ hierarchical mismatch recovery sketches, each of them sketches a level of the grammar tree of $x$.
Similarly for $\fsketch(y)$.
$\findstrings$ calls the recovery procedure $\hmrrecover$ on each corresponding pair of hierarchical sketches.
The call to $\hmrrecover$ provides a list $\stringmismatch$ of mismatch triples.
Each mismatch triple $((v,i),\alpha, \beta)$ asserts that the $i$-th bit of the grammar at node $v$ of the decomposition tree $\bvec$ for $x$
is $1$ if $\alpha$ is non-zero, and the same bit of a grammar for $y$ is $1$ if $\beta$ is non-zero.
Grammars for which we get some mismatch triple are decoded by $\bdecomp$.

\begin{algorithm}[H]
\label{algo:findstrings}
   \caption{$\findstrings(\fsketch(x),\fsketch(y))$}
   \KwIn{$\hmrsketch$ sketches $\fsketch(x)$ and $\fsketch(y)$ of grammars for tree decompositions of $x$ and $y$.}
   \KwOut{Sequences $\claimedstring_v(x)$ and $\claimedstring_v(y)$ indexed by some $\stringnodes \subseteq \BASE^{\le \tdepth}$ of substrings on which $x$ and $y$ differ.}

   \vspace{1mm}
   \hrule\vspace{1mm}

   Set $\stringnodes =\emptyset$.
 
   \For{$j \in \{1,\ldots,\tdepth\}$}{

        $\stringmismatch^j =  \hmrrecover(\fsketch^j(x),\fsketch^j(y))$.  
        

        $\stringnodes^j=\{ v \in \BASE^j;\; ((v,i),\alpha, \beta) \in \stringmismatch^j$ for some $i,\alpha,\beta \}$.

        \For{$v \in \stringnodes^j$}{

            $\claimedbv_v(x) = \claimedbv_v(y) = 0^{\bvsize}$.


            \For{each $((v',i),\alpha,\beta)\in \stringmismatch^j$ where $v'=v$}{

                \lIf{$\alpha \neq 0$}{$\claimedbv_{v,i}(x) = 1$}

                \lIf{$\beta \neq 0$}{$\claimedbv_{v,i}(y) = 1$}
                
            }


            $\claimedstring_v(x) = \decode(\claimedbv_v(x))$. 

            $\claimedstring_v(y) = \decode(\claimedbv_v(x))$. 

            \If{$\claimedstring_v(x)$ and $\claimedstring_v(y)$ are both defined}{Add $v$ to $\stringnodes$.}
        
        }

    }
        
    Output $(\claimedstring_v(x):v \in \stringnodes)$ and $(\claimedstring_v(y):v \in \stringnodes)$.

\end{algorithm}

If all non-zero entries of the grammar were recovered then $\bdecomp$ returns the corresponding string otherwise it returns $\undefnd$.
(The latter case occurs if only a subset of ones was recovered for the given grammar vector.)
This provides a collection of substrings of $x$ from its $\ztree$, and similarly for $y$.
We organize the substrings as a partial labeling of $T(\BASE^{\tdepth})$, for which $\stringnodes$ identifies the support of the labeling.
The two partial labelings for $x$ and $y$ are the output of the procedure.

Similarly, $\findlocations$ gets two sketches $\gsketch(x)$ and $\gsketch(y)$ that were produced by $\locationcondense$,
and a set of target nodes $\stringnodes$ from the tree decomposition of $x$ and $y$. 
$\gsketch(x)$ consists of $\gsize \tdepth-1$ hierarchical mismatch recovery sketches, each of them sketches a level of the $\lvec^*$ tree of $x$.
Similarly for $y$.
On each pair of sketches for $x$ and $y$ we apply $\hmrrecover$ to identify {\em differences} between the values at a given level of $\lvec^*$ tree for $x$ and $y$.
(A difference at node $v$ might come from two sources: either the actual values $\lvec^*_v$ for $x$ and $y$ differ or the associated
$\ztree^*_v$ differ.)
For each node $v$ where the recovery procedure identifies a {\em difference} we recover the value $\lvec^*_v$ for $x$ and $y$, respectively. 
After we find all recoverable {\em differences} between $\lvec^*$ for $x$ and $y$,
for each node $v \in \stringnodes$ where we identified a {\em difference}, we attempt to calculate the starting position in $x$ and $y$, resp., 
of the substring $\ztree_v$ of $x$ and $y$.
The starting position of $\ztree_v$ is given by the sum of $\lvec^*_u$ over all nodes $u$ in $T(\{0,1\}^{\tdepth \gsize})$ which we reach from right
on the path from $v$ to root.
(If any of the $\lvec^*_u$ does not have its value defined then the starting position for $v$ will remain undefined.)
We return two sequences $\claimedlocation_v(x)$ and $\claimedlocation_v(y)$ of starting positions we calculated.

\begin{algorithm}[H]
\label{algo:findlocations}
   \caption{$\findlocations(\gsketch(x),\gsketch(y),\stringnodes)$}
   \KwIn{$\hmrsketch$ sketches $\gsketch(x)$ and $\gsketch(y)$ of substring locations for refined tree decompositions of $x$ and $y$, $\stringnodes$ are nodes of the decomposition for which we want to recover their starting position.}
   \KwOut{Sequences $\claimedlocation_v(x)$ and $\claimedlocation_v(y)$ indexed by some $\claimednodes \subseteq \stringnodes$ of starting positions of substrings on which $x$ and $y$ differ.}

   \vspace{1mm}
   \hrule\vspace{1mm}
 
   \For{$j \in \{0,\ldots,\tdepth \gsize-1\}$}{

        $\locmismatch^j = \hmrrecover(\gsketch^j(x),\gsketch^j(y))$.  

        \For{$(v,\alpha,\beta) \in \locmismatch^j$}{

            $\claimedsize_v(x) = \alpha \mod n$ (over $\N$)

            $\claimedsize_v(y) = \beta \mod n$ (over $\N$)

        }
    }

    $\claimednodes=\emptyset$.

    \For{$v\in \stringnodes$}{
    

        $\claimedlocation_v(x) = \sum_{u \text{ a left ancestor of $v$}} \claimedsize_u(x)$.   

        $\claimedlocation_v(y) = \sum_{u \text{ a left ancestor of $v$}} \claimedsize_u(y)$.   

        
        \COMMENT: $u \in \{0,1\}^{<|v|}$ is a left ancestor of $v$ if $u$ is a prefix of $v$ and $v_{|u|+1}=1$. 

         \If{$\claimedlocation_v(x)$ and $\claimedlocation_v(y)$ are both defined}{Add $v$ to $\claimednodes$.}

    }
        
    Output $(\claimedlocation_v(x):v \in \claimednodes)$ and $(\claimedlocation_v(y):v \in \claimednodes)$.

\end{algorithm}

Procedures $\findstrings$ and $\findlocations$ are called from $\mreconstruct$ which attempts to reconstruct edit distance information
from a pair of sketches $\msketchout(x)$ and $\msketchout(y)$ obtained from a single run of $\msketch$.
$\msketchout(x)$ consists of $\fsketch(x)$ and $\gsketch(x)$, and similarly for $\msketchout(y)$.
Sketches $\fsketch(x)$ and $\fsketch(y)$ are sent to $\findstrings$ to recover differing pairs in $\ztree$ of $x$ and $y$:
collections $\claimedstring(x)$ and $\claimedstring(y)$ indexed by some subset $\stringnodes$ of the tree nodes.

Then $\findlocations$ is applied on $\gsketch(x)$, $\gsketch(y)$ and $\stringnodes$ to obtain 
starting locations $\claimedlocation(x)$ of substrings $\claimedstring(x)$ within $x$
and similarly starting positions $\claimedlocation(y)$ of substrings of $y$.
The set of substrings for which the starting position was recovered is identified by $\claimednodes$.
We let $\reportednodes$ be the set of nodes from $\claimednodes$ that do not have any other node from $\claimednodes$ 
on their path to the root in $T(\BASE^{\tdepth})$.
For each node  $v \in \reportednodes$ we have recovered substrings $\claimedstring_v(x)$ and $\claimedstring_v(y)$ and their respective
positions $\claimedlocation_v(x)$ and $\claimedlocation_v(y)$ in $x$ and $y$.
We evaluate the edit distance of the two substrings positioned at those locations and we calculate the costly edges of their canonical alignment.
Union of all the edges over $v \in \reportednodes$ is the output of $\mreconstruct$.

\begin{algorithm}[H]
\label{algo:mreconstruct}
   \caption{$\mreconstruct(\msketchout(x), \msketchout(y))$}
   \KwIn{A pair $\msketchout(x)$, $\msketchout(y)$ from a single run of $\msketch$, where $\msketchout(x)$ consists of $\fsketch(x)$ and $\gsketch(x)$.}
   \KwOut{Set of candidate costly edges for some pairs of substrings where $x$ and $y$ differ.}
   
   \vspace{1mm}
   \hrule\vspace{1mm}
 
   Call $\findstrings(\fsketch(x),\fsketch(y))$ to get  $(\claimedstring_v(x):v \in \stringnodes)$ and
$(\claimedstring_v(y):v \in \stringnodes)$ for some $\stringnodes \subseteq \BASE^{\le \tdepth}$. 


    Call $\findlocations(\gsketch(x),\gsketch(y),\stringnodes)$ to get $(\claimedlocation_v(x):v \in \claimednodes)$ and
$(\claimedlocation_v(y):v \in \claimednodes)$ for some $\claimednodes \subseteq \stringnodes$. 


    Let $\reportednodes=\{v \in \claimednodes;\;$ $v$ has no proper prefix in $\claimednodes\}$.

    \For{$v \in \reportednodes$}{

        Calculate the canonical alignment $\canon^+(\claimedstring_v(x),\claimedstring_v(y))$ of the string $\claimedstring_v(x)$ starting at location $\claimedlocation_v(x)$ and the string $\claimedstring_v(y)$ starting at location $\claimedlocation_v(y)$.

        Let $\claimededges_v = \costly{\canon^+(\claimedstring_v(x),\claimedstring_v(y)}$.

    }

    $\claimededges \leftarrow \bigcup_{w \in \reportednodes} \claimededges_v$.

   Output $\claimededges$.

\end{algorithm}

The main recovery function is given by $\edrecover$.
It receives two sketches for $x$ and $y$.
The sketch of $x$ contains Ostrovsky-Rabani fingerprint $\threshskout(x)$ of the whole $x$
and a sequence of sketches $\msketchout_1(x),\ldots,$ $\msketchout_{\numrep}(x)$. 
Similarly for the sketch of $y$.
If $\threshskout(x)$ differs from $\threshskout(y)$ then $x$ and $y$ are of edit distance more than $k$.
Otherwise we invoke procedure $\mreconstruct$ for each $\msketchout_{i}(x)$ and $\msketchout_{i}(y)$, $i=1,\dots,\numrep$.
Each call to  $\mreconstruct$ produces a collection of costly edges that should be part of the canonical alignment of $x$ and $y$.
Some of the edges might be misidentified by a given call to $\mreconstruct$ so as the final output of $\edrecover$ we output only edges that
$\mreconstruct$ outputs more than half of the time.

\begin{algorithm}[H]
\label{algo:edrecover}
   \caption{$\edrecover(\fullsketchout(x),\fullsketchout(y);k,n)$}
   \KwIn{Sketch of $x$: $\fullsketchout(x)=(\threshskout(x)$, $\msketchout_1(x),\ldots,$ $\msketchout_{\numrep}(x))$ and 
    sketch of $y$:  $\fullsketchout(y)=(\threshskout(y),$ $\msketchout_1(y),\ldots,\msketchout_{\numrep}(y))$.}
   \KwOut{Edit distance of $x$ and $y$ if it is smaller than $k$, and $\LRGE$ otherwise.}
   
   \vspace{1mm}
   \hrule\vspace{1mm}

   \lIf{$\threshskout(x) \neq \threshskout(y)$}{return $\LRGE$}

    \For{$i \in \{1,\ldots,\numrep\}$}{
       
        $\claimededges_i = \mreconstruct(\msketchout_i(x),\msketchout_i(y))$ 
    }

    Let $\outputset$ be the set of edges that appear in more than half of the sets $\claimededges_i$.

    \lIf{$|\outputset|>k$}{return $\LRGE$} 
    
    Return $\outputset$.

\end{algorithm}

\section{Analysis of the sketch-and-recover scheme}
\label{sec:analysis}

\subsection{Main result}
\label{subsec:main result}

An execution of $\fullsketch$ and $\edrecover$ on a pair of strings $x,y$ consists of evaluation of $\fullsketch(x)$ and $\fullsketch(y)$ and of $\edrecover(\fullsketch(x),\fullsketch(y))$ where all three evaluations use the same randomizing parameter $\rho$.  The execution of $\fullsketch(x)$, $\fullsketch(y)$ and $\edrecover(\fullsketch(x),\fullsketch(y))$ \emph{succeeds} provided that 
\begin{itemize}
    \item If $\ED(x,y) > k$ then $\edrecover$ outputs $\LRGE$, and
    \item If $\ED(x,y) \leq k$ then $\edrecover$ outputs $\costly{\canon^+(x,y)}$,
\end{itemize}
and fails otherwise.

In this section we prove the main result of the paper:

\begin{theorem}
\label{thm:main}
The pair of algorithms $\fullsketch$ and $\edrecover$ provide a sketch-and-recover scheme for the \critset{} for edit distance. 
Given sufficiently large size parameter $n$, $\gsize=\log n$, distance parameter $k\le n/40\orgap$ and alphabet $\Sigma_n$:
\begin{enumerate}
\item $\fullsketch$ runs on strings $x$ of
length at most $n$ and outputs a sketch $\fsketch(x)$ of length at most $O(k \orgap^2 \gsize ^{O(1)})$.
\item 
The running time of $\fullsketch$ is $\OO(|x| \cdot T(n))$ and the  running time of $\edrecover$ is $\OO(\min(n^2,k^3 \orgap^2))$, where $T(n)\ge n$ is the time to compute $\threshsk$ 
on inputs of length $n$. (Recall, $T(n)$ is polynomial in $n$.)
\item  For all $x$,$y$ of
length at most $n$, the execution of  $\fullsketch(x)$, $\fullsketch(y)$, and
$\edrecover($ $\fullsketch(x),$ $\fullsketch(y))$  fails with probability at most 1/3.
\end{enumerate}
\end{theorem}

As usual, the value 1/3 for failure probability is arbitrary, and can be replaced by any function  $\delta>0$ by repetition:
construct a sketch consisting of $r=O(\log(1/\delta))$ independent sketches of $x$ and of $y$, apply
the recovery algorithm to each pair of corresponding sketches, and take majority vote.

The bounds on the sketch size and running time are straightforward and given in Section~\ref{subsec:time and space}.

The most
significant part of the analysis is the upper bound on the failure probability. 
This proof is contained in Sections~\ref{subsec:Normal} to~\ref{subsec:LOC}.

\subsection{Time and space analysis}
\label{subsec:time and space}

In this section we prove parts 1 and 2 of Theorem~\ref{thm:main}.
For part 1 we bound the size of the output of $\fullsketch$.  The output
consists of $\numrep$ independent instances of the output of $\vectorcondense$ and $\locationcondense$.
Each call to $\vectorcondense$ produces $\tdepth$ sketches using $\hmrsketch$, each of bit-size $\OO(\orgap^2 k)$.
Each call to $\locationcondense$ produces $\tdepth \gsize-1$ sketches using $\hmrsketch$, each of bit-size $\OO(\orgap^2 k)$.
Hence, the total size of the sketch is $\OO(k \orgapvalue)$.

For part 2 we bound the time of $\fullsketch$ and $\edrecover$.
Let $m=|x|$ and $T(n)$ be the time to compute $\threshsk$ on strings of length at most $n$.
$\fullsketch$ makes one call to $\threshsk$ which takes time $T(n)$
and then makes $\numrep=\OO(1)$ independent calls to $\msketch$. 
$\msketch$ calls $\mdecomp$, $\vectorcondense$, and $\locationcondense$.

$\mdecomp$ constructs  $\ztree$, $\bvec$, $\tstree$, $\ztree^*$, $\lvec^*$, and $\fptree$.
Even though each of these is defined on trees which have (roughly) $n^{\tdepth}$ nodes, which is superpolynomial
in $n$, the time to compute them will be polynomial in $m$ as we use sparse representation for each labeling of the tree.
Sets $\nonemptynodes$ and $\nonemptynodes^*$ determine the nodes with non-empty labels in those trees.

Since $\ztree$ is a decomposition tree for $x$, each layer is a decomposition of $x$ and has at most $m$ nonempty strings, so  $|\nonemptynodes| \leq \tdepth m$.
Thus, the total cost of calls to $\bdecomp$ for strings on a single layer of $\ztree$ is $\OO(m)$. 
Hence, each layer of $\bvec$ contains in total $\OO(m)$ ones in the grammar bit-vectors. 
To construct $\tstree$, we apply $\threshsk$ to each nonempty $\ztree_v$.
This take total time $\OO(m \cdot T(n))$.

The decomposition tree $\ztree^*$ also uses a sparse representation and is explicitly defined
on the set $\nonemptynodes^*$ which consists of all nodes in the binary tree that are ancestors of
a node in $\nonemptynodes$.  
Since each node in $\nonemptynodes$ has at most $\tdepth\gsize$ ancestors
$|\nonemptynodes^*| = \OO(m)$.  
Hence, $\lvec^*_v$ can be computed in time $\OO(m)$.
We apply $\fingerprint$ to each node of $\ztree^*$ in $\nonemptynodes^*$.
We need $\OO(m)$ time to compute the fingerprints for nodes on a single layers of $\ztree^*$.
Thus the total time to compute $\ztree^*$, $\lvec^*$, and $\fptree$ is $\OO(m)$.

$\vectorcondense$ involves the computation of $\tdepth$ instances
of  $\hmrsketch$.  
The $j$-th instance is for $\bvec^j$ which is the $j$-th layer of $\bvec$.
Thus the sparse representation of $\bvec^j$ contains at most $\OO(m)$ entries
so the time to compute a single $\hmrsketch$ is $\OO(m + \orgapvalue k)$.

Similarly $\locationcondense$ involves the computation of $\tdepth\gsize-1$
instances of $\hmrsketch$.  
The $j$-th instance is for the function $\lvec^j$ which has at most $m$ non-zero entries 
so again the  time of $\hmrsketch$ is $\OO(m + \orgapvalue k)$.

Thus the overall  time of $\fullsketch$ is $\OO(m \cdot T(n) + m + \orgapvalue k)$.

Next we consider the time of $\edrecover$.  
We will argue first that its running time is at most $\OO(n^2)$.
The we point out how to make it faster.
$\findstrings$ is applied to each pair
$\fsketch^j(x),\fsketch^j(y)$ 
for $j \in \{1,\ldots,\tdepth\}$ and $\findlocations$ is applied to each pair $\gsketch^j(x),$ $\gsketch^j(y)$
for $j \in \{1,\ldots,\tdepth\gsize-1\}$.  

Each call to $\findstrings$ runs $\hmrrecover$ whose running time is $\OO(\orgapvalue k)$.  
This produces a set of at most $\OO(\orgapvalue k)$ mismatch pairs.  
This determines the set $\stringnodes$ of nodes, also of size
at most the number of mismatch pairs.
For each node $v \in \stringnodes$ we use the mismatch pairs to construct $\claimedbv^j_v$, which is
a proposed reconstruction of $\bvec_v$.  
We apply $\decode$ to each of these to get $\claimedstring_v$. 
The total time is within $\OO(n)$.

Similarly, each call to $\findlocations$ runs $\hmrrecover$.  
This produces a set of at most $\OO(\orgapvalue k)$ mismatch pairs.
For each node in $\stringnodes$ we calculate a sum of at most $\gsize \tdepth$ elements.
Hence, the total time is $\OO(n)$.

Last, for each node $v$ in $\stringnodes$ we have to compute the edit distance of $\claimedstring_v(x)$ and $\claimedstring_v(y)$
which takes time $O(|\claimedstring_v(x)| \cdot |\claimedstring_v(y)|)$.
By Corollary~\ref{cor:correctly reconstructed} (which is stated later) under {\em normal execution}, 
both $\claimedstring_v(x)$ and $\claimedstring_v(y)$ are real substrings of $x$ and $y$, respectively.
Thus the sum of lengths of $\claimedstring_v(x)$ at the same layer of the decomposition tree is at most $n$.
Thus the total time needed for the reconstruction is $\OO(n^2)$ under {\em normal execution}.
If the recovery algorithm runs for more time than $\OO(n^2)$ we can terminate it as the execution is {\em abnormal} and the result might be incorrect.
That happens with only small probability as shown later.

To obtain a faster running time for the recovery, observe that one does not have to expand the grammars to compute their canonical alignment.
Using the result of~\cite{ED_compressed_string_Soda22} one can compute the set of costly edges of the canonical alignment
for two strings represented by grammars of size at most $g$ in time $\OO(g+k'^2)$ where $k'$ is the edit distance of the two strings.
The total size of grammars we will recover is $\OO(k \orgap^2)$, their number is also at most $\OO(k \orgap^2)$, and the edit distance we care for is $k$.
Hence running the edit distance algorithm on all recovered pairs of grammars costs at most $\OO(k^3 \orgap^2)$.

\subsection{Upper bound on the failure probability}
\label{subsec:success}
Having proved parts 1 and 2 of Theorem~\ref{thm:main}, we now turn to the proof of Part 3.  In this section we state Theorem~\ref{thm:outputset} and use it to prove
Part 3 of Theorem~\ref{thm:main}.  The proof of Theorem~\ref{thm:outputset} will be divided into several subsections. 

We fix strings $x$ and $y$ and consider an execution of $\fullsketch(x)$, $\fullsketch(y)$ and $\edrecover(\fullsketch(x),\fullsketch(y))$.  We want to show that the probability of failure is at most
1/3.

Recall that the algorithm $\edrecover$ first compares $\threshskout(x)$ and $\threshskout(y)$ and outputs $\LRGE$ if they are unequal.  If they are equal then
it executes $\edrecover$ to produce the set $\outputset$ of annotated edges and outputs $\LRGE$ if
$|\outputset|>k$ and outputs $\outputset$ if $|\outputset|\leq k$.

It simplifies the discussion to modify
$\edrecover$ so that it starts by computing $\outputset$ and only then compares $\threshskout(x)$ to $ \threshskout(y)$ and outputs $\LRGE$ 
if the fingerprints differ and outputs $\outputset$ otherwise.
This modified algorithm clearly has the same output behavior as $\edrecover$ but has the analytic advantage that the  (random) value of $\outputset$ is independent of the event $\threshskout(x) = \threshskout(y)$.

We will prove:

\begin{theorem}
\label{thm:outputset}
If $\ED(x,y) < 20k \orgap$  then
$\prob[\outputset \neq \costly{\canon^+(x,y)}]\leq 1/4$.
\end{theorem}

Using this theorem, it is easy to bound the failure probability of $\edrecover$ to complete the proof of the third part of
Theorem~\ref{thm:main}.

We divide into three cases depending on $\ED(x,y)$.

{\bf Case (1):} $\ED(x,y) > 20k\orgap$. The algorithm correctly outputs $\LRGE$ 
unless $\threshsk(x) = \threshsk(y)$ and by 
Theorem~\ref{thm:OR}, $\threshsk(x)=\threshsk(y)$ happens with probability at most $1/n^4$. 

{\bf Case (2):}$20k \orgap > \ED(x,y) > k$. The algorithm correctly
outputs $\LRGE$ if
$\threshsk(x) \neq \threshsk(y)$  or if $\outputset=\costly{\canon^+(x,y)}$ (since $|\costly{\canon^+(x,y)}|=\ED(x,y)$).  
So the output will be incorrect only if $\outputset \neq \costly{\canon^+(x,y)}$.
By Theorem~\ref{thm:outputset} this happens with probability at most $1/4$. 

{\bf Case (3):}$k \geq \ED(x,y)$. Success requires the algorithm to output $\costly{\canon^+(x,y)}$.  
This will  fail only if $\threshsk(x) \neq \threshsk(y)$ or  $\outputset \neq \costly{\canon^+(x,y)}$.
By Theorem~\ref{thm:OR}, $\prob[\threshsk(x) \neq \threshsk(y)] \leq \frac{k \orgap}{20k \orgap} =\frac{1}{20} $ and by Theorem~\ref{thm:outputset}, $\prob[\outputset \neq \costly{\canon^+(x,y)}] \leq \frac{1}{4}$ so the overall error probability is less than 1/3..

It remains to prove Theorem~\ref{thm:outputset} which
is done in the sections that follow.  Since the hypothesis of
Theorem~\ref{thm:outputset} is that $\ED(x,y) < 20k\orgap$ this assumption is
made throughout.  







\subsection{Normal executions}
\label{subsec:Normal}

In this section we identify certain events in the execution of $\fullsketch$ and $\edrecover$ as \emph{abnormal} and prove
that  the probability of an abnormal execution tends to 0 as $n$ gets large. 

On input $x$, $\fullsketch$ constructs a sequence of sketches $\msketchout_i(x)$ for $i \in \{1,\ldots, \numrep\}$.

For each corresponding pair  $\msketchout_i(x)$ and $\msketchout_i(y)$, the recovery algorithm applies $\mreconstruct(\msketchout_i(x),\msketchout_i(y))$.  If the output is defined then it is a set $\claimededges_i$ of annotated edges. The set $\claimededges_i$ is supposed to ``approximate'' $\costly{\canon(x,y)}$ in a suitable sense.  The final output $\outputset$ of $\edrecover$ is the set of annotated edges
that appear in more than half the sets of $\claimededges_i$.

Each  run of $\msketch(x)$ constructs a labeling $\ztree(x)$ of the tree $T(\BASE^{\tdepth})$ by substrings and another labeling $\bvec(x)$ of the same tree where for each $v$, $\bvec_v$ is a bit-vector of length $\bvsize$ that encodes a grammar for $\bvec(x)$. By the properties of
$\bdecomp$ these
labelings satisfy:

\begin{description}
\item[\DONE(x):] $\ztree(x)$ is a decomposition tree for $x$
\item[\DTWO(x):] For each node $v$ for which $\ztree_v(x) \neq \emptystr$,
 $\bvec_v(x)$ is an encoding of $\ztree_v(x)$, i.e. $\ztree_v(x)=\decode(\bvec_v(x))$,
and the number of 1's in $\bvec_v(x)$ is at most $\gramsize_v$
\end{description}

We now identify four \emph{abnormal events}.  We will prove that it is very unlikely
that any of these conditions occur.

\begin{description}
\item[Fingerprinting abnormality.] There is a node $v$, such that $\ztree^*_v(x) \neq \ztree^*_v(y)$ and $\fptree_v(x) = \fptree_v(y)$.
\item[Threshold detection abnormality.] There is a  node $v$ at a level $j$, such that $\ED(\ztree_v(x),$ $\ztree_v(y)) \geq \gapthresh_j$ and $\tstree_v(x) = \tstree_v(y)$.
\item[HMR abnormality.]  One of the calls to $\hmrrecover$ within the recovery algorithm fails to satisfy
the Completeness and Soundness conditions (see Theorem~\ref{thm:HMR}).
\end{description}

An execution of $\fullsketch(x)$, $\fullsketch(y)$, and $\edrecover(\fullsketch(x),\fullsketch(y))$ is \emph{abnormal} if any of these abnormalities ever occurs and is \emph{normal} if no such abnormality occurs.  We define $\normex$ to denote the event that the execution is normal, and $\abnormex$ to denote the event that the execution is abnormal.

\begin{lemma}
\label{lemma:no abnormality 1}
For any $x,y$, $\prob[\abnormex]$ tends to 0 as $n$ gets large.
\end{lemma}

\begin{proof}
First consider decomposition abnormality for $x$.  
The decomposition tree is built in $\mdecomp$, which grows a decomposition starting at the root by calling $\bdecomp$ on each node for which $\ztree_v$ is nonempty.  
\DONE{} or \DTWO{} fail only if some of the calls to $\bdecomp$ returns $\undefnd$.
By Theorem~\ref{thm:BK} the probability that a particular call to $\bdecomp(z)$ returns $\undefnd$ is at most $1/n^4$.  
In a given call to $\mdecomp(x)$, the number of nonempty strings at each level of the decomposition is at most $|x| \leq n$ 
so the number of calls made to $\bdecomp(z)$ for nonempty $z$ is at most $\tdepth n$.  
So the probability that there is a call to $\bdecomp$ that returns $\undefnd$ is at most $\tdepth /n^3$.  
If every call to $\bdecomp$ is defined then a simple induction using the properties of $\bdecomp$ ensures that the resulting tree satisfies \DONE{} and \DTWO{}.  
Taking a union bound over the $\numrep$ calls to $\mdecomp$ still gives a probability of abnormality that is $\le 1/n$ for $n$ large enough.
Similarly for $y$.

The probability that a fingerprinting abnormality occurs at a particular node $v \in \nonemptynodes^*$ is at most $1/n^4$ by Theorem~\ref{thm:fingerprint}.   
Similarly, the probability that a  threshold detection abnormality occurs at $v\in \nonemptynodes$ is at most $1/n^4$ by Theorem~\ref{thm:OR}.
Taking a union bound over the at most $\gsize \tdepth n$ nodes in $\nonemptynodes$ and $\nonemptynodes^*$, and 
the $\numrep$ calls to $\msketch$ yields a probability of either of these abnormalities that is at most $1/n$ for $n$ large enough.

An HMR abnormality occurs if one of the calls to $\hmrrecover$ fails. 
(A failure of $\hmrrecover$ is determined by the randomness used to build the sketches of $x$ and $y$).  
We set the failure parameter of the hierarchical mismatch recovery scheme to be $1/n^4$ and 
the number of its instances is $(\gsize \tdepth)^{O(1)}$ so the overall probability of failure
tends to zero as $n$ grows. 
\end{proof}

\subsection{Reduction to the main lemma}
\label{subsec:Main lemma}
In this section we introduce the Main Lemma (Lemma~\ref{lemma:mreconstruct}) and use it to prove Theorem~\ref{thm:outputset}.

We need to define some subsets of the 
and edge sets of $\grid(x,y)$. 

The \emph{$i$-th vertical edge slice} $S_V(i)$ of the grid graph
$\grid(x,y)$ consists of all edges that join a vertex with first coordinate $i-1$ to a vertex with first coordinate $i$. The \emph{$j$-th horizontal edge slice} $S_H(j)$  consists of all edges that join a vertex with second coordinate $j-1$ to a vertex with second coordinate $j$. 

Note that a vertical edge slice consists of horizontal and diagonal edges but not vertical edges, and a horizontal edge slice consists of vertical and diagonal edges but not horizontal edges. The union of all of the slices is the entire edge set of $\grid(x,y)$.  
We leave the proof of the following observation to the reader.

\begin{proposition}
    \label{prop:exactly one edge}
    Any spanning path of $\grid(x,y)$ contains exactly one edge from each horizontal slice and
    from each vertical slice.  In particular, 
    $\costly{\canon(x,y)}$ contains at most one edge from any slice.
\end{proposition}

The sketch-and-recover scheme $\fullsketch + \edrecover$ operates by performing $\numrep$ independent executions
of $\msketch+\mreconstruct$ and taking majority vote.  The Main Lemma below states the key property of a single instance
of $\msketch+\mreconstruct$, and Theorem~\ref{thm:outputset} follows easily from it.  

\begin{lemma}
\label{lemma:mreconstruct}{(\textbf{Main Lemma})}
Let $x,y$ be strings such that $\ED(x,y) < 20k\orgap$. 
Consider a single run of $\msketch$ on $x$ and $y$ producing $\msketchout(x)$ and $\msketchout(y)$, followed by
a run of $\mreconstruct$ on $\msketchout(x),\msketchout(y)$, which produces the set $\claimededges$ of annotated edges. 
For any   (horizontal or vertical) edge slice $S$
\[
\prob[(\claimededges \cap S \neq \costly{\canon(x,y)} \cap S) \AND \normex]
\]
tends to 0 as $n$ gets large.
\end{lemma}

\begin{proof}[Proof of Theorem~\ref{thm:outputset} assuming Lemma~\ref{lemma:mreconstruct}]
Assume $\ED(x,y) < 20k\orgap$.
We must establish an upper bound on
$\prob[\outputset \neq \costly{\canon^+(x,y)}]$.  We have:

\begin{align*}
\prob[&\outputset  \neq \costly{\canon^+(x,y)}] \\
& \leq \prob[\abnormex]+\prob[\outputset \neq \costly{\canon^+(x,y)} \AND \normex].
\end{align*}
Since the first term tends to 0 as $n$ gets large it suffices to prove:

\begin{equation}
    \label{eqn:outputset}
  \prob[\outputset \neq \costly{\canon^+(x,y)} \AND \normex] \leq 
  \frac{1}{5},  
\end{equation}

The hypothesis of the theorem assumes $n$ is sufficiently large.  We choose $n$ large enough so that in Lemma~\ref{lemma:mreconstruct}
for any edge slice $S$ the probability that $\claimededges(x,y) \cap S \neq \costly{\canon(x,y)} \cap S$ is at most $1/10$.

We note that every edge of $\grid(x,y)$ lies in some (horizontal or vertical) edge slice.  It follows that $\outputset = \costly{\canon^+(x,y)}$ if and only if for all edge slices $S$ we have $\outputset \cap S= \costly{\canon^+(x,y)} \cap S$. 
Define the bad event $B(S)$ to be the event that for more than half of the
sets $(\claimededges_i:i \in \{1,\ldots, \numrep\})$, $\claimededges_i \cap S \neq \costly{\canon^+(x,y)} \cap S$.  By the definition of $\outputset$, if $\outputset \neq \costly{\canon^+(x,y)}$ then there is at least one slice $S$ for which $B(S)$ happens and so:

\[
\prob[\outputset \neq \costly{\canon^+(x,y)} \AND \normex] \leq \sum_S \prob[B(S) \AND \normex].
\]

Consider a particular slice $S$.  By Lemma~\ref{lemma:mreconstruct} and the above choice of $n$, 
for each $i$,
$\prob[(\claimededges_i \cap S \neq \costly{\canon^+(x,y)} \cap S) \AND \normex] \leq 1/10$. Therefore the
Chernoff-Hoeffding bound (Lemma~\ref{lemma:CH}), implies that 
$\prob[B(S) \AND \normex] \leq 2e^{-0.32(10\gsize+50)}$ which is easily less than $\frac{1}{10}2^{-\gsize} \le \frac{1}{10n}$.  
Since the number of horizontal and vertical slices are each at most $n$:
\begin{eqnarray*}
  \prob[\outputset \neq \costly{\canon^+(x,y)} \AND \normex] & \leq & 2n \max_S\prob[B(S) \AND \normex] \\
  & \leq & \frac{1}{5},
\end{eqnarray*}
as required to complete the proof of Theorem~\ref{thm:outputset}.
\end{proof}

It remains to prove Lemma~\ref{lemma:mreconstruct}.

\subsection{Preliminaries for the proof of Lemma~\ref{lemma:mreconstruct}}
\label{subsec:mreconstruct prelims}
Lemma~\ref{lemma:mreconstruct} concerns a single execution of $\msketch$ applied to $x$ and $y$ producing
outputs $\msketchout(x)$ and $\msketchout(y)$ and
a single instance of $\mreconstruct$ applied to $\msketchout(x)$ and $\msketchout(y)$.

$\msketch(x)$ generates three (partial) labelings  $\ztree,\bvec,$ and $\tstree$ of the nodes of
the tree $T(\BASE^d)$.
$\ztree(x)$ is a string decomposition tree of $x$. 
The set $\nonemptynodes(x)$ is the set of nodes $v$ such that $\ztree_v(x) \neq \emptystr$.
For all $v \in \nonemptynodes$, $\bvec_v(x)$ is a bit-vector of length $\bvsize$ such that \newline $\decode(\bvec_v(x))=\ztree_v(x)$;
$\tstree_v(x)$ is the Ostrovsky-Rabani fingerprint of $\ztree_v(x)$.

$\msketch(x)$ also generates three partial labelings $\ztree^*,\bvec^*,$ and $\fptree$ of the binary tree $T(\{0,1\}^{\tdepth\gsize})$.
$\ztree^*(x)$ is a refinement of $\ztree(x)$ to the tree $T(\{0,1\}^{\tdepth\gsize})$.
$\ztree^*_v(x)$ is explicitly defined on the set $\nonemptynodes^*$, consisting of those binary strings that are prefixes (ancestors) of nodes in $\nonemptynodes$.
For $v \in \nonemptynodes^*$, $\lvec^*(x)$ is the size of the string assigned by $\ztree^*$ to the left-child $v \circ 0$ of $v$,
and $\fptree_v(x)$ is the Karp-Rabin fingerprint of $\ztree^*_v(x)$.

The sketch $\msketchout(x)$ consists of two separate sketches: $\fsketch(x)$ and $\gsketch(x)$.
$\fsketch(x)$ is a sequence of $\tdepth$ HMR sketches $\fsketch^j(x)$ produced by $\vectorcondense$, one for each layer of grammars in $\bvec$.
$\gsketch(x)$ is a sequence of $\tdepth \gsize-1$ sketches $\gsketch^j(x)$ produced by $\locationcondense$ for each layer of $\lvec^*(x)$ other than the last layer.

We make the following definitions associated to
 $\ztree(x)$ and $\ztree(y)$:
 For each node $v$:
\begin{itemize}
\item  $\interval_v(x)$ denotes the interval location $\zint{\ztree}_v(x)$ where $\ztree_v(x)$ is located within $x$ (as defined in Section~\ref{subsec:tree decomp}).
$\startindex_v(x)$ is the left endpoint of $\interval_v(x)$.  
This is equal to the sum of the lengths of all strings at leaves that are to the left of (lexicographically precede) $v$. 
Note that $\interval_v(x)=(\startindex_v(x),\startindex_v(x)+|\ztree_v(x)|]$.
\item $\grid_v$ denotes the subgrid $\grid_{\interval_v(x) \times \interval_v(y)}$, 
and $\edges_v=\edges_v(x,y)$ denotes the set of its edges.  
$\edges_v^+$ denotes the set of annotated edges $\{e^+:e \in \edges_v\}$. 
\item  $\ED_v=\ED(\ztree_v(x),\ztree_v(y))$, 
and the canonical path $\canon(x,y)_{\interval_v(x) \times \interval_v(y)}$ is denoted by $\canon_v$.
$\costly{\canon_v}$ is the set of costly edges of $\canon_v$ (which has size $\ED_v$),
and $\costly{\canon^+_v}$ is the set of costly edges with annotations.
\end{itemize}

The procedure $\mreconstruct$, using subroutines $\findstrings$ and $\findlocations$  constructs a set $\claimededges$
of annotated edges as follows:

\begin{itemize}
    \item $\findstrings$ takes as input $\fsketch(x)$ and $\fsketch(y)$ and constructs
(partial) labelings $\claimedstring(x)$ and $\claimedstring(y)$ 
of the tree $T(\BASE^{\tdepth})$.  For a node $v \in \BASE^{\leq \tdepth}$, if $\claimedstring_v(x)$ (resp.
$\claimedstring_v(y)$) is defined, its value is a string.

\item $\findlocations$ takes as input $\gsketch(x)$ and $\gsketch(y)$ and constructs
two (partial) labelings $\claimedlocation(x)$ and $\claimedlocation(y)$ of $T(\BASE^{\tdepth})$. For a node $v \in \BASE^{\tdepth}$, if $\claimedlocation_v(x)$ (resp. $\claimedstring_v(y)$) is defined, its value is a nonnegative integer.
\item In $\mreconstruct$, the set $\claimednodes$ is defined to be the set of nodes $v$ for which $\claimedstring_v(x)$, $\claimedstring_v(y)$, $\claimedlocation_v(x)$ and $\claimedlocation_v(y)$ are all defined.
\item In $\mreconstruct$ the set $\reportednodes$ is defined to be the subset of $\claimednodes$ consisting of those nodes in $\claimednodes$ that have no ancestor in $\claimednodes$.
\item For each node $v \in \reportednodes$, $\mreconstruct$ computes a set of annotated edges $\claimededges_v$. 
\item The output of $\mreconstruct$ is the set
$\claimededges=\bigcup_{v \in \reportednodes} \claimededges_v$.
\end{itemize}

For each $v \in \reportednodes$ the set $\claimededges_v$ is computed in $\mreconstruct$ as follows.
$\claimedstring_v(x)$ and $\claimedstring_v(y)$ together with $\claimedlocation_v(x)$ and $\claimedlocation_v(y)$ determine an edit distance problem, and that problem has a unique canonical path.  The set of costly annotated edges in that path is  $\claimededges_v$. 
In the procedure $\mreconstruct$, $\claimededges_v$ is evaluated only for $v \in \reportednodes$, but the definition makes sense for all $v \in \claimednodes$, and in what follows we will refer to $\claimededges_v$ for $v \in \claimednodes$. 

We are hoping that $\claimededges$ is close, in a suitable probabilistic sense to the set $\costly{\canon^+(x,y)}$.  
To show this we will need to consider certain properties for classifying nodes.

The first set of properties is determined by $\ztree_v(x)$ and $\ztree_v(y)$ produced by the $\msketch$.
\begin{itemize}
 \item $v$ is \emph{compatible} if the
box $\interval_v(x) \times \interval_v(y)$ is $(x,y)$-compatible.  (Recall  that the box
$I \times J$
is $(x,y)$-compatible provided that the 
the canonical alignment $\canon(x,y)$ passes through the corner points of $\interval_v(x) \times \interval_v(y)$.)   
The importance of this property is that if $v$ is compatible then $\canon_v$ is equal to $\canon(x,y) \cap \edges_v$, 
the portion of $\canon(x,y)$ that lies inside $\grid_v$.    
This means that if we are able to reconstruct the substrings $\ztree_v(x)$ and $\ztree_v(y)$ and the intervals $\interval_v(x)$ and $\interval_v(y)$, $\costly{\canon^+_v}$ will give us the portion of $\costly{\canon^+(x,y)}$ that lies inside of $\grid_v$.
\item $v$ is \emph{compatibly split} provided that all of the children of $v$ are compatible.
This depends on the run of $\bdecomp$ on the strings $\ztree_v(x)$ and $\ztree_v(y)$ that occurs inside of $\mdecomp$. 
\end{itemize}



We say that:
\begin{itemize}
\item 
\emph{$\claimedstring_v(x)$ (respectively, $\claimedstring_v(y)$) is a correct reconstruction} provided that it is defined and equal
to $\ztree_v(x)$ (respective, $\ztree_v(y)$). 
\item
\emph{ $\claimedlocation_v(x)$ (resp. $\claimedlocation_v(y)$) is a correct reconstruction} provided that
it is defined and equal to $\startindex_v(x)$ (resp. $\startindex_v(y)$). 
\item We say that \emph{$v$ is correctly reconstructed} if $\claimedstring_v(x)$, $\claimedstring_v(y)$,
$\claimedlocation_v(x)$ and $\claimedlocation_v(y)$ are all correct reconstructions.
\end{itemize}

\begin{lemma}
    \label{lemma:correct reconstruction}
    Let $v \in \BASE^{\leq \tdepth}$. Under a normal execution:
    \begin{enumerate}
    \item If $\claimedstring_v(x)$ and $\claimedstring_v(y)$ are both defined by $\findstrings$ then $\claimedstring_v(x)=\ztree_v(x)$ and $\claimedstring_v(y)=\ztree_v(y)$.
    \item If $\claimedlocation_v(x)$ and $\claimedlocation_v(y)$ are both defined by $\findlocations$ then
    $\startindex_v(x)=\claimedlocation_v(x)$ and $\startindex_v(y)=\claimedlocation_v(y)$.
    \end{enumerate}
\end{lemma}

\begin{proof}
Fix a node $v \in \BASE^{\leq \tdepth}$ and let $j=|v|$. 

Proof of 1. $\claimedstring_v(x)$ and $\claimedstring_v(y)$ are determined
during iteration $j$ of the main loop in $\findstrings$ as follows:
$\hmrrecover(\fsketch^j(x),\fsketch^j(y))$ is evaluated and assigned to $\stringmismatch^j$.  
This is a set of mismatch triples of the form $((w,i),\alpha,\beta)$
where $(w,i) \in \BASE^j \times \{1,\ldots,\bvsize\}$, and $\alpha \neq \beta \in \bigfield$. 
Letting $\stringmismatch^j_v$ denote the set of those triples involving the node $v$, 
$\claimedbv_v(x)$ and $\claimedbv_v(y)$ are determined by $\stringmismatch^j_v$.

Under a normal execution, the call to $\hmrrecover$ is Sound, so $\stringmismatch^j_v$ is a subset of the mismatch information $\MI(\bvec^j_v(x),\bvec^j_v(y))$.  

\begin{claim}
    \label{claim:claimedbv}
    $\claimedbv_{v}(x) \leq \bvec_v(x)$ and $\claimedbv_v(y) \leq \bvec_v(y)$, where $\leq$ is entry-wise. 
\end{claim}

\begin{proof}
We prove the first inequality; the second is analogous. Fix an index $i \in \{1,\ldots,\bvsize\}$.
Assume that $\claimedbv_{v,i}(x)=1$ we must show that $\bvec_{v,i}(x)=1$.   
From the construction of $\claimedbv(x)$ in $\findstrings$, $\claimedbv_{v,i}(x)=1$ implies that there
is a mismatch triple  $((v,i),\alpha,\beta) \in \stringmismatch^j$ such that $\alpha \neq 0$.
Since $\stringmismatch^j \subseteq \MI(\bvec^j(x),\bvec^j(y))$ this implies that $\bvec^j_{v,i}(x)=\alpha \neq 0$ which implies that $\bvec_{v,i}(x)c=1$ by the definition of $\bvec^j$.
\end{proof}

Now if $\claimedbv_v(x) = \bvec_v(x)$ and $\claimedbv_v(y)=\bvec_v(y)$ then we have
$\claimedstring_v(x)=\decode(\claimedbv_v(x))=\decode(\bvec_v(x))=\ztree_v(x)$ and similarly $\claimedstring_v(y)=\ztree_v(y)$.
The conclusion of the first implication holds in this case.
Now consider the case  $\claimedbv_v(x) \neq \bvec_v(x)$.  
By Claim~\ref{claim:claimedbv}, $\claimedbv_v(x) < \bvec_v(x)$.
By the minimality property of $\decode$ (see Section~\ref{subsec:bdecomp}) $\claimedstring_v(x)=\decode(\claimedbv_v(x))$ is
$\undefnd$.
So the premise of the implication does not hold. 
Similarly when $\claimedbv_v(y) \neq \bvec_v(y)$.
This completes the proof of the first part of the lemma.

The argument for the second part of the lemma is similar, but has important differences.  $\claimedlocation_v(x)$ and $\claimedlocation_v(y)$ are determined in $\findlocations$.  The
main loop of $\findlocations$ has $\tdepth \gsize$ iterations, corresponding to the levels of
the binary tree $T(\{0,1\}^{\tdepth\gsize})$.  In iteration $i$, $\hmrrecover$ is applied
to $\gsketch^i(x),\gsketch^i(y)$ to produce $\locmismatch^i$. The sketches $\gsketch^i(x)$ and
$\gsketch^i(y)$ are sketches of the $D$-vectors $\leftchild^i(x)$ and $\leftchild^i(y)$ for $D=\{0,1\}^i$.
$\locmismatch^i$ is a set of  mismatch triples of the form $(w,\alpha,\beta)$ where $w \in \{0,1\}^i$
and $\alpha,\beta \in \bigfield$.  
Let $L_i$ be the subset of nodes $u$ that appear as first coordinate of some mismatch triple in $\locmismatch^i$. 
$\locmismatch^i$ is used to define $\claimedsize_u(x)$ and $\claimedsize_u(y)$ for all $u \in L_i$.
The soundness of $\hmrrecover$ (since the execution is normal) implies $L_i \subseteq \MI(\leftchild^i(x),\leftchild^i(y))$ and 
this implies $\claimedsize_u(x)=\leftchild_u(x)$ and $\claimedsize_u(y)=\leftchild_u(y)$.  

Now $\claimedlocation_v(x)$ is the sum of $\claimedsize_u(x)$ over left ancestors $u$ of $v$.
It is $\undefnd$ if any of the summands is undefined.  
If all summands are defined, 
then by the above, the sum is equal to the sum of $\leftchild_u(x)$ over all left ancestors $u$ of $v$
which is equal to $\startindex_v(x)$.  
The same argument applies to show $\claimedlocation_v(y)=\startindex_v(y)$, completing the proof of the second part of the lemma.
%
\end{proof}

\begin{corollary}
\label{cor:correctly reconstructed}
     Under a normal execution, suppose that $v \in \BASE^{\leq \tdepth}$ is placed in 
     $\claimednodes$ by $\mreconstruct$.  Then
     \begin{enumerate}
         \item Then $v$ is correctly reconstructed, i.e. $\claimededges_v=\costly{\canon^+_v}$.
         \item If, in addition, $v$ is compatible, then $\claimededges_v=\costly{\canon^+(x,y)} \cap \edges^+_v$.
     \end{enumerate}
\end{corollary}

\begin{proof}
By the code for $\mreconstruct$, $v \in \claimednodes$ implies $\claimedstring_v(y)$, $\claimedlocation_v(x)$ and $\claimedlocation_v(y)$ are all defined, and $\claimededges_v$
is the set of costly edges in the canonical path for the edit distance problem defined by $\claimedstring_v(x)$, $\claimedstring_v(y)$, $\claimedlocation_v(x)$ and $\claimedlocation_v(y)$.
By Lemma~\ref{lemma:correct reconstruction}, this is equal to the set of costly edges in the canonical
path associated to 
$\ztree_v(x)$, $\ztree_v(y)$, $\interval_v(x)$ and $\interval_v(y)$ which is the set
    $\canon^+_v$.   For the second part, when $v$ is compatible,  the conclusion follows from the first part together with Proposition~\ref{prop:canonical}.
\end{proof}

\section{Proof of Lemma~\ref{lemma:mreconstruct}.}
\label{sec:main lemma proof}

Suppose $S$ is an arbitrary slice. Without loss of generality, assume that $S=S_V(b)$ is the $b$-th vertical edge slice.  We say that
\emph{slice $S$ is correctly reconstructed} provided that $\claimededges(x,y) \cap S = \costly{\canon^+(x,y)} \cap S$.
We need to show:

\begin{equation}
    \label{eqn:mreconstruct}
    \prob[\text{slice $S$ is not correctly reconstructed} \AND \normex] \rightarrow 0
\end{equation}
as $n$ gets large.

Since $\ztree(x)$ is a string decomposition tree of $x$, for each
level $j \in \{0,\ldots,\tdepth\}$ there is a unique node $w(j)$ of $T(\BASE^d)$ such
that $b \in \interval_{w(j)}(x)$, and the sequence $w(0),\ldots,w(\tdepth)$ forms a root-to-leaf path.
The following lemma implies that, under a normal execution, whether $S$ is correctly reconstructed
only depends on the reconstruction of nodes in $w(0),\ldots,w(\tdepth)$.

\begin{proposition}
    \label{prop:agree on S}
    Assume a normal execution, and suppose $S$ is the $b$th vertical edge slice.
    \begin{enumerate}
    \item $\claimededges \cap S=\bigcup_{j \in \{0,\ldots,\tdepth\}:w(j) \in \reportednodes}\claimededges_{w(j)} \cap S$
    \item
    If $w(j)$ is compatible and $w(j) \in \claimednodes$ then $\claimededges_{w(j)} \cap S=\costly{\canon^+} \cap S$.
    \end{enumerate}  
\end{proposition}

\begin{proof}
For (1), by definition in $\mreconstruct$ 
$\claimededges \cap S=\bigcup_{v \in \reportednodes} \claimededges_v \cap S$.
To establish the claimed equality it suffices to show that $\claimededges_v \cap S=\emptyset$
for $v \not\in \{w(0),\ldots,w(\tdepth)\}$. 
By Corollary~\ref{cor:correctly reconstructed}, $\claimededges_v=\costly{\canon^+_v}$.
The path $\canon_v$ only contains edges in vertical slices $S_V(\ell)$ for which $\ell \in \interval_v(x)$.
Since $b \not\in \interval_v(x)$, $\canon_v \cap S = \emptyset$.

For (2), suppose $w(j)$ is compatible.  By 
Corollary~\ref{cor:correctly reconstructed}(2), $\claimededges_{w(j)}=\costly{\canon^+} \cap \edges^+_{w(j)}$.
Since $b \in \interval_{w(j)}(x)$, $\edges^+_{w(j)}$ includes
the edge of $\canon$ from slice $S$, and the equality follows.
\end{proof}

As a corollary we get  a sufficient condition for $S$ to be correctly reconstructed:

\begin{corollary}
\label{cor:S correct}
Assume a normal execution.  If either of the following conditions hold then $S$ is correctly reconstructed:
    \begin{enumerate}
\item $\costly{\canon^+} \cap S = \emptyset$ and none of the nodes
$w(0),\ldots,w(\tdepth)$ is in $\claimednodes$.
\item There is an index $j \in \{1,\ldots,\tdepth\}$ such that $w(j) \in \claimednodes$
and $w(0),\ldots,w(j)$ are all compatible.
    \end{enumerate}
\end{corollary}

\begin{proof}
For (1), if $\{w(0),\ldots,w(\tdepth)\} \cap \claimednodes=\emptyset$
then by Proposition~\ref{prop:agree on S}(1), $\claimededges \cap S=\emptyset=\costly{\canon^+} \cap S$.

For (2), let $j$ be the least  index satisfying the hypothesis of (2). 
By definition of $\reportednodes$ in $\mreconstruct$, $w(j)$ is
the unique node among $w(0),\ldots,w(\tdepth)$ that belongs to $\reportednodes$ and so
by Proposition~\ref{prop:agree on S}(1), $\claimededges \cap S = \claimededges_{w(j)} \cap S$,
which is equal to $\costly{\canon^+} \cap S$ by Proposition~\ref{prop:agree on S}(2). 
\end{proof}

Our strategy is to show that with fairly high probability one of the two sufficient conditions of
this corollary is satisfied.  To this end we define the following events, indexed by
$j \in \{0,\ldots,\tdepth\}$:

\begin{description}
    \item[$\eventT(j)$] is the event  that $\ED_{w(j)} < \gapthresh_j$.
    \item[$\eventT(\leq j)$] is the event that $\ED_{w(i)}<\gapthresh_i$ for all $i \leq j$.
    \item[$\eventC(j)$] is the event that $w(j)$ is compatible.
    \item[$\eventC(\leq j)$] is the event that $w(0),\ldots,w(j)$ are all compatible.
\end{description}

Recall that we have a global hypothesis that $\ED(x,y) < \lceil 20k\orgap \rceil_2=\gapthresh_0$, so that $\eventT(0)$ holds.
Also, $\eventC(0)$ holds trivially.  The following proposition says that if $\ED_{w(j)}<\gapthresh_j$ and $w(j)$ is compatible then $w(j+1)$ is almost certainly compatible:

\begin{proposition}
\label{prop:U(j)}
\begin{enumerate}
    \item For all $j \in \{1,\ldots,\tdepth\}$,   $\prob[\eventT(j-1) \AND \eventC(j-1) \AND \neg \eventC(j) \AND \normex] \leq \frac{1}{\log^2n}$.
    \item $\prob [\exists j \in \{1,\ldots,\tdepth\}\text{ such that } \eventT(\leq j-1) \AND \neg \eventC( \leq j) \AND \normex] \leq \frac{1}{\log n}$.
    
\end{enumerate}
\end{proposition}

\begin{proof}
For (1), note that events $\eventT(j-1)$ and $\eventC(j-1)$ are determined by the first $j-1$ levels of $\ztree(x)$ and $\ztree(y)$.  Consider an arbitrary setting of the first $j-1$ levels of $\ztree(x)$ and $\ztree(y)$ for which $\eventT(j-1)$ and $\eventC(j-1)$ hold.  Then $\ED_{w(j-1)}<\gapthresh_{j-1}$.
By Theorem~\ref{thm:BK}, the conditional probability given this setting that $w(j-1)$ is not compatibly split is at most:

\begin{equation}
    \label{eqn:NCS}
    \frac{\bkfailure \ED_{w(j-1)}}{\gramsize_j}\leq \frac{\bkfailure \gapthresh_{j-1}}{\gramsize_j} \le \frac{1}{\log^2n},
\end{equation}
where the final inequality uses the definitions of $\gapthresh_{j-1}$ and $\gramsize_j$.   
Since this holds for any setting of the first $j-1$ levels of $\ztree(x)$
and $\ztree(y)$ for which $\eventT(j-1)$ and $\eventC(j-1)$ hold, we conclude that:

\[
\prob[\eventT(j-1) \AND \eventC(j-1) \AND \neg \eventC(j) \AND \normex] \leq \frac{1}{\log^2 n},
\]
as required for (1).

For (2), if there is an index $j$ such that 
$\eventT(\leq j-1) \AND \neg \eventC(\leq j)$ and $i$ is the least such index then $\eventT(i-1) \AND \eventC(i-1) \AND \neg \eventC(i)$ holds.
Therefore:

\begin{multline}
\prob[\exists j \in \{1,\ldots,\tdepth\}, \eventT(\leq j-1) \AND \neg \eventC( \leq j) \AND \normex]  \\
\leq \sum_{i=1}^{j} \prob[\eventT(j-1) \AND \eventC(j-1) \AND \neg \eventC(j) \AND \normex] \leq \frac{j}{\log^2 n} \leq \frac{1}{\log n}
\end{multline}
where the second to last inequality uses the first part of this proposition.
\end{proof}

Let  $q \in \{0,\ldots,\tdepth\}$ be the largest index for which $\eventT(q)$ holds.  We emphasize that
$q$ is a random variable that depends on the execution.  The following result says that with
fairly high probability
$w(0),\ldots,w(q)$ are all compatible and if $q<\tdepth$ then with fairly high probability $w(q+1)$ is also 
compatible.

\begin{corollary}
\label{cor:U(j)}
\begin{enumerate}
\item $\prob[(q<\tdepth) \AND \text{$w(0),\ldots,w(q+1)$ are not all compatible } \AND \normex] \leq \frac{1}{\log n}$.
    \item $\prob[(q=\tdepth) \AND \text{$w(0),\ldots,w(\tdepth)$ are not all compatible }  \AND \normex] \leq \frac{1}{\log n}$.
    \end{enumerate}
\end{corollary}

\begin{proof}
For both parts, $\eventT(\le q)$ holds so the event whose probability is being bounded is included in the event
whose probability is bounded in the second part of Proposition~\ref{prop:U(j)}. 
\end{proof}

To prove  (\ref{eqn:mreconstruct}) we need an upper bound on:
\begin{multline}
\label{eqn:S not reconstructed}
    \prob[\text{$S$ is not correctly reconstructed} \AND \normex] \\
    = \prob[\text{$S$ is not correctly reconstructed} \AND \normex \AND (q=\tdepth)]\\
    + \prob[\text{$S$ is not correctly reconstructed} \AND \normex \AND (q<\tdepth)]
\end{multline}

The first term is easy to bound:

\begin{proposition}
    \label{prop:secondary case} 
    \[
    \prob[\text{$S$ is not correctly reconstructed} \AND \normex \AND (q=\tdepth)] \leq \frac{1}{\log n}
    \]
\end{proposition}

\begin{proof}
We write the probability as:
\begin{multline}
    \prob[\text{$S$ is not correctly reconstructed} \AND \normex \AND (q=\tdepth)] \\
    = \prob[\text{$S$ is not correctly reconstructed} \AND \normex \AND (q=\tdepth) \AND \eventC(\leq \tdepth)]\\
     + \prob[\text{$S$ is not correctly reconstructed} \AND \normex \AND (q=\tdepth) \AND \neg \eventC(\leq \tdepth)].
    \end{multline}

In the final expression, the second term is at most $\frac{1}{\log n}$
by Corollary~\ref{cor:U(j)}(2).
We claim that the first term is 0.  
To prove this we assume a normal execution where $q=\tdepth$ and $\eventC(\le \tdepth)$ both hold and show that $S$ is necessarily correctly reconstructed.  
Since $\eventC(\le \tdepth)$ holds, every node in $w(0), w(1),\ldots,w(\tdepth)$ is compatible. 
If one of them belongs to $\claimednodes$ then by Corollary~\ref{cor:S correct}(2), $S$ is correctly reconstructed.

Now assume that no node of $\{w(0),\ldots,w(\tdepth)\}$ belongs to $\claimednodes$.
By assumption, $\ED_{w(\tdepth)}<\gapthresh_{\tdepth}=1$, which implies that $\ztree_{w(\tdepth)}(x)$ and $\ztree_{w(\tdepth)}(y)$ are the same string, which 
implies that $\costly{\canon_{w(\tdepth)}}=\emptyset$, and the desired conclusion
follows from  Corollary~\ref{cor:S correct}(1).
\end{proof}

To bound the second term in  (\ref{eqn:S not reconstructed}), we have:

\begin{multline}
\prob[\text{$S$ not correctly reconstructed} \AND \normex \AND (q < \tdepth)] = \\
\prob[\text{$S$ is not correctly reconstructed} \AND \normex \AND (q < \tdepth) \AND \neg\eventC(\leq q+1)] + \\
\prob[\text{$S$ is not correctly reconstructed} \AND \normex \AND (q < \tdepth) \AND \eventC(\leq q+1)].
\label{S not correct}
\end{multline}

The first term is at most $\frac{1}{\log n}$ by the Corollary~\ref{cor:U(j)}(1).
To bound the second term define the events:

\begin{description}
\item[\STR$(v)$:] $\claimedstring_{v}(x)$ and $\claimedstring_{v}(y)$ are both defined.
\item[\LOC$(v)$:] $\claimedlocation_{v}(x)$ and $\claimedlocation_{v}(y)$ are both defined.
\end{description}

We will prove:

\begin{lemma}
\label{lemma:STR}
    $\prob[\neg\STR(w(q+1)) \AND \normex \AND (q<\tdepth) \AND \eventC(\leq q+1)] \leq \frac{1}{\log n}.$
\end{lemma}

\begin{lemma}
\label{lemma:LOC}
$\prob[\neg\LOC(w(q+1)) \AND \normex \AND (q<\tdepth) \AND \eventC(\leq q+1)] \leq \frac{1}{\log n}.$ 
\end{lemma}

$\STR(w(q+1))$ and $\LOC(w(q+1))$ together imply that $w(q+1)$ is correctly reconstructed. 
If in addition $\eventC(w(q+1))$ holds then by  Corollary~\ref{cor:S correct}(2), $S$ is correctly reconstructed.  Therefore
these two lemmas imply:

\begin{align*}
\prob & [\text{$S$ is not correctly reconstructed} \AND \normex \AND (q < \tdepth) \AND \eventC(\leq q+1)] \leq \\
&\prob[\neg\STR(w(q+1)) \AND \normex \AND (q<\tdepth) \AND \eventC(\leq q+1)] + \\
&\prob[\neg\LOC(w(q+1)) \AND \normex \AND (q<\tdepth) \AND \eventC(\leq q+1)] \leq \frac{2}{\log n}.
\end{align*}

This bounds the second term on the right of (\ref{S not correct}).  Combined with the
upper bound on the first term we have $\prob[\text{$S$ not correctly reconstructed} \AND \normex \AND (q < \tdepth)] \leq \frac{3}{\log n}$, and combining this 
with Proposition~\ref{prop:secondary case} yields
$\prob[\text{$S$ is not correctly reconstructed} \AND \normex]  \leq \frac{4}{\log n}$ as required
for Lemma~\ref{lemma:mreconstruct}.

To prove Lemmas~\ref{lemma:STR} and~\ref{lemma:LOC} we will need. 

\begin{proposition}
\label{prop:sufficient}
Assume a normal execution for strings $x$,$y$.
Let $v \in \BASE^{\leq \tdepth}$  with  $j=|v|$ such that $\ztree_v(x) \neq \ztree_v(y)$.  
\begin{enumerate}
\item 
The following two conditions are together sufficient for 
$\claimedstring_v(x)$ and $\claimedstring_v(y)$ to both be defined:
\begin{enumerate}
    \item
    \label{condition 1a}$\tstree_v(x) \neq \tstree_v(y)$
    \item 
    \label{condition 1b} For  the HMR-recovery problem associated to $(\bvec^j,T(\BASE^j\times \{1,\dots,\bvsize\}),\stringcap^{\leq j},4(\tdepth+1),1/n^4)$  (see the procedure
$\vectorcondense$), all of the nodes $v_{\leq 0},\ldots,v_{\leq j-1}$ are $\frac{1}{4(\tdepth+1)}$-underloaded (as defined in
Section~\ref{subsec:HMR})
\end{enumerate}
\item
\label{condition 2}
The following condition is sufficient  for $\claimedlocation_v(x)$ and $\claimedlocation_v(y)$ to both be defined: View $v$ as a node of $T(\{0,1\}^{\tdepth \gsize})$ where
its depth is $j \gsize$. For every $i\leq j\gsize$, for the HMR-recovery problem associated to $(\lvec^{i},T(\{0,1\}^{i}),\locationcap^{\leq i},4\tdepth\gsize,1/n^4)$  all of the nodes $v_{\leq 0},\ldots,v_{\leq i-1}$ are $\frac{1}{4\tdepth\gsize}$-underloaded. 
\end{enumerate} 
\end{proposition}

\begin{proof}
In both parts we apply Theorem~\ref{thm:HMR}. 
Since the execution is normal, there is no HMR-abnormality so all executions of
$\hmrrecover$ satisfy Completeness (see Section~\ref{subsec:HMR}), that every mismatch triple with accessible index is correctly recovered, and Soundness, that no incorrect triple is recovered.

When $\hmrrecover$ is applied to $\hmrrecover(\fsketch^j(x),\fsketch^j(y))$ it recovers all $\frac{1}{4\tdepth}$-accessible
mismatch triples correctly. 
We claim that all indices belonging to
$\{(v,i):i \in \{1,\ldots,\bvsize\}\}$ are $\frac{1}{4\tdepth}$-accessible and
for this we need that $v_{\leq 0},\ldots,v_{\leq j}=v$ are all $\frac{1}{4\tdepth}$-underloaded.
Condition \ref{condition 2}(b) says that $v_{\leq 0},\ldots,v_{\leq j-1}$ are all $\frac{1}{4\tdepth}$-underloaded. $v$ itself is also $\frac{1}{4\tdepth}$-underloaded, because 
the load on $v$ is at most the number of mismatches between $\bvec^{j}_v(x)$ and $\bvec^{j}_v(y)$
which is at most the total number of non-zero entries in both.  By Theorem~\ref{thm:BK},
$\bvec_v(x)$ and $\bvec_v(y)$ each have at most $\gramsize_j$ non-zero entries, so the load of $v$ is at most 
$2\gramsize_j$ which is at most $\frac{1}{4\log(n)}\stringcap_j$ (with room to spare). 

Therefore all indices in $\{(v,i):i \in \{1,\ldots,\bvsize\}\}$ are $\frac{1}{4\tdepth}$-accessible
and so
all mismatches between $\bvec^j_v(x)$ and $\bvec^j_v(y)$ are recovered correctly.  Thus
$\claimedbv_v(x)=\bvec_v(x)$ and $\claimedbv_v(y)=\bvec_v(y)$ and so
$\decode(\claimedbv_v(x))=\ztree_v(x)$ and $\decode(\claimedbv_v(y)=\ztree_v(y)$, completing the proof of the first part.

For the second part, we are now working on the binary tree $T(\{0,1\}^{\tdepth\gsize}$),
from the code for $\findlocations$, if $\claimedsize_{v_{\leq i}}(x)$
is defined for all $i <j$ then $\claimedlocation_v(x)$ will be defined.  $\claimedsize_{v_{\leq i}}(x)$
is defined from the run of $\hmrrecover$ on the HMR-sketch of $(\lvec^i,$ $T(\{0,1\}^i),$ $\locationcap^{\le i},$ $4\tdepth\gsize, 1/n^4)$. 
The mismatch triples of this run correspond to nodes $w$ at level $i$ for which $\lvec^i_w(x)\neq \lvec^i_w(y)$ or $\fptree_w(x) \neq \fptree_w(y)$.  
Since we are assuming a normal execution, $\fptree_w(x) \neq \fptree_w(y)$ if and only if $\ztree^*_w(x) \neq \ztree^*_w(y)$. 
Since $\ztree_v(x) \neq \ztree_v(y)$ by hypothesis, it follows that $\ztree_w(x) \neq \ztree_w(y)$ for all ancestors  $w$ of $v$. 
(If $\ztree_w(x)=\ztree_w(y)$ then $\bdecomp$ acts identically on both and so the $\ztree(x)$ and $\ztree(y)$ would be identical below $w$.)
Therefore for $i \leq j$, there is a mismatch triple corresponding to $v_{\leq i}$.  
Thus in the HMR-recovery problem associated to $(\lvec^i,T(\{0,1\}^i),\locationcap^{\le i}, 4\tdepth\gsize,1/n^4)$, 
if  every node on the path
from $v_{\leq i}$ to the root (including $v_{\leq i}$) is $\frac{1}{4\tdepth\gsize}$-underloaded
then the mismatch triple associated to $v_{\leq i}$ will be recovered. The hypothesis of
the second part ensures says that $v_{\leq 0},\ldots,v_{\leq i-1}$ are all $\frac{1}{4\tdepth\gsize}$-underloaded.  $v_{\leq i}$ is also because its load is just 1. Therefore
$\claimedsize_{v_{\leq i}}(x)$ and $\claimedsize_{v_{\leq i}}(y)$ will both be defined, as required to
conclude that $\claimedlocation_v(x)$ and $\claimedlocation_v(y)$ to both be defined.
\end{proof}

It remains to prove Lemmas ~\ref{lemma:STR} and ~\ref{lemma:LOC}.

\subsection{Proof of Lemma~\ref{lemma:STR}}
\label{subsec:STR}

For $i,j$ with  $0 \leq i \leq j \leq \tdepth$ define the following event:

\begin{description}
\item[$\eventF(i,j)$:] $w(i)$ is $\frac{1}{4(\tdepth+1)}$-overloaded 
    with respect to the HMR-recovery problem associated with \newline $(\bvec^j,T(\BASE^j\times \{1,\dots,\bvsize\}),\stringcap^{\leq j},4(\tdepth+1),1/n^4)$,
    which means $\hmrload{\stringcap}^{\leq j}_{w(i)} \geq \frac{1}{4(\tdepth+1)}\stringcap^{\leq j}_i$.
\end{description}

Note that by the first part of Proposition~\ref{prop:sufficient}:

\begin{multline}
\prob[\neg\STR(w(q+1)) \AND  \normex \AND (q<\tdepth) \AND \eventC(\leq q+1)]\\  
    \leq 
    \prob[(\tstree_{w(q+1)}(x) = \tstree_{w(q+1)}(y)) \AND (q<\tdepth) \AND \normex]\\ 
     + \prob[\OR_{0\leq i\leq q} (\eventF(i,q+1) \AND \normex \AND (q<\tdepth)\AND \eventC(\leq i)\AND T(\leq i))]\\
     =   \prob[\OR_{0\leq i\leq q} (\eventF(i,q+1) \AND \normex \AND (q<\tdepth)\AND \eventC(\leq i)\AND T(\leq i))]\\
     \leq  \prob[\OR_{0\leq i\leq q} (\eventF(i,q+1) \AND \normex \AND (q<\tdepth)\AND \eventC(\leq i)\AND T(\leq i))]\\
\leq  
\prob[\OR_{0\leq i<j \leq  \tdepth}( \eventF(i,j) \AND \normex \AND \eventC(\leq i)\AND T(\leq i))]
\label{eqn:sum of F}
\end{multline}

The equality holds  since
$\ED_{w(q+1)} \geq \gapthresh_{q+1}$ by definition of $q$ and so
a normal execution implies $\tstree_{w(q+1)}(x) \neq \tstree_{w(q+1)}(y)$.


Yo prove an upper bound on the right hand side of (~\ref{eqn:sum of F})
we will prove an upper bound on $\expected[\hmrload{\stringcap}^{\leq j}_{w(i)}]$ and  use
Markov's inequality.  To obtain such an upper bound, we introduce an auxiliary integer labeling
$\loadestimator^j$ on the nodes of $T(\BASE^{\leq j})$ which depends 
on the execution of $\fullsketch$ on $x$ and $y$ and is an upper bound on $\hmrload{\stringcap}^{\leq j}$. (Below,
$\dfactor$ is the parameter appearing in Theorem~\ref{thm:BK}.)

\[ \loadestimator^j_v = \begin{cases} 
          \min(\dfactor \ED_v , \stringcap_j)  & \text{if $|v|=j$ and $\tstree_v(x) =\tstree_v(y)$} \\
          2\gramsize_j &\text{if $|v|=j$ and $\tstree_v(x) \neq \tstree_v(y)$}\\
          \sum_{w \in  \child(v)} \loadestimator^j_w & \text{if $|v|<j$ and $v$ is compatibly split}\\
          \stringcap_{|v|} & \text{if $|v|<j$ and $v$ is not compatibly split}
       \end{cases}
    \]


    \begin{proposition}
    \label{prop:loadestimator}
    Suppose $j \in \{1,\ldots,\tdepth\}$.  For all nodes $v$ of $T(\BASE^{\leq j})$,
    $\hmrload{\stringcap}^{\leq j}_v \leq \loadestimator^j_v$.
    \end{proposition}
    
\begin{proof}
    Fix $j \in \{1,\ldots,\tdepth\}$ and let $v$ be a node of $T(\BASE^{\leq j})$. We proceed by reverse induction on $|v|$. 
    
    The base case is $|v|=j$. $\hmrload{\stringcap}^{\leq j}$ is defined on the nodes of $T(\BASE^{j} \times \{1,\ldots,\bvsize\})$.
    The leaves of this tree are pairs $(v,i) \in \BASE^j \times \{1,\ldots,\bvsize\}$ and by definition,
    $\hmrload{\stringcap}^{\leq j}_{v,i}$ is 1 if and only if $\bvec^j_{v,i}(x) \neq \bvec^j_{v,i}(y)$. 
    Thus for $v$ with $|v|=j$
     in $T(\BASE^{\leq j})$, by definition $\hmrload{\stringcap}^{\leq j}_v = \min(\stringcap_j,\sum_{i=1}^{\bvsize} \hmrload{\stringcap}^{\leq j}_{v,i})=\min(\stringcap_j,\HAM(\bvec^j_v(x),\bvec^j_v(y)))$.

     $\HAM(\bvec^j_v(x),\bvec^j_v(y))=\HAM(\bvec_v(x) \times \tstree_v(x),\bvec_v(y) \times \tstree_v(y))$.If $\tstree_v(x)=\tstree_v(y)$, by Theorem~\ref{thm:BK}, $\HAM(\bvec_v(x),\bvec_v(y)) \leq \dfactor \ED_v$, and so $\hmrload{\stringcap}^{\leq j}_v
     \leq \min(\dfactor \ED_v , \stringcap_j)$.
    If $\tstree_v(x)\neq \tstree_v(y)$, this is at most the total number of nonzero entries of $\bvec_v(x)$ and $\bvec_v(y)$ which is at most $2\gramsize_j$.

    Now suppose $v$ is a node of depth $i<j$.  By definition:
    \[
    \hmrload{\stringcap}^{\leq j}_v=\min(\stringcap_{|v|}, \sum_{w \in \child(v)}\hmrload{\stringcap}^{\leq j}_{w}).
    \]

    If $v$ is not compatibly split, then $\loadestimator^j_v=\stringcap_{|v|}  \geq \hmrload{\stringcap}^{\leq j}_v$, as required.

    If $v$ is compatibly split, then by the induction hypothesis

    \begin{eqnarray*}
        \loadestimator^j_v & = &\sum_{w \in \child(v)} \loadestimator^j_w\\
        & \geq & \sum_{w \in \child(v)} \hmrload{\stringcap}^{\leq j}_w \\
        & \geq & \hmrload{\stringcap}^{\leq j}_v.
    \end{eqnarray*}
    \end{proof}

Let $i \leq \tdepth$. 
For the analysis we will fix the outcome of the decomposition tree up-to level $i$,
and analyze the decomposition conditioned on the fixed part.
Let  $g^{\leq i}(x),g^{\leq i}(y)$ be possible outcomes for the random variables $\ztree^{\leq i}(x)$ and $\ztree^{\leq i}(y)$ 
and let $A_i=A(g^{\leq i}(x),g^{\leq i}(y))$ denote the event that $\ztree^{\leq i}(x)=g^{\leq i}(x)$ and $\ztree^{\leq i}(y)=g^{\leq i}(y)$. 
The event $A_i$ determines the value for $\ED_v$ for all $v$ at depth at most $i$, and we denote
this value by $\ED_v(A_i)$.  $A_i$ also determines whether $\eventC(\leq i)$ and $\eventT(\leq i)$ hold.
We now will show:

\begin{proposition}
\label{prop:loadestimator 2}
Let $i \leq j \leq \tdepth$.
Let $g^{\leq i}(x)$ and $g^{\leq i}(y)$ be a possible outcome for the first  $i$ levels of $\ztree(x)$ and $\ztree(y)$.
Let $v$ be a node in $T(\BASE^{\leq j})$ at depth $i$ such that under the event $A_i$, $v$ is compatible and
$\ED_v < \gapthresh_v$.
Then 
\begin{enumerate}
    \item $\expected[\loadestimator^j_v|A_i] \leq \critpar \ED_v(A_i) \times (j-i+2)$. 
    \item For any $\loadpar$, $\prob[\hmrload{\stringcap}^{\leq j}_v 
    \geq \frac{1}{\loadpar}\stringcap_i|A_i] \leq \frac{\loadpar}{4\log^5 n}$.
\end{enumerate}
\end{proposition}

Note that the event $\hmrload{\stringcap}^{\leq j}_v \geq \frac{1}{\loadpar} \stringcap^{\leq j}_i$ is precisely the event that $v$ is $\loadpar$-overloaded, and this will allow
us to prove an upper bound on the right hand side of (~\ref{eqn:sum of F}).


\begin{proof}
The main thing to prove is the first part.  The second part follows easily from the first part:

\begin{eqnarray*}
\prob[\hmrload{\stringcap}^{\leq j}_v 
    \geq \frac{1}{\loadpar}\stringcap_i|A_i] 
    & \leq & \frac{\loadpar\,\expected[ \hmrload{\stringcap}^{\leq j}_v|A_i]}{\stringcap_i}\\
    &\leq & \frac{\loadpar\,\expected[\loadestimator^j_v|A_i]}{\stringcap_i}\\
    & \leq & \frac{\loadpar\, \critpar \ED_v(A_i) (2\log n)}{\stringcap_i}\\
    & \leq & \frac{\loadpar\, \critpar \gapthresh_i (2 \log n)}{\stringcap_i}\\
    & \leq & \frac{\loadpar}{4\log^5 n}.
\end{eqnarray*}
Here the first inequality comes from Markov's inequality, the second comes from 
Proposition~\ref{prop:loadestimator}, the third comes from the first part of this proposition, the fourth comes from the hypothesis that $\ED_v(A_i) <\gapthresh_i$ and the final inequality follows from the definitions of $\gapthresh_i$ and $\stringcap_i$. 

We prove the first part by 
reverse induction on $i$.  The base case is $i=j$. 
In this case, using Theorem~\ref{thm:OR} we have:

\begin{eqnarray}
\expected[\loadestimator^j_v|A_j] & = &
   \hphantom{+} \; \prob[\tstree_v(x)=\tstree_v(y)|A_j]\cdot \dfactor \cdot \ED_v(A_j) \nonumber \\
& &          + \; \prob[\tstree_v(x)\neq \tstree_v(y)|A_j]\cdot 2\gramsize_j \nonumber \\
& \leq & \dfactor \ED_v(A_j) + \frac{\ED_v(A_j) \cdot \orgap}{\gapthresh_j} \cdot 2\gramsize_j \nonumber\\
& \leq & (\dfactor+\critpar)\cdot \ED_v(A_j) \leq 2\critpar \cdot \ED_v(A_j).\label{eqn:critpar}
\end{eqnarray}

Now suppose $i<j$ and condition on $A_i$. Suppose $v$ is a node at depth $i$.   

Let $\CS$ be the set of nodes that are compatibly split

\begin{eqnarray*}
    \expected[\loadestimator^j_v|A_i] & = & \prob[v \not\in \CS|A_i] \cdot \stringcap_i
    +\prob[v \in \CS|A_i] \cdot \expected[\loadestimator^j_v|A_i \AND v \in CS]\\
    &\leq & \prob[v \not\in \CS|A_i] \cdot \stringcap_i
    +\expected[\loadestimator^j_v|A_i \AND v \in \CS].
\end{eqnarray*}

By Theorem~\ref{thm:BK}: 
\begin{equation}
    \label{eqn:v not compatibly split}
    \prob[v \not\in \CS|A_i] \cdot \stringcap_i \leq
\bkfailure \frac{\ED_v(A_i)}{\gramsize_{i+1}} \stringcap_i \leq \frac{\critpar}{2} \ED_v(A_i).
\end{equation}

To bound the second summand we condition further on the $(i+1)$-st level of $\ztree(x)$ and $\ztree(y)$.
This extension determines whether $v$ is compatibly split, and we 
let $g^{\leq i+1}(x),g^{\leq i+1}(y)$ denote a possible extension of $g^{\leq i}(x),g^{\leq i}(y)$ to level $i+1$ such that $v$ is compatibly split. Let $A_{i+1}=A(g^{\leq i+1}(x),g^{\leq i+1}(y))$.

Let $\CS_v(A_{i+1})$ be the set of children of $v$ that are compatibly split
and $\NCS_v(A_{i+1})$ be the set of children of $v$ that are not compatibly split, conditioned on $A_{i+1}$. In what follows we write
$\#A$ for the cardinality of the set $A$ instead of $|A|$ to avoid confusion with the $|$ in conditional expectation.
\begin{eqnarray*}
\expected[\loadestimator^j_v|A_{i+1}] & = & \sum_{w \in \child(v)} 
\expected[\loadestimator^j_w|A_{i+1}] \\
&\leq & \sum_{w \in \child(v)} \prob[w \in \NCS_v|A_{i+1}] \stringcap_{i+1}\\
&&+\sum_{w \in \child(v)}\prob[w \in CS_v|A_i+1]\critpar \ED_w(A_{i+1}) \times (j-i+1)\\
& \leq & \stringcap_{i+1}\times \expected[\#\NCS(v)|A_{i+1}]\\
& & + \critpar \ED_v(A_i) \times (j-i+1)\\
\end{eqnarray*}

For each child $w$ of $v$, by Theorem~\ref{thm:BK},
$\prob[w \in \NCS(v)|A_{i+1}] \leq \bkfailure\frac{\ED_w(A_{i+1})}{\gramsize_{i+2}}$. 
Therefore $\expected[\#NCS(v)|A_{i+1}] \leq \sum_{w \in \child(v)} \frac{\bkfailure \ED_w(A_{i+1})}{\gramsize_{i+2}} = \bkfailure\frac{\ED_v(A_i)}{\gramsize_{i+2}}$ where the final equality uses
that $v$ is compatibly split.
Therefore:

\begin{eqnarray} 
\label{eqn:NCS(v)}
\stringcap_{i+1}\times \expected[\#\NCS(v)|A_{i+1}]& \leq & 
\bkfailure \frac{\stringcap_{i+1}}{\gramsize_{i+2}}\ED_v(A_i) \leq \frac{1}{2}\critpar \ED_v(A_i).
\end{eqnarray}

and so:

\begin{eqnarray*}
\expected[\loadestimator^j_v|A_{i+1}] 
& \leq & \critpar \ED_v(A_i) \times (j-i+\frac{3}{2})\\
\end{eqnarray*}

 This upper bound is the same for all extensions $A_{i+1}$ of $A_i$ for
which $v$ is compatibly split, so:

\begin{eqnarray*}
\expected[\loadestimator^j_v|A_{i} \AND (v \in CS)] 
& \leq & \critpar \ED_v(A_i) \times (j-i+\frac{3}{2})
\end{eqnarray*}

Combining with (\ref{eqn:v not compatibly split}) yields the desired conclusion:
\begin{eqnarray*}
\expected[\loadestimator^j_v] 
& \leq & \critpar \ED_v(A_i) \times (j-i+2)
\end{eqnarray*}
\end{proof}

Using this Proposition we obtain:

\begin{lemma}
    \label{lemma:F(i,j)} 
    \begin{enumerate}
    \item
        For all $i,j$ with $0 \leq i \leq  j \leq \tdepth$,  $\prob[\eventF(i,j) \AND \eventC(\leq i) \AND T(\leq i) ] \leq \frac{1}{\log^4 n}$. 
    \item $\prob[\OR_{0 \leq i < j \leq \tdepth}(\eventF(i,j) \AND \normex \AND \eventC(\leq i) \AND T(\leq i))] \leq \frac{1}{\log^2 n}$.
    \end{enumerate}
\end{lemma}

\begin{proof}
For (1), apply Proposition ~\ref{prop:loadestimator 2}(2) with $\loadpar=4(\tdepth+1) \leq 4\log n$
to get $\prob[\eventF(i,j) \AND \normex|A_i] \leq \frac{1}{\log^4 n}$, and then average over all choices of $A_i$ for which $\eventC(\leq i) \AND T(\leq i)$ hold.
The second part  follows immediately by applying a union bound to the first part,
\end{proof}

Part (2) of this lemma implies that the right hand side of (\ref{eqn:sum of F}) is
at most  $\frac{1}{\log n}$, as required to prove Lemma~\ref{lemma:STR}.

\subsection{Proof of Lemma~\ref{lemma:LOC}}
\label{subsec:LOC}

The proof of this lemma is closely related to that of
Lemma~\ref{lemma:STR} and we will reuse parts of that proof.  
One difference is that in $\findlocations$ the HMR-recovery operates on the tree $T(\{0,1\}^{\tdepth\gsize})$ rather than  $T(\BASE^{\tdepth})$.
Since the depth of this tree is $\tdepth\gsize$ the overload parameter for the HMR-recovery scheme
is set to $\frac{1}{4\tdepth\gsize}$ instead of $\frac{1}{4(\tdepth+1)}$.
Levels of the binary tree are denoted
by ordered pairs $(i,r)$ where $0 \leq i \leq \tdepth$ and $0 \leq r \leq \gsize-1$.
Level $(i,r)$ refers to level $i\gsize+r$ of $T(\{0,1\}^{\tdepth\gsize})$.  (If $i=\tdepth$ then $r$ can only be 0.)
Level $(i,0)$ in the binary tree corresponds to level $i$ in $T(\BASE^{\tdepth})$.
The order on levels $(i,r)$ is lexicographic,   We say that 
$v$ has depth $(i,r)$ if $|v|=i\gsize+r$, and has depth at most $(i,r)$ if
$|v| \leq i\gsize+r$.

The nodes in the path of $T(\{0,1\}^{\tdepth\gsize})$ that corresponds to the path from $w(0)$ to $w(\tdepth)$ in $T(\BASE^{\tdepth})$ are denoted by $w(i,r)$ where $(0,0) \leq (i,r) \le (\tdepth,0)$.

For $m=i\gsize+r$ we abbreviate the HMR-recovery problem
$(\lvec^m,T(\{0,1\}^m),\locationcap^{\leq m},4\tdepth\gsize,1/n^4)$ by $\locprob(i,r)$.
 We write $\locationcap^{\leq (i.r)}$ for
$\locationcap^{\leq m}$.  Recall that by definition, for a node $v$ at depth $(i',r') \le (i,r)$,
its capacity for $\locprob(i,r)$ is $\locationcap^{\leq (i,r)}_v=\stringcap_{i'}$, and its load is denoted $\hmrload{\locationcap}^{\le (i,r)}_{v}$.

The algorithm $\findlocations$, perform $\tdepth \gsize$ instances
of HMR-recovery 
corresponding to the labelings $\lvec^{j}$ for  $0\leq j \leq \gsize\tdepth-1$.  This is not
done for the final level $\gsize\tdepth$ because $\lvec^*_v$ is not undefined for leaves of
the binary tree.    
For analysis purposes it is convenient to augment the
algorithm by defining $\lvec_v=0$
for all $v \in \nonemptynodes^*$ and performing HMR sketch-and-recover for level $\gsize\tdepth$, and
including it as part of the sketch.  
We assume that the (trivial) recovery of this level defines $\claimedsize^*_v(x)=\claimedsize^*_v(y)=0$ for all nodes $v \in \{0,1\}^{\tdepth\gsize}$ where $\ztree^*_v(x) \neq \ztree^*_v(y)$.) 
This additional step plays no role in recovering the
canonical alignment of $x,y$, but makes it easier to reuse a part of the proof of 
Lemma ~\ref{lemma:STR}.  We must prove:

$\prob[\neg\LOC(w(q+1,0)) \AND \normex \AND (q<\tdepth) \AND \eventC(\leq q+1)] \leq \frac{1}{\log n}.$ 

Analogous to the events $\eventF(i,j)$ defined in the previous section,  
for $(i,r) \leq (j,t)$, we define the bad event:

\begin{description}
\item[$\eventL((i,r),(j,t))$] is the event that
$w(i,r)$ is $\frac{1}{4\tdepth\gsize}$-overloaded 
with respect to $\locprob(j,t)$, i.e.,
$\hmrload{\locationcap}^{\le (j,t)}_{w(i,r)} \geq \frac{1}{4\tdepth\gsize}\stringcap_i$.
\end{description}
We have:
\begin{eqnarray}
    \prob[& \neg\LOC&(w(q+1,0))  \AND  \normex \AND (q<\tdepth) \AND \eventC(\leq q+1)] \nonumber \\  
    & \leq & 
    \sum \prob[\OR_{(i,r) <  (j,t) < (q+1,0)}\eventL((i,r),(j,t)) \AND \normex \AND (q<\tdepth)\AND \eventC(\leq i) \AND \eventT(\leq i)] \nonumber \\
    & \leq & 
     \prob[\OR_{(i,r) <  (j,t) \leq  (\tdepth,0)}\eventL((i,r),(j,t)) \AND \normex \AND \eventC(\leq i)\AND \eventT(\leq i)] \label{eqn:L}
\end{eqnarray}

We will bound this last quantity by $\frac{1}{\log n}$.
The first step is to show that it suffices to bound the probability
of bad events of the form $\eventL((i,0),(j,0))$:

\begin{proposition}
\label{prop:restrict levels}
For $(i,r) < (j,t) \leq  (\tdepth,0)$, if $t \neq 0$, event $\eventL((i,r),(j,t))$ implies (is contained in)  event $ \eventL((i,0),(j+1,0))$, and  if $t = 0$ event $\eventL((i,r),(j,t))$ implies   $\eventL((i,0),(j,0))$, and therefore:

\begin{multline}
\prob[\OR_{(i,r) <  (j,t) \leq  (\tdepth,0)}\eventL((i,r),(j,t)) \AND \normex \AND \eventC(\leq i) \AND \eventT(\leq i)] \\
=
\prob[\OR_{i < j \leq \tdepth}\eventL((i,0),(j,0)) \AND \normex \AND \eventC(\leq i)\AND \eventT(\leq i)] 
\end{multline}

\end{proposition}

\begin{proof}
\begin{claim}
For any $(i,r) < (j,t)$, $\eventL((i,r),(j,t))$ implies $\eventL((i,0),(j,t))$.
\end{claim} 

\begin{proof}
From the recurrence relation for   $\hmrload{\locationcap}^{\le (j,t)}$, and
the fact that $\locationcap$ of a node is at most $\locationcap$ of its parent it follows 
that  $\hmrload{\locationcap}^{\le (j,t)}$ at a node is 
at most the $\hmrload{\locationcap}^{\le (j,t)}$ of its parent, and therefore
$\hmrload{\locationcap}^{\le (j,t)}$ does not decrease along paths to the root.  Therefore $\hmrload{\locationcap}^{\le (j,t)}_{w(i,r)} \leq \hmrload{\locationcap}^{\le (j,t)}_{w(i,0)}$. Since $\locationcap_{(i,0)}=\locationcap_{(i,r)}=\stringcap_i$,
$\eventL((i,r),(j,t))$ implies $\eventL((i,0),(j,t))$.
\end{proof}

\begin{claim}
If $(i,0) \leq (j,t)$ with $t \neq 0$ then   $\eventL((i,0),(j,t))$ implies $\eventL((i,0),(j+1,0))$.
\end{claim}

\begin{proof}
This follows from a more general claim: If $v$ is any node of the binary tree at level
$(i,r) \leq (j,t)$ and $(j,t)<(j',t')$ then $\hmrload{\locationcap}^{\le (j,t)}_v \leq \hmrload{\locationcap}^{\le (j',t')}_v$.
We prove this by reverse induction on $(i,r)$. 
For the basis, $(i,r)=(j,t)$, we have that $\hmrload{\locationcap}^{\le (j,t)}_v$ is 0 or 1,  and is 1 if and only if $\ztree^*_v(x) \neq \ztree^*_v(y)$.  So we need that if $\ztree^*_v(x) \neq \ztree^*_v(y)$
then $\hmrload{\locationcap}^{\le (j',t')}_v \geq 1$.  Since $\ztree^*_v(x) \neq \ztree^*_v(y)$
there is a descendant $w$ of $v$ at level $(j',t')$ such that $\ztree^*_w(x) \neq \ztree^*_w(y)$
and so $\hmrload{\locationcap}^{\le (j',t')}_w =1$.  Since $\hmrload{\locationcap}^{\le (j',t')}$ does not
decrease along paths to the root, $\hmrload{\locationcap}^{\le (j',t')}_v \geq 1$, as required 
for the basis step.  For the induction step, if $(i,r)<(j,t)$ then by induction we have:

\begin{eqnarray*}
 \hmrload{\locationcap}^{\le (j,t)}_v & = & \min(\locationcap^{\le (j,t)}_v,\sum_{w \text{ child of }v} \hmrload{\locationcap}^{\le (j,t)})\\
 & \le & \min(\locationcap^{\le (j',t')}_v,\sum_{w \text{ child of }v} \hmrload{\locationcap}^{\le 
 (j',t')}_v) = \hmrload{\locationcap}^{\le (j',t')}_v,
\end{eqnarray*}
\end{proof}
Combining the claims we have that for $t \neq 0$,
$\eventL((i,r),(j,t))$ implies $\eventL((i,0),(j,t))$ implies $\eventL((i,0),(j,0))$,
and if $t=0$ then the first claim gives
$\eventL((i,r),(j,0))$ implies $\eventL((i,0),(j,0))$.
\end{proof}

It is useful to recall the connection between the trees
$T(\BASE^{\tdepth})$ and $T(\{0,1\}^{\gsize d}$.  Since $\BASE=\{0,1\}^{\gsize}$, there
is a natural mapping from $\BASE^{\tdepth}$ to $\{0,1\}^{\gsize d}$, where
$v=v_1\ldots,v_j$ maps to the binary string $v_1 \concat \cdots \concat v_j$
of length $\gsize j$.  This gives a 1-1 correspondence from nodes at level $j$
in $T(\BASE^{\tdepth})$ to nodes at level $(j,0)$ in $T(\{0,1\}^{\gsize d}$.
For node $v$ in $T(\BASE^{\tdepth}$, let $v'$ be the associated node in $T(\{0,1\}^{\gsize d}$.
The HMR-recovery problem for $(\bvec^j,T(\BASE^j\times \{1,\dots,\bvsize\}),\stringcap^{\leq j},4\tdepth\gsize,1/n^4)$ analyzed in the previous subsection
operated on the tree $T(\BASE^j\times \{1,\dots,\bvsize\})$. 
\begin{proposition}
\label{prop:string > location}
For any $j \leq \tdepth$, and for any node $v$ of $T(\BASE^j)$, for any execution,
if the corresponding node $v'$ in $T(\{0,1\}^{\gsize \tdepth})$ is $\frac{1}{4\tdepth\gsize}$-overloaded for $\locprob(j,0)$ then $v$
is  $\frac{1}{4\tdepth\gsize}$-overloaded for the HMR-recovery problem for
$(\bvec^j,T(\BASE^j\times \{1,\dots,\bvsize\}),\stringcap^{\leq j},4\tdepth\gsize,1/n^4)$.
\end{proposition}

\begin{proof}
First we claim that for any non-leaf node $v$ of $T(\BASE^j\times \{1,\dots,\bvsize\})$, $\hmrload{\locationcap}^{\le (j,0)}_{v'} \leq \hmrload{\stringcap}^{\leq j}_v$.
We
prove this by reverse induction on $|v|$. For the basis, $|v|=j$,  
$\hmrload{\locationcap}^{\le (j,0)}_{v'} \in \{0,1\}$ and 
is 1 if and only if $\ztree_v(x) \neq \ztree_v(y)$.
If $\ztree_v(x) \neq \ztree_v(y)$ then $\bvec_v(x) \neq \bvec_v(y)$, which implies
that for some $a \in \{1,\ldots,\bvsize\}$ the child $(v,a)$ of $v$ satisfies  $\bvec_{v\circ a}(x) \neq\bvec_{v \circ a}(y)$ which implies that $\hmrload{\stringcap}^{\leq j}_v \geq 1$.

For the induction step, suppose $|v|=i<j$.   For a node $w$ at level less than $(i+1,0)$
in $T(\{0,1\}^{\tdepth \gsize}$,
let $D(u)$ be the set of its descendants at level $(i+1,0)$.

\begin{claim}
For  all $w$ that is a descendant of $v$ 
at level less than $(i+1,0)$ (including $v$ itself):
$\hmrload{\locationcap}^{\le (j,0)}_u=\min(\stringcap_i,\sum_{w \in D(u)} \hmrload{\locationcap}^{
\leq (j,0)}_w)$.
\end{claim}

This claim follows easily by (reverse) induction using the recurrence
$\hmrload{\locationcap}^{\le (j,0)}_v=\min(\locationcap_v,\sum_{w \text{ a child of }v} \hmrload{\locationcap}^{\leq (j,0)}_u)$ and the fact that  $\locationcap_{(i,r)}=\stringcap_i$ for
any $r \in [0,\gsize)$. 

Observe that $D(v')=\{w':w \text{ is a child of $v$ in } T(\BASE^d)\}$.
From the previous claim, we have: 
\begin{eqnarray*}
\hmrload{\locationcap}^{\leq (j,0)}_{v'}& = & \min(\stringcap_i,\sum_{w \text{ a child of $v$ in } T(\BASE^d)} \hmrload{\locationcap}^{\leq (j,0)}_{w'}\\
& \leq  & \min(\stringcap_i,\sum_{w \text{ a child of v}}\hmrload{\stringcap}^{\leq j}_w)
= \hmrload{\stringcap}^{\leq j}_v..
\end{eqnarray*}

Since $\locationcap^{\leq (j,0)}_{v'}=\stringcap^{\leq j}_v$, it follows that if
$v'$ is overloaded with respect to $\locprob(j,0)$ then $v$
is  $\frac{1}{4\tdepth\gsize}$-overloaded with respect to
$(\bvec^j,T(\BASE^j\times \{1,\dots,\bvsize\}),\stringcap^{\leq j},4\tdepth\gsize,1/n^4)$.
\end{proof}

We can now complete the proof of Lemma ~\ref{lemma:LOC}
Given these propositions we have:
\begin{multline}
\prob[\OR_{(i,r) <  (j,t) \leq (\tdepth,0)}\eventL((i,r),(j,t)) \AND \normex \AND \eventC(\leq i) \AND \eventT(\leq i)] \\
=
\prob[\OR_{i < j \leq \tdepth}\eventL((i,0),(j,0)) \AND \normex \AND \eventC(\leq i)\AND \eventT(\leq i)] \\
\leq \sum_{i<j \leq \tdepth} \prob[\eventL((i,0),(j,0)) \AND \normex \AND \eventC(\leq i)\AND \eventT(\leq i)]\\
= \sum_{i<j \leq \tdepth} \prob[\hmrload{\locationcap}^{\le (j,0)}_{w(i,0)} \leq \frac{1}{4\tdepth\gsize}\stringcap_i\AND \normex \AND \eventC(\leq i)\AND \eventT(\leq i)]\\ 
\le \sum_{i<j \leq \tdepth} \prob[\hmrload{\stringcap}^{\leq j}_{w(i)}\leq \frac{1}{4\tdepth\gsize}\stringcap_i\AND \normex \AND \eventC(\leq i)\AND \eventT(\leq i)].
\end{multline}

Using the second part of Proposition~\ref{prop:loadestimator 2}, each term in the final expression
is at most $\frac{\tdepth\gsize}{\log^5 n} \leq \frac{1}{\log^3 n}$, and summing over all $i,j$
gives an upper bound of $\frac{1}{\log n}$, as required to prove Lemma~\ref{lemma:LOC}.

\section{Constraints for the choice of our parameters}
\label{sec:constaints}

In Section~\ref{subsec:sketch code} we defined various parameters for the scheme.
The parameters $\gapthresh_i$, $\gramsize_i$ and $\stringcap_i$ are adjusted to satisfy various constraints that arise in the analysis.
As stated in Theorem~\ref{thm:main}, the length of the sketch is $\tilde{O}(\stringcap_0)$ so we choose the parameters to satisfy the constraints while (approximately) minimizing $\stringcap_0$. 
For convenient reference, here we collect the constraints that these parameters must satisfy.  

\begin{itemize}
    \item $\gapthresh_0=\lceil 20k\orgap \rceil_2$.  This ensures that condition $\eventC(0)$ holds as noted prior
    to Proposition~\ref{prop:U(j)},
    \item $\gapthresh_{\tdepth}=1$ which is used inside the proof of Proposition~\ref{prop:secondary case}.
\end{itemize}
For $j \in \{1,\ldots,\tdepth\}$:
\begin{enumerate}

    \item 
    \label{constraint 1} $\gramsize_j \geq \stringcap_{j-1} \frac{2 \bkfailure}{\critpar}$. This is used both in (\ref{eqn:v not compatibly split}) and (\ref{eqn:NCS(v)}) to bound the contribution of nodes that are not compatibly split to the expectation of $\hmrload{\stringcap}^{\leq j}$. 
    \item
    \label{constraint 2} $\gapthresh_j \geq \frac{2\orgap}{\critpar}\gramsize_j$. This is used in (\ref{eqn:critpar}) to bound the contribution to
    $\hmrload{\stringcap}^{\leq j}$ from leaves for which $\tstree_v(x) \neq \tstree_v(y)$.
    \item
    \label{constraint 3}$\gramsize_{j} \geq (\bkfailure \log^2n) \cdot \gapthresh_{j-1}$.  This is used in (\ref{eqn:NCS}) to bound the probability that a node $v$ for which $\tstree_v(x) = \tstree_v(y)$ is not compatibly split by $\frac{1}{\log^2 n}$.  This constraint does not need to be explicitly enforced because it follows from the combination of constraints \ref{constraint 1}  and \ref{constraint 4}
    \item
    \label{constraint 4}$\stringcap_j \geq (8 \critpar\log^6n ) \cdot \gapthresh_j$. This allows us to use Markov's inequality
    to deduce the second part of Proposition~\ref{prop:loadestimator 2} to bound the probability that a node
    is overloaded with respect to $\hmrload{\stringcap}^{\leq j}$.
    \label{constraint 5}
    $\stringcap_j \geq 4\gramsize_j$. When doing $\hmrrecover$ on the sketches $\fsketch^j(x),\fsketch^j(y)$, every node $v$ at level $j$ is underloaded.
    \item $\stringcap_j \geq (8\log n)\cdot \gramsize_j$. This is used in the proof of Proposition~\ref{prop:sufficient}. This constraint does not need to be explicitly enforced because it follows from the combinations of constraints \ref{constraint 3} and
    \ref{constraint 4}.
    \item 
    \label{constraint 6}
    $\critpar \geq \dfactor$.  This is used in the basis step of the proof of Proposition~\ref{prop:loadestimator 2}.  It is implied by the constraint obtained by multiplying together \ref{constraint 1}, \ref{constraint 2}, and \ref{constraint 4}.
\end{enumerate}

\bibliographystyle{alpha}
\bibliography{edit_sketch}

\end{document}